\documentclass[11pt]{article}

\usepackage[a4paper,margin=1in]{geometry}
\usepackage{amsmath,amssymb,amsthm,amsfonts}
\usepackage{mathtools}
\usepackage{bm}
\usepackage[round,authoryear]{natbib}
\usepackage{booktabs}
\usepackage{array}
\usepackage{enumitem}
\usepackage[protrusion=true,expansion=false]{microtype}
\usepackage{xcolor}
\usepackage{hyperref}
\hypersetup{colorlinks=true,linkcolor=blue!50!black,citecolor=blue!50!black,urlcolor=blue!50!black,hypertexnames=false}
\usepackage[capitalise,noabbrev]{cleveref}
\usepackage{graphicx}

\usepackage{float}

\newtheorem{theorem}{Theorem}[section]
\newtheorem{proposition}[theorem]{Proposition}

\newtheorem{corollary}[theorem]{Corollary}
\theoremstyle{definition}
\newtheorem{definition}[theorem]{Definition}
\newtheorem{assumption}[theorem]{Assumption}
\newtheorem{example}[theorem]{Example}
\theoremstyle{remark}
\newtheorem{remark}[theorem]{Remark}

\DeclareMathOperator*{\argmax}{arg\,max}

\DeclareMathOperator{\Lip}{Lip}
\DeclareMathOperator{\Var}{Var}
\DeclareMathOperator{\Cov}{Cov}
\DeclareMathOperator{\KL}{KL}
\DeclareMathOperator{\VaR}{VaR}
\DeclareMathOperator{\CVaR}{CVaR}
\DeclareMathOperator{\EVaR}{EVaR}
\DeclareMathOperator{\esssup}{ess\,sup}
\DeclareMathOperator{\essinf}{ess\,inf}
\DeclareMathOperator{\tr}{tr}
\DeclareMathOperator{\supp}{supp}
\DeclareMathOperator{\Dir}{Dir}
\DeclareMathOperator{\NIW}{NIW}

\DeclareMathOperator{\Law}{Law}
\DeclareMathOperator{\softmax}{softmax}
\DeclareMathOperator{\softplus}{softplus}

\newcommand{\R}{\mathbb{R}}
\newcommand{\E}{\mathbb{E}}
\newcommand{\PP}{\mathcal{P}}
\newcommand{\W}{\mathcal{W}}
\newcommand{\F}{\mathcal{F}}
\newcommand{\B}{\mathcal{B}}

\newcommand{\StateSpace}{\mathcal{S}}
\newcommand{\A}{\mathcal{A}}
\newcommand{\M}{\mathcal{M}}
\newcommand{\Z}{\mathcal{Z}}
\newcommand{\Q}{\mathcal{Q}}
\newcommand{\N}{\mathcal{N}}
\newcommand{\IP}{\mathrm{IP}}
\newcommand{\Tpi}{\mathcal{T}^\pi}
\newcommand{\Wbar}{\overline{W}}
\newcommand{\Vex}{V^{\mathrm{ex}}}
\newcommand{\Vcond}{V^{\mathrm{cond}}}
\newcommand{\Rpred}{R^{\mathrm{pred}}}
\newcommand{\Rpredn}{R^{\mathrm{pred}}_n}
\newcommand{\Trans}{\mathsf{P}} % transition kernel (distinct from P)

\title{\bfseries Distributional Portfolio Optimization (DPO):\\
A Unified Framework for Distributions over\\
Weights, Returns, and Parameters}
\author{Miquel Noguer i Alonso\\
\\ Artificial Intelligence Finance Institute (AIFI)}
\date{\today}

\begin{document}
\maketitle

%==============================================================================
% ABSTRACT (PATCH 1)
%==============================================================================
\begin{abstract}
\sloppy
\noindent
Classical portfolio optimization treats expected returns, covariances, and
allocations as deterministic. Modern practice replaces at least one by a
distribution---a posterior over parameters, a law of future returns, a
stochastic allocation policy, or a distributional-robustness set. We call
\emph{distributional portfolio optimization} (DPO) the unified framework in
which weights, returns, and parameters are all modelled as probability measures,
organized around the joint coupling $\Gamma_\theta(dw,dr)$ of allocations and
returns and its marginal triple $(W,R,P)$.

The contribution is synthetic and structural: we organize Bayesian, robust,
chance-constrained, stochastic-allocation, and distributional-RL portfolio
methods through this coupling and prove boundary results connecting them. These
include a portfolio specialization of the Wasserstein--CVaR duality of
\citet{mohajerin2018data}, reducing the linear-loss case to empirical CVaR plus
the dual-norm penalty $(\varepsilon/\alpha)\|w\|_{\ast}$; a static
no-randomization theorem; a Bayesian credible-radius calibration of Wasserstein
DRO (a posterior calibration of the ambiguity radius, distinct from the
distribution-free coverage radius); a Gaussian-isotropic second-order
conservatism bound with first-order tightness in nonsmooth models; a conditional
two-sided rate
$W_1(R^{\mathrm{pred}}_n,R_{\widehat\theta_n})=\Theta(n^{-(1+\alpha^\star)/2})$
governed by the local boundary H\"older exponent $\alpha^\star\in[0,1]$; and a
risk-shifted distributional Bellman contraction at rate $\gamma+\lambda L_\rho$.

Numerical experiments are illustrative. A controlled calibration experiment
shows that across factor DGPs at $K\in\{10,25,50\}$ the credible-radius rule
lands within $3$--$7$ bp of the oracle out-of-sample tail risk and strictly beats
a $24$-month validation-tuned radius while spending no validation data. On a
$K=25$ DJIA backtest, equal-weight $1/N$, no-view Black--Litterman, and
Ledoit--Wolf shrinkage attain higher Sharpe than every distributional method
both gross and net of cost; the operational claim is therefore confined to
calibration-without-validation and turnover, not raw-return dominance over
shrinkage or naive baselines.

\end{abstract}

\newpage

\tableofcontents
\bigskip

%==============================================================================
% SECTION 1: INTRODUCTION
%==============================================================================
\section{Introduction}\label{sec:intro}

The classical Markowitz problem \citep{markowitz1952portfolio}
\begin{equation}
\max_{w\in\W} w^\top \mu - \frac{\gamma}{2} w^\top \Sigma w
\label{eq:markowitz}
\end{equation}
takes three inputs---the expected return $\mu$, the covariance $\Sigma$, and
the feasible set $\W$---and produces a single allocation $w^\star$. Every
quantity on both sides is treated as deterministic. The justification for
using \eqref{eq:markowitz} as a proxy for expected-utility maximization
extends beyond Gaussian returns: as formalized by
\citet{benveniste2024untangling}, the maximum-expected-utility allocation
coincides with the mean--variance allocation for any standard concave utility
whenever the return distribution lies in the mean--variance-equivalent (MVE)
class, encompassing all elliptical distributions and a broad family of
asymmetric ones. The DPO framework developed below preserves this universality
at its base case ($W, P$ degenerate) and extends to settings where one or
more of $W, R, P$ are non-trivial.

In realistic practice, none of them are:
\begin{itemize}
\item the parameters $(\mu, \Sigma)$ are estimated from data, hence random; a
      Bayesian approach replaces them by a posterior distribution;
\item the return vector itself is random, and a full description requires more
      than first two moments---higher-order moments, tail behavior, and
      asymmetric dependence all contribute to utility and risk functionals;
\item the weights $w$ may themselves be modeled stochastically---as a
      sampling distribution on $\W$ for resampled efficiency, or as a
      stochastic policy for a reinforcement-learning agent.
\end{itemize}

Each of these generalizations has been developed independently in a distinct
literature: Bayesian portfolio theory \citep{jorion1986bayes,
black1992global, avramov2010bayesian}, coherent risk measures
\citep{artzner1999coherent, rockafellar2000optimization}, distributionally
robust optimization \citep{delage2010distributionally, mohajerin2018data},
chance-constrained programming \citep{charnes1959chance, prekopa1995stochastic},
random portfolios \citep{michaud1998efficient, burns2007random}, hierarchical
risk parity and its return-adjusted variant
\citep{lopez2016building, noguer2025rahrp}, and distributional reinforcement
learning \citep{bellemare2017distributional, dabney2018distributional}. These
subliteratures rarely communicate with one another despite a shared
structure: in each, the object being optimized over is a probability measure
rather than a vector.

The present paper collects these threads under a single analytical
framework---distributional portfolio optimization. The aims are threefold:
(i) make precise the shared mathematical structure---the distributional
triple $(W, R, P)$; (ii) identify duality and domination relations connecting
methods originally posed as alternatives; and (iii) derive finite-sample
guarantees, contraction arguments, and posterior-concentration results that
make the framework amenable to rigorous asymptotic analysis. The framework
is complementary to the reinforcement-learning portfolio optimization (RLPO)
theory described in \S\ref{sec:rlpo-relationship}, with which it shares
notation; where RLPO emphasizes the dynamic Markov decision process underlying
portfolio rebalancing, DPO emphasizes the single-period or averaged
distributional structure of the decision problem.

\paragraph{Positioning relative to existing work.}
The contribution is primarily synthetic and structural. Classical
mean--variance and Bayesian portfolio theory provide the deterministic and
posterior-predictive faces of the framework. Coherent and spectral risk
theory provide the law-invariant scalar risk functionals. Wasserstein DRO
supplies ambiguity sets and dual regularization formulas. Stochastic
allocation methods, including resampling and Dirichlet policies, make the
allocation law nontrivial. Distributional RL extends the same scalar-outcome
logic to dynamic return distributions. The present paper does not claim to
subsume these literatures; rather, it provides a common notation in which
their assumptions, equivalences, and non-equivalences can be compared
directly. The technical contributions sit at the boundaries between these
subliteratures---Bayesian-to-Wasserstein bridges, no-randomization theorems,
distributional Bellman contraction with action-coupled rewards, and first-
versus second-order tightness examples---rather than within any single
subliterature.

%------------------------------------------------------------------------------
\subsection{The distributional triple}\label{sec:triple}

The defining observation of DPO is that every portfolio problem has three
orthogonal sources of randomness, collected into a distributional triple:
\[
(W, R, P) : W \in \PP(\W), \; R \in \PP(\R^K), \; P \in \PP(\Theta),
\]
where $W$ is a distribution over weights, $R$ a distribution over returns,
and $P$ a distribution over the parameters governing $R$. Each point in the
triple encodes a different problem class:
\begin{itemize}
\item $(W = \delta_w, R, P = \delta_\theta)$: classical Markowitz with known
      parameters.
\item $(W = \delta_w, R, P)$: Bayesian portfolio optimization, with
      predictive return $R = \int R_\theta \, dP(\theta)$.
\item $(W = \delta_w, R \in \A, P)$: distributionally robust portfolio
      optimization.
\item $(W \neq \delta, R, P)$: random portfolios, RL policies, Dirichlet
      weight priors.
\item $(W, R, P)$ all non-trivial: hierarchical Bayesian models with
      stochastic policies, the full generalization.
\end{itemize}

The triple is not merely bookkeeping. It organizes the literature
(Section~\ref{sec:taxonomy}), isolates which duality theorems apply to which
setting, and reveals geometric structure through the Fisher--Rao and
Wasserstein metrics on each marginal.

The structure of the framework is:
\begin{equation}
\underbrace{(W, R, P)}_{\text{inputs}}
\xrightarrow{\IP_\#}
\underbrace{Y^\star_{W,R,P} \in \PP(\R)}_{\text{induced outcome}}
\xrightarrow{u, \rho}
\underbrace{V \in \R}_{\text{objective}}.
\label{eq:dpo-pipeline}
\end{equation}
The induced outcome distribution is the scalar object on which all utility
and risk functionals act, and the pipeline \eqref{eq:dpo-pipeline} organizes
the rest of the paper.

%------------------------------------------------------------------------------
\subsection{Contributions and organization}\label{sec:contributions}

The main contributions are:
\begin{itemize}
\item a measure-theoretic formulation of DPO on Polish spaces
      (Section~\ref{sec:prelim});
\item a unified treatment through the distributional triple, including the
      explicit identification of the induced portfolio-outcome distribution
      as the scalar object on which all utility and risk functionals act
      (Section~\ref{sec:taxonomy});
\item rigorous existence, limit, and concentration results for Bayesian DPO
      via utility-tilted posteriors (Section~\ref{sec:bayesian-dpo});
\item for completeness, a self-contained restatement of the law-invariant
      risk and duality toolkit we build on---Rockafellar--Uryasev, Kusuoka,
      spectral, and distortion representations (Section~\ref{sec:returns}),
      the \citet[Cor.~5.1]{mohajerin2018data} Wasserstein duality together
      with its explicit portfolio-CVaR conic specialization
      (Theorem~\ref{thm:wcvar-dro}, Proposition~\ref{prop:explicit-cvar-program})
      and finite-sample guarantees (Section~\ref{sec:dro}), and
      Charnes--Cooper / Calafiore--Campi chance constraints with an explicit
      $e/(e-1)$-form scenario bound (Section~\ref{sec:cc}); none of the
      underlying duality theorems is claimed as new, and the only local
      increment in these sections is the portfolio-CVaR specialization;
\item a self-contained treatment of Dirichlet and logistic-normal policy
      classes, with score, entropy, expected-utility, and policy-gradient
      identities, and a static no-randomization theorem
      (Theorem~\ref{thm:no-randomization}) delimiting when stochastic
      allocations are useful (Section~\ref{sec:weights});
\item a complete proof of contraction for the distributional Bellman operator in the maximal
$p$-Wasserstein metric, together with a risk-shifted distributional Bellman contraction
(Theorem~\ref{thm:risk-sensitive-bellman-contraction}) at rate $\gamma+\lambda L_\rho$, where $L_\rho$ is the
$W_1$-Lipschitz modulus of the one-step Markov risk functional and $M_\rho<\infty$ is its bounded baseline (Section~\ref{sec:distributional-rl});
\item a domination theorem (Theorem~\ref{thm:bayes-dro-domination}) bounding Bayesian risk-averse utility maximization
by Wasserstein DRO with an explicit data-driven radius, a Gaussian-isotropic calculation
(Proposition~\ref{prop:gaussian-isotropic-second-order}) showing that the first-order radius may be conservative by one order
in symmetric location models, a formal second-order Bayes--DRO radius
(Proposition~\ref{prop:second-order-radius}) for the centered symmetric case, a first-order tightness proposition
(Proposition~\ref{prop:first-order-tight}) for affine-deterministic and locally non-smooth models, a sufficient
modulus-of-smoothness bound (Theorem~\ref{thm:modulus-interpolation}), and a conditional two-sided
Bayes--Wasserstein rate theorem (Theorem~\ref{thm:sharp-two-sided-rate}) showing that, under posterior-moment,
tail-negligibility, and explicit lower-witness hypotheses,
\[
W_1(R^{\mathrm{pred}}_n,R_{\widehat\theta_n})
=
\Theta\!\left(n^{-(1+\alpha_n^\star)/2}\right),
\]
where $\alpha_n^\star=\inf_{\theta\in B(\widehat\theta_n,r_n)}\alpha^\star(\theta)\in[0,1]$
is the conservative local boundary H\"older exponent around the realized posterior
center (Section~\ref{sec:synthesis});
\item the Bures--Wasserstein Lipschitz constant for the Gaussian--NIW model
      promoted to a Proposition with explicit constants
      (Proposition~\ref{prop:gaussian-niw-lipschitz},
      Section~\ref{sec:synthesis});
\item numerical and empirical illustrations consistent with the theoretical
      predictions, including a Ledoit--Wolf-shrinkage baseline that frames the
      high-dimensional regime as an empirical question rather than an
      established dominance claim (Section~\ref{sec:numerical}).
\end{itemize}

The paper's novel mathematical content is concentrated in
Section~\ref{sec:synthesis}: the Bayes--Wasserstein domination
(Theorem~\ref{thm:bayes-dro-domination}), the second-order conservatism
calculation (Proposition~\ref{prop:gaussian-isotropic-second-order}), and the
conditional two-sided rate theorem (Theorem~\ref{thm:sharp-two-sided-rate}).
Sections~\ref{sec:returns}--\ref{sec:cc} are background we unify;
Sections~\ref{sec:weights}--\ref{sec:distributional-rl} are largely
recapitulation, with the static no-randomization theorem
(Theorem~\ref{thm:no-randomization}) and the risk-shifted contraction
(Theorem~\ref{thm:risk-sensitive-bellman-contraction}) as the local
increments.

A companion paper \citep{noguer2026rlpo} develops the dynamic-control side of
this picture, casting portfolio choice as risk-sensitive policy optimization.
The dynamic-control analog (RLPO) is described in
\S\ref{sec:rlpo-relationship} below.

\paragraph{What is new here.}
The paper does not claim that Bayesian portfolio choice, coherent risk, Wasserstein DRO,
stochastic policies, or distributional reinforcement learning are individually new. Its
contribution is to formulate these objects as faces of a single coupled distributional
portfolio problem, to distinguish the marginal uncertainty triple $(W,R,P)$ from the
decision-relevant joint coupling $\Gamma_\theta$, and to prove a collection of boundary
results that make the connections mathematically explicit: outcome-law sufficiency,
static no-randomization, Wasserstein--CVaR regularization, Bayesian credible-radius
calibration of DRO balls, a conditional two-sided Bayes--Wasserstein rate theorem
parameterized by the boundary H\"older exponent $\alpha^\star\in[0,1]$, a
distributional Bellman contraction in a portfolio MDP, and a risk-shifted distributional
Bellman contraction at rate $\gamma+\lambda L_\rho$. The CVaR conic program is an
explicit specialization of Mohajerin Esfahani and Kuhn (2018, Cor. 5.1), not a new
duality theorem; the contribution of this part is the explicit portfolio specialization and
computational identification: for the canonical linear-loss CVaR problem, Wasserstein
robustness is equivalent to empirical CVaR plus the dual-norm penalty
$(\varepsilon/\alpha)\|w\|_\ast$.

%------------------------------------------------------------------------------
% PATCH 2: SECTION 1.3 RELATIONSHIP TO RLPO (full rewrite)
%------------------------------------------------------------------------------
\subsection{Relationship to RLPO}\label{sec:rlpo-relationship}

The reinforcement-learning portfolio optimization (RLPO) framework formulates
portfolio rebalancing as a Markov decision process
$(\StateSpace, \A, \Trans, R, \gamma)$ with policy $\pi(\cdot \mid s)$ and value
function
\[
V^\pi(s) = \E^\pi\!\left[\sum_{t \geq 0} \gamma^t R_t \,\middle|\, s_0 = s\right].
\]
DPO and RLPO should be viewed as parallel cross-sections of a common
dynamic--distributional theory rather than as literal reductions of one
another. Two observations make the relationship precise.

\emph{First,} the static DPO objective carries a law-invariant scalar risk
functional $\rho$ that has no direct $\gamma = 0$ analog in RLPO. Setting
$\gamma = 0$ in RLPO reduces the value function to single-period expected-reward
maximization $\argmax_w \E[R(s, w, s') \mid s]$, which is risk-neutral;
recovering the static DPO with a static $\rho$ from RLPO requires encoding
$\rho$ as a one-period Markov risk measure in the sense of
\citet{ruszczynski2010risk} and stopping the recursion at $T = 1$. The two
formulations are therefore dual: the static-decision face emphasizes the
law-invariant scalar functional, the dynamic-control face emphasizes the
recursion. The correct relationship is
\[
\text{one-period RLPO} + \text{static law-invariant risk/utility functional}
\;\leadsto\;
\text{static DPO}.
\]

\emph{Second,} under the parallel identification of state and information
set, the DPO triple specializes the dynamic objects:
\begin{itemize}
\item the RLPO state $s$ corresponds to the information set at the decision
      time, and the DPO parameter law $P$ to the posterior over $(\mu, \Sigma)$
      given $s$;
\item the RLPO policy $\pi(\cdot \mid s)$ corresponds to the DPO weight
      distribution $W$ conditional on the observed state, collapsing to
      $\delta_{w^\star(s)}$ for a deterministic policy;
\item the RLPO one-period reward $R(s, w, s') = w^\top r' - c(w, w_+)$ is a
      draw from $R \in \PP(\R^K)$.
\end{itemize}
Conversely, DPO extends to the dynamic setting by letting all three of
$(W, R, P)$ carry a time index. The distributional Bellman theory of
Section~\ref{sec:distributional-rl} is the explicit dynamic extension, and
Theorem~\ref{thm:distributional-bellman} provides the $\gamma$-contraction
of the return-distribution operator.

\begin{table}[H]
\centering\small
\begin{tabular}{lll}
\toprule
Object & Static (DPO) & Dynamic (RLPO) \\ \midrule
Weights          & $W \in \PP(\W)$                  & $\{W_t(\cdot \mid s_t)\}_{t \geq 0}$, policy kernel \\
Returns          & $R \in \PP(\R^K)$                & $\{R(s_t, w_t, s_{t+1})\}_{t \geq 0}$, transition kernel \\
Parameters       & $P \in \PP(\Theta)$              & $\{P_t\}_{t \geq 0}$, posterior at each step \\
Risk functional  & $\rho$ on $\PP(\R)$              & nested Markov risk measure $\{\rho_t\}$ \\
Outcome          & $Y^\star = \IP_\#(W \otimes R)$  & $Z^\pi(s,w) = \Law_\Trans(\sum_t \gamma^t r_t \mid s_0, w_0)$ \\
Recursion        & N/A (single period)              & $Z^\pi = \Tpi Z^\pi$, contraction in $\Wbar_p$ (Thm.~\ref{thm:distributional-bellman}) \\
$\gamma=0$ limit & \begin{tabular}[t]{@{}l@{}}static decision problem;\\$\rho$, $P$ added on top\end{tabular}
                 & \begin{tabular}[t]{@{}l@{}}one-period expected reward,\\no risk augmentation\end{tabular} \\
\bottomrule
\end{tabular}
\caption{Static distributional triple (DPO) versus its time-indexed dynamic
counterpart (RLPO). Neither framework is a literal reduction of the other:
DPO carries a static law-invariant $\rho$ that requires nesting to embed in
RLPO, and RLPO carries a recursion that is absent from the static formulation.
The two are parallel cross-sections of a common dynamic--distributional
theory.}
\label{tab:dpo-rlpo}
\end{table}

%==============================================================================
% SECTION 2: MEASURE-THEORETIC PRELIMINARIES
%==============================================================================
\section{Measure-theoretic preliminaries}\label{sec:prelim}

We fix notation and collect facts used throughout. Standard references are
\citet{billingsley2012}, \citet{villani2009optimal}, and \citet{bogachev2007}.

\subsection{State spaces and probability measures}

\paragraph{Notation.}
Throughout, $\W \subset \R^K$ denotes the feasible set of portfolio weights
and $w \in \W$ a realized portfolio vector. We use $W \in \PP(\W)$ for the
law of a stochastic allocation; in dynamic policy contexts the same object
is also called the policy and denoted $\pi$. The DPO triple is
\[
(W, R, P) \in \PP(\W) \times \PP(\R^K) \times \PP(\Theta),
\]
with $W$ the weight distribution, $R$ the return distribution (typically
$R_\theta$ the conditional law given $\theta$), and $P$ the parameter
distribution. We reserve the calligraphic $\W$ exclusively for the feasible
set so that $\W$ and $W$ never refer to the same object, and we write
$w \sim W$ for a draw from the weight law. The Polish probability-space
machinery applies on $\W$, $\R^K$, and $\Theta$ symmetrically.

Let $(\Omega, \F, \Pr)$ be a complete probability space. The asset universe
has $K < \infty$ assets, returns $r \in \R^K$, weights $w \in \R^K$. The
feasible set $\W \subseteq \R^K$ is non-empty convex compact; long-only
problems take $\W = \Delta^{K-1}$. For a Polish space $(X, d_X)$ with Borel
$\sigma$-algebra $\B(X)$, $\PP(X)$ denotes Borel probability measures,
equipped with weak convergence. The subset $\PP_p(X)$ of measures with
finite $p$th moment is equipped with the $p$-Wasserstein metric
\begin{equation}
W_p(\mu, \nu) = \inf_{\pi \in \Pi(\mu, \nu)}
\left( \int_{X \times X} d_X(x, y)^p \, d\pi(x, y) \right)^{1/p}, \quad p \geq 1.
\label{eq:wasserstein-def}
\end{equation}
$(\PP_p(X), W_p)$ is itself Polish.

\subsection{Parameter space}

The parameters $\theta \in \Theta$ live in a finite-dimensional Euclidean
open set $\Theta \subseteq \R^{d_\theta}$ (or, in semi-parametric settings,
a Banach space with separable predual). The map $\theta \mapsto R_\theta$
is Borel measurable; continuity and Lipschitz regularity in Wasserstein
distance are imposed as needed at each use.

\subsection{Utility and risk functionals}

A utility function $u: \R \to \R$ is concave, non-decreasing, continuous.
Expected utility under $R \in \PP(\R^K)$ is
\begin{equation}
U(w; R) = \int_{\R^K} u(r^\top w) \, dR(r),
\label{eq:utility}
\end{equation}
well-defined whenever $u^-(r^\top w) \in L^1(R)$.

\begin{definition}[Coherent risk measure on losses; \citealp{artzner1999coherent}]
\label{def:coherent}
A functional $\rho: L^\infty(\Omega, \F, \Pr) \to \R$ is a coherent risk
measure on losses if for all $X, Y \in L^\infty$:
\begin{enumerate}[label=(\roman*)]
\item Monotonicity: $X \leq Y$ a.s.\ implies $\rho(X) \leq \rho(Y)$.
\item Translation equivariance: $\rho(X + c) = \rho(X) + c$ for $c \in \R$.
\item Positive homogeneity: $\rho(\lambda X) = \lambda \rho(X)$ for
      $\lambda \geq 0$.
\item Subadditivity: $\rho(X + Y) \leq \rho(X) + \rho(Y)$.
\end{enumerate}
A risk measure is convex if (iii) is dropped and (iv) replaced by
$\rho(\lambda X + (1 - \lambda) Y) \leq \lambda \rho(X) + (1 - \lambda) \rho(Y)$
for $\lambda \in [0, 1]$. Throughout the paper, $\rho$ is applied to a loss
$L$, so high $L$ is bad and high $\rho(L)$ flags risky positions.
\end{definition}

\begin{example}[Canonical coherent risk measures]
The following are coherent on $L^\infty$ under the loss convention: CVaR at
level $\alpha \in (0, 1)$,
$\CVaR_\alpha(L) = \alpha^{-1} \int_0^\alpha \VaR_u(L) \, du$ with
$\VaR_u(L) = \inf\{\ell: \Pr(L \leq \ell) \geq 1 - u\}$; entropic VaR,
$\EVaR_\alpha(L) = \inf_{\eta > 0} \eta \log(\alpha^{-1} \E[e^{L/\eta}])$;
worst-case loss, $\rho^{\sup}(L) = \esssup L$; and any finite mixture
$\sum_i \lambda_i \rho_i$ with $\lambda_i \geq 0$ and each $\rho_i$ coherent.
\end{example}

\begin{theorem}[Dual representation; \citealp{follmer2011stochastic}]
\label{thm:dual-rep}
Let $\rho: L^\infty \to \R$ be a coherent risk measure on losses satisfying
the Fatou property. Then there exists a non-empty convex weakly closed set
$\Q \subseteq \PP(\Omega)$ of probability measures absolutely continuous
w.r.t.\ $\Pr$ such that
\begin{equation}
\rho(X) = \sup_{Q \in \Q} \E_Q[X].
\label{eq:dual-rep}
\end{equation}
Conversely, every functional of this form is coherent under the loss
convention and has the Fatou property.
\end{theorem}

The set $\Q$ is the dual representation of $\rho$.
Theorem~\ref{thm:dual-rep} is the bridge between coherent-risk and
distributionally-robust portfolio theories.

\subsection{The Wasserstein geometry on \texorpdfstring{$\PP_p(\R^K)$}{Pp(R\^{}K)}}

\begin{proposition}[Kantorovich--Rubinstein duality]\label{prop:kr}
For $\mu, \nu \in \PP_1(\R^K)$,
\[
W_1(\mu, \nu) = \sup\left\{ \int f \, d\mu - \int f \, d\nu :
f: \R^K \to \R, \, \Lip(f) \leq 1 \right\}.
\]
\end{proposition}

\begin{proposition}[Lipschitz integral bound]\label{prop:lip-int}
If $f: \R^K \to \R$ is $L$-Lipschitz, then for any
$\mu, \nu \in \PP_1(\R^K)$,
$|\int f \, d\mu - \int f \, d\nu| \leq L \cdot W_1(\mu, \nu)$.
\end{proposition}

\begin{proposition}[Completeness]\label{prop:complete}
$(\PP_p(\R^K), W_p)$ is a complete separable metric space. $\mu_n \to \mu$
in $W_p$ iff $\mu_n \to \mu$ weakly and
$\int \|x\|^p \, d\mu_n \to \int \|x\|^p \, d\mu$.
\end{proposition}

\subsection{Standing assumptions}

\begin{assumption}[Standing regularity]\label{ass:standing}
Unless stated otherwise, the following are imposed throughout the paper.
\begin{enumerate}[label=(\roman*)]
\item The feasible set $\W \subset \R^K$ is non-empty, convex, compact, with
      $B_\W := \sup_{w \in \W} \|w\| < \infty$.
\item Return laws lie in $\PP_1(\R^K)$, and the conditional model
      $\theta \mapsto R_\theta$ is Borel measurable.
\item The utility $u: \R \to \R$ is non-decreasing, concave, and continuous;
      when Wasserstein-continuity is required, $u$ is assumed
      $L_u$-Lipschitz.
\item Risk measures act on losses $L(w, r) = -r^\top w$. Unless otherwise
      stated, $\rho$ is a law-invariant convex risk measure on losses
      (Definition~\ref{def:coherent}), lower semicontinuous with respect to
      $W_1$ convergence on $\PP_1(\R)$.
\item In static product-form problems, the allocation draw and return draw
      are conditionally independent given the information set and parameter:
      $\Gamma_\theta(dw, dr) = W_\theta(dw) R_\theta(dr)$.
      This is a modelling restriction, not part of the definition of DPO. In
      state-dependent policies, Bayesian policies, Thompson sampling,
      hierarchical randomization, or dynamic portfolio problems, we allow a
      general coupling $\Gamma_\theta \in \PP(\W \times \R^K)$, whose
      marginals are the weight law and return law.
\end{enumerate}
\end{assumption}

\begin{remark}[Compactness is mild]
Compactness of $\W$ covers the long-only simplex, box constraints, leverage
caps, and long--short portfolios with bounded gross exposure. It excludes
unbounded leverage, which is economically ill-posed without transaction-cost
or margin frictions.
\end{remark}

%------------------------------------------------------------------------------
% PATCH 3: NOTATION CONVENTION
%------------------------------------------------------------------------------
\subsection{Notation guide}

Several symbols carry context-dependent meaning across the static and
dynamic sections. Table~\ref{tab:notation} fixes the conventions.

\paragraph{Convention on overloaded symbols.}
The symbol $P$ denotes the parameter law throughout
Sections~\ref{sec:bayesian-dpo}--\ref{sec:weights}. In
Section~\ref{sec:distributional-rl} on dynamic distributional control, we
write $\Trans$ for the Markov transition kernel of the underlying MDP to
avoid collision with the parameter-law notation; the meaning is also
reaffirmed locally. Similarly, $\gamma$ is used as the risk-aversion
coefficient in static objectives and as the discount factor in the dynamic
Bellman equations; the static and dynamic sections are non-interleaved, so
no formula combines both meanings of $\gamma$.

\begin{table}[H]
\centering\footnotesize
\begin{tabular}{lll}
\toprule
Symbol & Meaning in static DPO (\S\S\ref{sec:bayesian-dpo}--\ref{sec:weights}) & Meaning in dynamic/RL (\S\ref{sec:distributional-rl}) \\
\midrule
$\W$            & feasible portfolio set                       & action space $\A$ \\
$w$             & portfolio weights                            & current action \\
$W$ or $Q$      & law of stochastic weights                    & policy-induced action law \\
$R$             & return law                                   & --- \\
$R_\theta$      & return law conditional on parameter $\theta$ & --- \\
$P$             & parameter / posterior law                    & --- \\
$\Trans$        & ---                                          & transition kernel \\
$P_n$           & posterior after $n$ observations             & --- \\
$\Rpred, \Rpredn$ & predictive return law $\int R_\theta dP_n$ & --- \\
$\Gamma_\theta$ & joint coupling on $\W \times \R^K$ given $\theta$ & --- \\
$\gamma$        & risk-aversion coefficient                    & discount factor \\
$\rho$          & law-invariant risk functional on scalar losses & --- \\
$Y^{\mathrm{cond}}_{\Gamma,\theta}$ & induced scalar outcome $\IP_\#\Gamma_\theta$ (parameter-fixed) & --- \\
$Y^\star_\Gamma$ & induced scalar outcome $\IP_\#\bar\Gamma$ (parameter-mixed) & --- \\
$Y^\star_{W,R,P}$ & product-form synonym for $Y^\star_\Gamma$ when $\Gamma = W \otimes R$ & --- \\
$Z^\pi$         & ---                                          & return-distribution function under $\pi$ \\
$\Tpi$          & ---                                          & distributional Bellman operator \\
$\Wbar_p$       & ---                                          & supremal Wasserstein metric on $\Z_p$ \\
\bottomrule
\end{tabular}
\caption{Notation guide.}
\label{tab:notation}
\end{table}

%==============================================================================
% SECTION 3: TAXONOMY
%==============================================================================
\section{The distributional triple: taxonomy}\label{sec:taxonomy}

We collect existing portfolio methodologies into a single taxonomy indexed by
the marginal triple $(W, R, P)$, while emphasizing that the mathematically
complete primitive is the joint coupling $\Gamma_\theta(dw, dr)$ of
allocations and returns.

\begin{definition}[DPO problem: coupled primitive and marginal triple]
\label{def:dpo-primitive}
Let $\W \subset \R^K$ be the feasible allocation set, let $u: \R \to \R$ be
a utility, let $\rho$ be a law-invariant risk functional on scalar losses,
let $\gamma \geq 0$ be a risk-aversion coefficient, and let
$P \in \PP(\Theta)$ be a parameter law.

The most general single-period DPO primitive is a Markov kernel
$\theta \mapsto \Gamma_\theta \in \PP(\W \times \R^K)$,
where $\Gamma_\theta$ is the joint law of the allocation $w$ and return
vector $r$ conditional on $\theta$. Write $\Gamma_\theta^W$ and
$\Gamma_\theta^R$ for its two marginals. The integrated joint law is
\[
\bar\Gamma(dw, dr) := \int_\Theta \Gamma_\theta(dw, dr) \, dP(\theta).
\]
The marginal triple associated with the coupled model is
\[
(\bar W, \bar R, P), \quad
\bar W(A) := \bar\Gamma(A \times \R^K), \quad
\bar R(B) := \bar\Gamma(\W \times B).
\]
The triple is useful as a taxonomy of marginal uncertainty sources, but the
decision-relevant primitive is the coupling $\Gamma_\theta$. Marginals alone
do not generally identify the law of the portfolio outcome
(Remark~\ref{rmk:marginals-no-id}).

The product-form DPO model is the special case
$\Gamma_\theta(dw, dr) = W_\theta(dw) R_\theta(dr)$.
If $W_\theta \equiv W$ does not depend on $\theta$, then
$\bar\Gamma(dw, dr) = W(dw) \Rpred(dr)$ with
$\Rpred := \int_\Theta R_\theta \, dP(\theta)$,
which recovers the product law used in the static sections. The ex-ante DPO
value is
\begin{equation}
\Vex(\Gamma, P) = \int u(y) \, dY^\star_\Gamma(y) - \gamma \rho(-Y^\star_\Gamma),
\label{eq:vex}
\end{equation}
and the conditional DPO value is
\begin{equation}
\Vcond(\Gamma, P) = \int_\Theta \left[\int u(y) \, dY^{\mathrm{cond}}_{\Gamma,\theta}(y)
- \gamma \rho(-Y^{\mathrm{cond}}_{\Gamma,\theta})\right] dP(\theta),
\label{eq:vcond}
\end{equation}
with $Y^{\mathrm{cond}}_{\Gamma,\theta}$ and $Y^\star_\Gamma$ defined in
Definition~\ref{def:induced} below. For a linear utility and a nonlinear
law-invariant $\rho$, $\Vex$ and $\Vcond$ generally differ.
\end{definition}

The ex-ante objective evaluates utility and risk against the marginal
predictive law, treating posterior uncertainty as additional ``volatility''
in the return mixture. The conditional objective evaluates utility and risk
separately under each $\theta$ and averages, exposing
parameter-by-parameter sensitivity. Both are economically defensible; we
use the ex-ante form by default and flag the conditional form where it
differs.

\subsection{The induced portfolio-outcome distribution}

The coupling $\Gamma_\theta$ captures the distributional inputs, but every
utility and risk functional ultimately operates on a derived scalar random
variable: the portfolio return $Y = r^\top w$.

\begin{definition}[Induced portfolio-outcome distribution]\label{def:induced}
Let $\IP: \R^K \times \R^K \to \R$, $\IP(w, r) = r^\top w$. For a coupled
DPO primitive $\Gamma_\theta$, define the conditional scalar
portfolio-outcome law
\[
Y^{\mathrm{cond}}_{\Gamma,\theta} := \IP_\# \Gamma_\theta \in \PP(\R),
\]
and the ex-ante (parameter-mixed) scalar portfolio-outcome law
\begin{equation}
Y^\star_\Gamma := \IP_\# \bar\Gamma \in \PP(\R).
\label{eq:y-star}
\end{equation}
In the product-form case $\Gamma_\theta = W \otimes R_\theta$,
$Y^\star_\Gamma = \IP_\# (W \otimes \Rpred)$,
which is the scalar outcome law used throughout the independent static
model. We retain the notation $Y^\star_{W,R,P}$ for this product-form case
when the marginal triple $(W, R, P)$ uniquely determines the coupling.
\end{definition}

\begin{remark}[Marginals do not identify the outcome law]\label{rmk:marginals-no-id}
The marginal triple $(W, R, P)$ does not determine the law of $r^\top w$
unless a coupling is specified. For example, in the scalar long--short case
let $W = R = \tfrac{1}{2}\delta_{-1} + \tfrac{1}{2}\delta_1$. The
comonotone coupling
$\Gamma^+ = \tfrac{1}{2}\delta_{(-1,-1)} + \tfrac{1}{2}\delta_{(1,1)}$ and
the countermonotone coupling
$\Gamma^- = \tfrac{1}{2}\delta_{(-1,1)} + \tfrac{1}{2}\delta_{(1,-1)}$
have the same marginals $W$ and $R$, but $rw = 1$ $\Gamma^+$-a.s.\ and
$rw = -1$ $\Gamma^-$-a.s. Thus the product law $W \otimes R$ is not innocuous:
it is an independence assumption. The coupled law $\Gamma$, rather than the
marginal triple alone, is the mathematically complete object.
\end{remark}

We use $Y^{\mathrm{cond}}_{\Gamma,\theta}$ for the conditional version
(parameter-fixed) and $Y^\star_\Gamma$ for the aggregate (parameter-mixed).
When $P = \delta_\theta$ the two coincide.

\begin{proposition}[Sufficiency of the induced law for ex-ante objectives]
\label{prop:induced-sufficiency}
Let $\rho$ be any law-invariant risk functional on $\PP(\R)$ (coherent,
convex, distortion, or spectral). Then for any DPO triple $(W, R, P)$ and
the linear loss $L(w, r) = -r^\top w$, the ex-ante value $\Vex(W, R, P)$ in
\eqref{eq:vex} depends on $(W, R, P)$ only through $Y^\star_{W,R,P}$. In
particular, two DPO triples that induce the same $Y^\star$ yield identical
$\Vex$. The same statement does not hold for the conditional value $\Vcond$
in \eqref{eq:vcond} unless $\rho$ is affine on $\PP(\R)$ or $R_\theta$ is
$P$-a.s.\ constant in $\theta$.
\end{proposition}

\begin{proof}
The ex-ante objective is a function of the law of $r^\top w$ under
$w \sim W$, $r \sim \Rpred$, which is exactly $Y^\star_{W,R,P}$;
law-invariance of $\rho$ and the integral against $u$ then give the claim.
For the negative claim, take $\rho = \CVaR_\alpha$ and two triples with
$W = \delta_w$ fixed: one with $R_\theta \equiv R_0$ and another with
$R_\theta \in \{R_+, R_-\}$ each on a half of $P$ such that
$(R_+ + R_-)/2 = R_0$. The ex-ante values agree, but $\Vcond$ generally
does not by Jensen's inequality applied to the convex CVaR.
\end{proof}

\begin{proposition}[Wasserstein ambiguity lifts to outcome ambiguity]
\label{prop:wass-lift}
Let $\A \subseteq \PP(\R^K)$ be a $W_1$-ambiguity set of return laws with
radius $\varepsilon$ centered at $\hat R_n$, and for each $R \in \A$ let
$Y_R = \IP_\#(W \otimes R)$ with $W$ held fixed. The map $R \mapsto Y_R$ is
Lipschitz with constant $c_W := \E_{w \sim W} \|w\|_\ast$, where
$\|\cdot\|_\ast$ is the dual of the ground norm underlying $W_1$. Hence
$\sup_{R \in \A} W_1(Y_R, Y_{\hat R_n})
\leq c_W \cdot \sup_{R \in \A} W_1(R, \hat R_n) \leq c_W \varepsilon$,
and the induced outcome ambiguity set is contained in the closed $W_1$-ball
of radius $c_W \varepsilon$ centered at $Y_{\hat R_n}$.
\end{proposition}

\begin{proof}
Fix $R, R' \in \PP(\R^K)$ and let $\pi$ be an optimal coupling realizing
$W_1(R, R')$. Draw $w \sim W$ independently of $(r, r') \sim \pi$; set
$y = r^\top w$, $y' = r'^\top w$. Then $y \sim Y_R$, $y' \sim Y_{R'}$. By
the dual-norm inequality $|(r - r')^\top w| \leq \|r - r'\| \|w\|_\ast$,
$\E|y - y'| \leq c_W \cdot W_1(R, R')$. The primal definition of $W_1$ gives
$W_1(Y_R, Y_{R'}) \leq c_W W_1(R, R')$. Specializing and taking the supremum
yields the claim.
\end{proof}

\begin{theorem}[Joint continuity of the induced outcome law]
\label{thm:pushforward-continuity}
Let $\W \subset \R^K$ be compact, $W, W' \in \PP_1(\W)$,
$R, R' \in \PP_1(\R^K)$, and let $Y_{W,R} = \IP_\#(W \otimes R)$. Then
\[
W_1(Y_{W,R}, Y_{W',R'}) \leq m_1(R) W_1(W, W') + B_\W W_1(R, R'),
\]
where $m_1(R) = \int \|r\| dR(r)$ and $B_\W = \sup_{w \in \W} \|w\|$.
Consequently, if $W_n \to W$ and $R_n \to R$ in $W_1$ with
$\sup_n m_1(R_n) < \infty$, then $Y_{W_n, R_n} \to Y_{W,R}$ in $W_1$.
\end{theorem}

\begin{proof}
Take any couplings $(w, w')$ of $(W, W')$ and $(r, r')$ of $(R, R')$,
independent of one another, and form the product coupling. Then
$|w^\top r - w'^\top r'| \leq \|w - w'\| \|r\| + \|w'\| \|r - r'\|$
by the dual-norm inequality. Taking expectations,
$\E|w^\top r - w'^\top r'| \leq \E\|r\| \cdot \E\|w - w'\| + B_\W \cdot \E\|r - r'\|$.
The pair $(w^\top r, w'^\top r')$ is a coupling of $(Y_{W,R}, Y_{W',R'})$,
so the primal definition of $W_1$ on $\R$ gives the result.
\end{proof}

\begin{theorem}[Existence of DPO optimizers; Lipschitz-utility version]
\label{thm:dpo-existence}
Let $\W \subset \R^K$ be compact and let $R \in \PP_1(\R^K)$ be fixed. For
$Q \in \PP(\W)$, define $Y_{Q,R} := \IP_\#(Q \otimes R)$. Assume that
$\gamma \geq 0$, that $u: \R \to \R$ is $L_u$-Lipschitz and continuous, and
that $\rho$ is lower semicontinuous with respect to $W_1$-convergence on
$\PP_1(\R)$. Then the static product-form problem
\begin{equation}
\sup_{Q \in \PP(\W)} \left\{ \E_{w \sim Q, r \sim R}[u(r^\top w)]
- \gamma \rho(-Y_{Q,R}) \right\}
\label{eq:dpo-stochastic}
\end{equation}
admits an optimizer $Q^\star \in \PP(\W)$. Restricting to deterministic
allocations $Q = \delta_w$, the deterministic problem
$\sup_{w \in \W} \{\E_R[u(r^\top w)] - \gamma \rho(-r^\top w)\}$
also admits an optimizer. If $u$ is continuous but not Lipschitz, the same
conclusion holds under the stronger assumption that $R$ has compact support.
\end{theorem}

\begin{proof}
Since $\W$ is compact, $\PP(\W)$ is compact under weak convergence, and on
a compact metric space weak convergence coincides with $W_1$-convergence.
Let $Q_n \to Q$ in $W_1$. Couple $w_n \sim Q_n$ and $w \sim Q$ optimally
and draw $r \sim R$ independently. Then
$|\E[u(r^\top w_n)] - \E[u(r^\top w)]| \leq L_u \E_R\|r\| \cdot \E\|w_n - w\| \to 0$.
Hence the expected-utility term is continuous in $Q$. By the continuity of
$Q \mapsto Y_{Q,R}$ in $W_1$ (Theorem~\ref{thm:pushforward-continuity}) and
lower semicontinuity of $\rho$, the map
$Q \mapsto -\gamma \rho(-Y_{Q,R})$ is upper semicontinuous. The objective
is therefore upper semicontinuous on the compact set $\PP(\W)$; Weierstrass
gives existence.
\end{proof}

\subsection{Taxonomy table}

\begin{table}[H]
\centering\footnotesize
\caption{Portfolio methodologies as points in $(W, R, P)$. ``$\delta$''
denotes a point mass.}
\label{tab:taxonomy}
\begin{tabular}{lllll}
\toprule
Method & $W$ & $R$ & $P$ & Description \\
\midrule
Markowitz                     & $\delta_w$ & $R_\theta$ (mean--cov.\ law) & $\delta_\theta$ & Deterministic params/weights \\
Mean--CVaR                    & $\delta_w$ & yes & $\delta_\theta$ & Risk functional of $R$ \\
Black--Litterman              & $\delta_w$ & Gaussian & posterior & Bayesian linear views \\
Bayesian                      & $\delta_w$ & predictive & posterior & Full posterior predictive \\
Robust                        & $\delta_w$ & ambiguity & $\delta_\theta$ & Worst case over $R \in \A$ \\
Wasserstein DRO               & $\delta_w$ & $W_p$-ball & $\delta_\theta$ & Ambiguity as transport ball \\
Chance-constrained            & $\delta_w$ & yes & $\delta_\theta$ & Probabilistic loss bound \\
Resampled efficiency          & $W$        & yes & posterior & Michaud-type resampling \\
HRP                           & $\delta_w$ & $\Sigma$ only & $\delta_\theta$ & Hierarchical recursive bisection \\
RA-HRP                        & $\delta_w$ & mean+cov & $\delta_\theta$ or post. & Return-adjusted HRP \\
Bayesian HRP                  & $W$        & Gaussian & posterior & Hierarchical stochastic tree \\
Dirichlet RL policy           & $W = \Dir(\alpha_\theta(s))$ & yes & $\delta_\theta$ & Stochastic policy class \\
Distributional RL             & $\delta_w$ per step & full return & $\delta_\theta$ & Learn Law(return) \\
\bottomrule
\end{tabular}
\end{table}

Three observations follow from Table~\ref{tab:taxonomy}. First, most of the
modern portfolio literature inhabits faces of the marginal triple where
exactly one of $W, R, P$ is non-trivial; interior cells where two or more
uncertainty sources are simultaneously non-trivial remain comparatively
underexplored. Second, several boundaries between cells are rigorous limits
when the limiting operation is specified precisely: Dirichlet policies
collapse to deterministic weights as concentration diverges; utility-tilted
posteriors approach optimistic or worst-case criteria as the tilt parameter
tends to infinity; and CVaR or chance-constrained formulations can be linked
through conservative penalty limits. A flat prior alone, however, does not
generically produce a worst-case robust objective. Third, the duality
results of Sections~\ref{sec:dro} and~\ref{sec:synthesis} establish
operational bridges between cells that look distinct, but these bridges
should be read as domination or calibration results rather than literal
equivalences.

%==============================================================================
% SECTION 4: BAYESIAN DPO
%==============================================================================
\section{Parameter distributions: Bayesian DPO}\label{sec:bayesian-dpo}

The Bayesian face of DPO elevates $P$ to a first-class object. Given
$r \mid \theta \sim R_\theta$ and prior $\theta \sim P_0$, the posterior
$P_n = P(\cdot \mid r_{1:n})$ defines the predictive
\begin{equation}
\Rpredn(A) = \int_\Theta R_\theta(A) \, dP_n(\theta), \quad A \in \B(\R^K).
\label{eq:predictive}
\end{equation}

\subsection{The Bayesian portfolio problem}

\begin{proposition}[Fubini reduction]
\label{prop:fubini}
Under integrability of $u(r^\top w)$ w.r.t.\ $P_n(\theta) \otimes R_\theta(r)$,
\[
\E_{\theta \sim P_n} \E_{r \sim R_\theta}[u(r^\top w)]
= \E_{r \sim \Rpredn}[u(r^\top w)].
\]
\end{proposition}

\begin{proposition}[Gaussian--NIW predictive is multivariate t]
\label{prop:niw-predictive}
Let $r_1, \ldots, r_n \mid \mu, \Sigma \overset{iid}{\sim} \N(\mu, \Sigma)$
with prior $(\mu, \Sigma) \sim \NIW(\mu_0, \lambda_0, \Psi_0, \nu_0)$,
$\nu_0 > K - 1$. Then the posterior is
$\NIW(\mu_n, \lambda_n, \Psi_n, \nu_n)$ with
$\lambda_n = \lambda_0 + n$, $\nu_n = \nu_0 + n$,
\[
\mu_n = \frac{\lambda_0 \mu_0 + n \bar r}{\lambda_n}, \quad
\Psi_n = \Psi_0 + S + \frac{\lambda_0 n}{\lambda_n}(\bar r - \mu_0)(\bar r - \mu_0)^\top,
\]
and the predictive is multivariate Student's $t$ with $\nu_n - K + 1 > 0$
degrees of freedom:
\begin{equation}
\Rpredn = t_K\!\left(\nu_n - K + 1, \mu_n,
\frac{\lambda_n + 1}{\lambda_n (\nu_n - K + 1)} \Psi_n\right).
\label{eq:niw-pred}
\end{equation}
\end{proposition}

\begin{remark}[Multivariate t is MVE; preservation of universality]
The predictive in \eqref{eq:niw-pred} is elliptical, hence MVE in the sense
of \citet{benveniste2024untangling}: the Bayesian-optimal portfolio
coincides with the mean--variance allocation computed from the predictive
mean and scale matrix, for any standard concave utility. Bayesian DPO under
Gaussian--NIW conjugacy inherits the universality of deterministic
Markowitz despite heavier tails. For the mean--variance interpretation to
be well-defined, we additionally require $\nu_n - K + 1 > 2$ (i.e.\
$\nu_n > K + 1$), ensuring that the predictive multivariate $t$ has finite
second moments; the existence of the density only requires $\nu_n - K + 1 > 0$.
\end{remark}

\subsection{Black--Litterman as a special case}

\begin{sloppypar}
Black--Litterman fits the present framework as Bayesian DPO with
$(W, R, P) = (\delta, \text{Gaussian}, \text{Gaussian})$. The prior is
$\mu \sim \N(\pi, \tau \Sigma)$ and the view equation is
$P \mu = q + \epsilon$ with $\epsilon \sim \N(0, \Omega)$. The posterior is
$\mu \mid \text{views} \sim \N(\mu^{BL}, \Sigma^{BL})$ with the standard
formulas of \citet{black1992global}. The predictive law is
$\N(\mu^{BL},\, \Sigma + \Sigma^{BL})$, and the associated
mean--variance portfolio is
$w^{BL} = (\gamma \Sigma^{\mathrm{eff}})^{-1} \mu^{BL}$.
\end{sloppypar}

\subsection{Utility-tilted posteriors}

When the decision-relevant quantity is the portfolio rather than the
parameters, \citet{kadane1995prime} and \citet{polson2000bayesian} advocate
the utility-tilted posterior
\begin{equation}
P_n^u(d\theta; w) \propto \exp\!\left(\lambda U(w; R_\theta)\right) P_n(d\theta),
\label{eq:tilted}
\end{equation}
with $U(w; R_\theta) = \E_{r \sim R_\theta}[u(r^\top w)]$ and sharpness
$\lambda > 0$.

\begin{assumption}[Uniform exponential envelope]\label{ass:exponential-envelope}
For each $\bar\lambda < \infty$,
$\E_{P_n}[\exp\{\bar\lambda U_\#(\theta)\}] < \infty$, where
$U_\#(\theta) := \sup_{w \in \W} |U(w; R_\theta)|$.
Equivalently, for every compact $\Lambda \subset [0, \infty)$ there exists
$M_\Lambda \in L^1(P_n)$ such that, for all $\lambda \in \Lambda$ and all
$w \in \W$,
$\exp\{\lambda |U(w; R_\theta)|\} \leq M_\Lambda(\theta)$ $P_n$-a.s.
\end{assumption}

\begin{theorem}[Existence and limits of the tilted portfolio]\label{thm:tilted}
Let $u$ be concave and non-decreasing, let $\W$ be compact and convex, and
suppose that $w \mapsto U(w; R_\theta)$ is continuous for $P_n$-a.e.\
$\theta$. Assume the uniform exponential envelope of
Assumption~\ref{ass:exponential-envelope}. For $\lambda > 0$, define
\[
V_\lambda^+(w) := \frac{1}{\lambda} \log \E_{P_n}\!\left[\exp\{\lambda U(w; R_\theta)\}\right],
\quad
V_\lambda^-(w) := -\frac{1}{\lambda} \log \E_{P_n}\!\left[\exp\{-\lambda U(w; R_\theta)\}\right],
\]
with $V_0^+(w) = V_0^-(w) := \E_{P_n}[U(w; R_\theta)]$. Then, for every
$\lambda \geq 0$, both $V_\lambda^+$ and $V_\lambda^-$ are continuous on
$\W$ and admit maximizers. For each fixed $w$, $\lambda \mapsto V_\lambda^+(w)$
is non-decreasing and $\lambda \mapsto V_\lambda^-(w)$ is non-increasing;
$V_\lambda^-(w) \leq V_0^\pm(w) \leq V_\lambda^+(w)$ for all $\lambda \geq 0$.
If the envelope condition holds for all $\lambda > 0$, then the large-tilt
limits hold in the extended-real sense:
$\lim_{\lambda \to \infty} V_\lambda^+(w) = \esssup_{P_n} U(w; R_\theta)$,
$\lim_{\lambda \to \infty} V_\lambda^-(w) = \essinf_{P_n} U(w; R_\theta)$.
\end{theorem}

\begin{proof}
Continuity follows from dominated convergence under the envelope.
Monotonicity is the standard property of entropic certainty equivalents.
The limits as $\lambda \to \infty$ are the log-Laplace principle applied to
$X = \pm U(w; R_\theta)$.
\end{proof}

% PATCH 4: REMARK 4.6 ADDITIONS AND BOUNDARY-LIMIT NOTE
\begin{remark}[Verification for standard utilities]\label{rmk:utility-envelope}
Bounded-above utility controls only the upper tail relevant to $V_\lambda^+$;
it does not by itself verify the two-sided envelope in
Assumption~\ref{ass:exponential-envelope}. For example, CARA utility
$u(x) = -e^{-ax}$ is bounded above by zero, but it can be arbitrarily
negative when $x$ is sufficiently negative. Therefore the robust tilt
$V_\lambda^-$ still requires lower-tail control. A convenient sufficient
condition is any one of: (i) returns are uniformly bounded on the posterior
support; (ii) the parameter posterior is restricted to a compact set on
which $\sup_{w \in \W} \E_{R_\theta}[e^{c|r^\top w|}] < \infty$ for the
required range of $c$; (iii) the problem is solved after return truncation.
For log and CRRA utilities, the domain restriction $1 + r^\top w > 0$ must
also be enforced $R_\theta$-a.s.

\smallskip
Finally, the large-tilt limits in Theorem~\ref{thm:tilted} are extended-real
limits. In unbounded-support models, $\esssup_{P_n} U(w; R_\theta)$ may be
$+\infty$ or $\essinf_{P_n} U(w; R_\theta)$ may be $-\infty$, even when
$V_\lambda^\pm(w)$ is finite for each finite $\lambda$. The theorem
identifies the limiting criterion but does not imply that the limiting
optimizer is finite or attained without additional compactness or tail
assumptions.
\end{remark}

\subsection{Posterior concentration and asymptotic plug-in equivalence}

\begin{assumption}[Posterior concentration]\label{ass:posterior-concentration}
There exist $M > 0$ and $\theta^\star \in \Theta$ such that for all
sufficiently large $n$, $\E_{P_n}[\|\theta - \theta^\star\|] \leq M n^{-1/2}$.
\end{assumption}

In a correctly specified, regular parametric model with a non-singular Fisher
information and a prior that places positive mass in a neighbourhood of the
true parameter, Assumption~\ref{ass:posterior-concentration} follows from the
Bernstein--von Mises theorem together with standard uniform integrability of
the posterior (see, e.g., \citealp{vandervaart2000}, Ch.~10). Under
misspecification, $\theta^\star$ is the KL-projection of the data-generating
distribution onto the parametric family; \citet{kleijn2012bernstein}
establish a Bernstein--von Mises-type theorem at $\theta^\star$ in Hellinger
or KL contraction, but obtaining the displayed posterior first-moment bound
$\E_{P_n}\|\theta - \theta^\star\| = O(n^{-1/2})$ additionally requires
uniform integrability of $\|\theta - \theta^\star\|$ under $P_n$. We therefore
treat Assumption~\ref{ass:posterior-concentration} as a hypothesis to be
verified case by case: in well-specified Gaussian--NIW models it follows
directly from the conjugate posterior, and in misspecified models it is
known to hold for sufficiently regular parameterizations with subgaussian
score, with $M$ depending on the model bias and the curvature of the KL
projection.

\begin{remark}[Failure modes of the parametric rate]
Four failure modes worth flagging: (1) rank-deficient samples ($n < K$);
(2) heavy-tailed returns with infinite fourth moment; (3) persistent
(near-unit-root) processes; (4) structurally misspecified covariance. The
data-driven radius $\varepsilon_n$ in
Theorem~\ref{thm:bayes-dro-domination} is computed from the posterior and
self-corrects.
\end{remark}

% PATCH 5: PROP 4.9 DEMOTED TO REMARK
\begin{remark}[Asymptotic plug-in equivalence as a special case of Theorem~\ref{thm:bayes-dro-domination}]
\label{rmk:plug-in-equivalence}
The asymptotic plug-in equivalence stated in earlier versions of this paper
as a stand-alone proposition is a quantitative specialization of the
Bayesian--DRO domination of Theorem~\ref{thm:bayes-dro-domination}: under
Assumption~\ref{ass:posterior-concentration} and a local Lipschitz model
map, the first-order radius $\varepsilon_n^W = O(n^{-1/2})$ controls the
plug-in--predictive gap on Lipschitz losses, and combining this with the
standard posterior bias bound
$\|\hat\theta - \theta^\star\| = O_p(n^{-1/2})$ yields the rate
$|\E_{\Rpredn}[\ell] - \E_{R_{\theta^\star}}[\ell]| = O(n^{-1/2})$ for
Lipschitz losses $\ell$. We therefore do not state it as an independent
proposition.
\end{remark}

%==============================================================================
% SECTION 5: COHERENT RISK AND TAIL MODELING
%==============================================================================
\section{Return distributions: coherent risk and tail modeling}\label{sec:returns}

When $R$ is non-trivial and $P$ is degenerate, DPO reduces to optimizing a
functional of $R$. The key analytical device is the dual representation of
the chosen risk measure.

\subsection{CVaR and the Rockafellar--Uryasev theorem}

\begin{theorem}[\citealp{rockafellar2000optimization}]\label{thm:ru}
Let $L(w, r) = -r^\top w$ and
\begin{equation}
F_\alpha(w, \eta) = \eta + \frac{1}{\alpha} \E_{r \sim R}[(L(w, r) - \eta)_+].
\label{eq:ru}
\end{equation}
Then $\CVaR_\alpha(L(w, \cdot)) = \min_{\eta \in \R} F_\alpha(w, \eta)$, and
the minimizer is $\eta^\star(w) = \VaR_\alpha(L(w, \cdot))$. The functional
$F_\alpha$ is jointly convex in $(w, \eta)$ for the linear loss.
\end{theorem}

For finitely-supported $R = \tfrac{1}{S} \sum_s \delta_{r_s}$, mean--CVaR
optimization is the LP
\begin{equation}
\min_{w \in \W, \eta \in \R, z \in \R^S_{\geq 0}}
\eta + \frac{1}{\alpha S} \sum_s z_s
\quad \text{s.t.} \quad z_s \geq -r_s^\top w - \eta.
\label{eq:cvar-lp}
\end{equation}

\subsection{Kusuoka and distortion representations}

\begin{theorem}[Kusuoka representation; \citealp{kusuoka2001law}]\label{thm:kusuoka}
Let $\rho: L^\infty \to \R$ be law-invariant, coherent on losses, and
satisfy the Fatou property on a non-atomic probability space. Then there
exists a non-empty convex weakly closed set $\M$ of probability measures on
$(0, 1]$ such that
\begin{equation}
\rho(L) = \sup_{\varphi \in \M} \int_{(0, 1]} \CVaR_u(L) \, d\varphi(u).
\label{eq:kusuoka}
\end{equation}
\end{theorem}

Two parametric subfamilies realize Theorem~\ref{thm:kusuoka} with closed
forms.

\paragraph{Spectral risk measures.}
For a non-increasing right-continuous density $\varphi: [0, 1] \to \R_+$
with $\int_0^1 \varphi(u) \, du = 1$, the spectral risk
\citep{acerbi2002spectral} is
$\rho_\varphi(L) = \int_0^1 \VaR_u(L) \varphi(u) \, du$.

\paragraph{Distortion risk measures.}
For a concave increasing distortion $g: [0, 1] \to [0, 1]$ with $g(0) = 0$,
$g(1) = 1$, the distortion risk on a loss $L$ \citep{wang2000class} is
$\rho_g(L) = \int_0^\infty g(S_L(x)) dx + \int_{-\infty}^0 (g(S_L(x)) - 1) dx$,
where $S_L(x) = \Pr(L > x)$.

\begin{example}[Three canonical distortions]
(i) CVaR: $g^{\CVaR}(u) = \min(u/\alpha, 1)$;
(ii) Wang transform: $g_\lambda(u) = \Phi(\Phi^{-1}(u) + \lambda)$ for
$\lambda \geq 0$, giving $\rho_{g_\lambda}(L) = \mu_L + \lambda \sigma_L$
for $L \sim \N(\mu_L, \sigma_L^2)$;
(iii) Proportional hazards: $g^{PH}(u) = u^\beta$ for $\beta \in (0, 1]$.
\end{example}

\subsection{Entropic value-at-risk and KL duality}

\begin{definition}[EVaR; \citealp{ahmadi2012entropic}]\label{def:evar}
For $\alpha \in (0, 1)$,
$\EVaR_\alpha(L) = \inf_{\eta > 0} \eta \log(\alpha^{-1} \E[e^{L/\eta}])$.
EVaR is coherent on $L^\infty$ and satisfies
$\VaR_\alpha \leq \CVaR_\alpha \leq \EVaR_\alpha$.
\end{definition}

%==============================================================================
% SECTION 6: DRO
%==============================================================================
\section{Distributionally robust DPO}\label{sec:dro}

DRO replaces a single $R$ by an ambiguity set $\A \subseteq \PP(\R^K)$ and
solves
\begin{equation}
\max_{w \in \W} \inf_{R \in \A} \E_{r \sim R}[u(r^\top w)]
\quad \text{or} \quad
\min_{w \in \W} \sup_{R \in \A} \rho_R(-r^\top w).
\label{eq:dro-problem}
\end{equation}

\subsection{Moment-based ambiguity}

\begin{definition}[Delage--Ye ambiguity]
For $\mu \in \R^K$, $\Sigma \succ 0$, $\kappa \geq 0$,
\[
\A^{\mathrm{mom}}(\mu, \Sigma, \kappa) = \{R \in \PP_2(\R^K):
(\E_R[r] - \mu)^\top \Sigma^{-1} (\E_R[r] - \mu) \leq \kappa,
\Cov_R(r) \preceq \Sigma\}.
\]
\end{definition}

\begin{theorem}[\citealp{delage2010distributionally}]
For piecewise-linear concave $u(x) = \min_k\{a_k x + b_k\}$, the inner
worst-case expectation reduces to an SDP with $O(K^2)$ variables and
$O(K^2)$ constraints. The outer maximization over $w \in \W$ preserves
convexity for polyhedral $\W$.
\end{theorem}

\subsection{Wasserstein DRO: linear-loss case}

\begin{definition}[Wasserstein ambiguity]
Given $\hat R_n = n^{-1} \sum_{i=1}^n \delta_{r_i}$ and $\varepsilon \geq 0$,
\[
\A_p(\varepsilon) = \{R \in \PP_p(\R^K): W_p(R, \hat R_n) \leq \varepsilon\}.
\]
\end{definition}

\begin{theorem}[Wasserstein DRO duality; \citealp{mohajerin2018data}, $p=1$]
\label{thm:wasserstein-duality}
For any upper semicontinuous loss $\ell(w, \cdot): \R^K \to \R$ with at
most linear growth in $r$,
\begin{equation}
\sup_{R \in \A_1(\varepsilon)} \E_{r \sim R}[\ell(w, r)]
= \inf_{\lambda \geq 0}
\left\{ \lambda \varepsilon
+ \frac{1}{n} \sum_{i=1}^n \sup_{r' \in \R^K}
\left[\ell(w, r') - \lambda \|r' - r_i\|\right] \right\}.
\label{eq:mek-duality}
\end{equation}
\end{theorem}

\begin{corollary}[Affine loss closed form]\label{cor:affine-dro}
For the linear loss $\ell_w(r) = -r^\top w$ with $\ell^s$-norm ground
metric (dual norm $\|\cdot\|_{(s^\ast)}$),
\begin{equation}
\sup_{R \in \A_1(\varepsilon)} \E_R[-r^\top w]
= \E_{\hat R_n}[-r^\top w] + \varepsilon \|w\|_{(s^\ast)},
\label{eq:affine-dro}
\end{equation}
or equivalently
$\inf_{R \in \A_1(\varepsilon)} \E_R[r^\top w]
= \E_{\hat R_n}[r^\top w] - \varepsilon \|w\|_{(s^\ast)}$.
\end{corollary}

%------------------------------------------------------------------------------
% PATCH 6: SECTION 6.3 TONE-DOWN
%------------------------------------------------------------------------------
\subsection{Wasserstein DRO of CVaR}

The piecewise-linear-convex case of \citet[Cor.~5.1]{mohajerin2018data}
extends Theorem~\ref{thm:wasserstein-duality} to CVaR for a generic convex
loss. The contribution of this subsection is the specialization to the
canonical CVaR loss with the standard $\ell^s$ ground norm and the resulting
explicit closed form (Corollary~\ref{cor:worst-case-cvar}) and conic program
(Proposition~\ref{prop:explicit-cvar-program}). The reformulation is the
form required for implementation on commodity solvers; we do not claim a
duality theorem beyond Mohajerin Esfahani--Kuhn Cor.~5.1.

\begin{theorem}[Wasserstein DRO of CVaR]\label{thm:wcvar-dro}
For the piecewise-linear loss $L(w, r) = -r^\top w$ and the
Rockafellar--Uryasev representation \eqref{eq:ru}, the worst-case CVaR over
$\A_1(\varepsilon)$ with $\ell^s$ ground metric admits the closed form
\begin{equation}
\sup_{R \in \A_1(\varepsilon)} \CVaR_\alpha^R(L(w, \cdot))
= \inf_{\eta \in \R} \left\{
\eta + \frac{1}{\alpha} \E_{\hat R_n}[(L(w, r) - \eta)_+]
+ \frac{\varepsilon}{\alpha} \|w\|_{(s^\ast)}
\right\}.
\label{eq:wcvar-dro}
\end{equation}
\end{theorem}

\begin{proof}
The Rockafellar--Uryasev integrand can be written as the maximum of two
affine functions of the return vector:
\begin{align*}
\theta_{\rm RU}(r;\tau,w)
&:= \tau + \frac{1}{\alpha}(-r^\top w-\tau)_+  \\
&= \max\left\{ \tau,\; \left(1-\frac{1}{\alpha}\right)\tau
        - \frac{1}{\alpha}r^\top w \right\}.
\end{align*}
Thus the two slope vectors in the notation of
\citet[Cor.~5.1]{mohajerin2018data} are $a_1 = 0$ and
$a_2 = -w/\alpha$, so
$\max_j\|a_j\|_{s^\ast} = \|w\|_{s^\ast}/\alpha$. Applying that corollary
to the robust CVaR epigraph gives directly
\begin{align*}
\sup_{R\in\mathcal A_1(\varepsilon)} \CVaR_\alpha^R(-r^\top w)
&= \inf_{\tau\in\R}
\left\{
  \frac{1}{n}\sum_{i=1}^n
  \theta_{\rm RU}(r_i;\tau,w)
  + \frac{\varepsilon}{\alpha}\|w\|_{s^\ast}
\right\}  \\
&= \inf_{\tau\in\R}
\left\{
\tau + \frac{1}{\alpha n}\sum_{i=1}^n (-r_i^\top w-\tau)_+
+ \frac{\varepsilon}{\alpha}\|w\|_{s^\ast}
\right\},
\end{align*}
which is \eqref{eq:wcvar-dro}. This is not obtained by first proving a
fixed-$\tau$ robust-expectation identity and then exchanging $\sup_R$ with
$\inf_\tau$; the strong-duality/minimax step for this piecewise-affine convex
loss class is exactly the content supplied by
\citet[Cor.~5.1]{mohajerin2018data}. Renaming the auxiliary variable $\tau$
back to $\eta$ yields the displayed statement.
\end{proof}

\begin{corollary}[Worst-case CVaR is regularized empirical CVaR]
\label{cor:worst-case-cvar}
For the linear loss $L = -r^\top w$,
\[
\sup_{R \in \A_1(\varepsilon)} \CVaR_\alpha^R(L(w, \cdot))
= \CVaR_\alpha^{\hat R_n}(L(w, \cdot))
+ \frac{\varepsilon}{\alpha} \|w\|_{(s^\ast)}.
\]
\end{corollary}

\begin{proposition}[Explicit Wasserstein-CVaR program]
\label{prop:explicit-cvar-program}
The Wasserstein-CVaR DRO portfolio problem is equivalent to the conic
program
\begin{equation}
\begin{aligned}
\min_{w, \eta, z, t} \quad & \eta + \frac{1}{\alpha n} \sum_{i=1}^n z_i + \frac{\varepsilon}{\alpha} t \\
\text{s.t.} \quad & z_i \geq -r_i^\top w - \eta, \quad i = 1, \ldots, n, \\
& z_i \geq 0, \quad i = 1, \ldots, n, \\
& \|w\|_\ast \leq t, \\
& w \in \W.
\end{aligned}
\label{eq:wcvar-program}
\end{equation}
For $\|\cdot\|_\ast = \|\cdot\|_2$, this is an SOCP; for $\|\cdot\|_1$ or
$\|\cdot\|_\infty$, an LP.
\end{proposition}

\subsection{Finite-sample DRO guarantees}

\begin{theorem}[\citealp{mohajerin2018data}; sharp rate from \citealp{fournier2015rate}]
\label{thm:dro-finite-sample}
Under light-tail conditions on $R^\star$, there exist $c_1, c_2 > 0$ such
that for any $\delta \in (0, 1)$,
$\varepsilon_n(\delta) = (c_1 \log(c_2/\delta) / n)^{1/\max(K, 2)}$
satisfies $\Pr(W_1(\hat R_n, R^\star) \leq \varepsilon_n(\delta)) \geq 1 - \delta$.
The exponent $1/\max(K, 2)$ matches the sharp rate of
$\E[W_1(\hat R_n, R^\star)]$ established for absolutely continuous $R^\star$
in \citet[Theorem~1]{fournier2015rate}. At the critical dimension $K = 2$,
logarithmic factors can appear depending on the moment and concentration
regime; since the empirical application has $K = 25$, this critical-dimensional
refinement does not affect the interpretation here.
\end{theorem}

\subsection{\texorpdfstring{$\phi$-divergence ambiguity and EVaR--KL duality}{phi-divergence ambiguity and EVaR-KL duality}}

\begin{definition}[$\phi$-divergence ambiguity]
For convex $\phi: \R_+ \to \R$ with $\phi(1) = 0$,
$D_\phi(R \| R_0) = \int \phi(dR / dR_0) dR_0$ and
$\A^\phi(\eta) = \{R \ll R_0: D_\phi(R \| R_0) \leq \eta\}$.
\end{definition}

\begin{proposition}[KL-DRO and EVaR; \citealp{bental2013robust}]
Let $\A^{\KL}(\eta) = \{R \ll R_0: \KL(R \| R_0) \leq \eta\}$ and
$\ell = -r^\top w$. Then
\begin{equation}
\sup_{R \in \A^{\KL}(\eta)} \E_R[\ell]
= \inf_{\lambda > 0} \{\lambda \log \E_{R_0}[e^{\ell / \lambda}] + \lambda \eta\},
\label{eq:kl-dro}
\end{equation}
which coincides with $\EVaR_\alpha(\ell)$ at confidence level
$\alpha = e^{-\eta}$, equivalently $\eta = \log(1/\alpha)$: the entropic
penalty $\lambda \eta$ in \eqref{eq:kl-dro} is exactly the
$\eta \log(\alpha^{-1})$ term in Definition~\ref{def:evar} after the
substitution $\eta \mapsto \lambda$. KL radius and EVaR confidence level are
therefore two parameterizations of the same one-parameter family of
coherent risk measures.
\end{proposition}

\subsection{Comparison and tractability}

\begin{table}[H]
\centering\footnotesize
\caption{Comparison of DRO ambiguity structures.}
\label{tab:dro-comparison}
\begin{tabular}{llll}
\toprule
Ambiguity & Data fit & Regularizer & Tractability \\
\midrule
Moment (Delage--Ye)                            & $\mu, \Sigma$ only & SDP convex & SDP \\
Wasserstein, linear loss                       & surrounds $\hat R_n$ & dual-norm on $w$ & LP/QP/SOCP \\
Wasserstein, CVaR (Thm.~\ref{thm:wcvar-dro})   & surrounds $\hat R_n$ & $\varepsilon/\alpha$ dual-norm & LP/SOCP \\
KL / $\phi$-div                                & abs.\ continuity vs $R_0$ & entropic & convex, closed form \\
\bottomrule
\end{tabular}
\end{table}

%==============================================================================
% SECTION 7: CHANCE-CONSTRAINED
%==============================================================================
\section{Chance-constrained DPO}\label{sec:cc}

The canonical chance-constrained portfolio problem
\citep{charnes1959chance} is
\begin{equation}
\max_{w \in \W} \E_R[r^\top w]
\quad \text{s.t.} \quad
\Pr_R(r^\top w \leq -a) \leq \varepsilon,
\label{eq:cc-problem}
\end{equation}
with significance level $\varepsilon \in (0, 1)$ and loss bound $a > 0$.

\subsection{Gaussian reformulation}

\begin{proposition}[Gaussian SOC reduction]
If $r \sim \N(\mu, \Sigma)$ with $\Sigma \succ 0$, the chance constraint in
\eqref{eq:cc-problem} is equivalent to the SOC constraint
\begin{equation}
\mu^\top w - \Phi^{-1}(1 - \varepsilon) \sqrt{w^\top \Sigma w} \geq -a,
\label{eq:cc-soc}
\end{equation}
which is convex for $\varepsilon < 1/2$.
\end{proposition}

\subsection{CVaR inner approximation}

\begin{proposition}[CVaR convex inner approximation]
For every $R \in \PP_1(\R^K)$,
$\CVaR_\varepsilon(-r^\top w) \leq a \Rightarrow \Pr_R(r^\top w \leq -a) \leq \varepsilon$.
The converse fails in general.
\end{proposition}

\begin{remark}[Quantitative tightness gap]\label{rmk:cvar-tightness}
For continuous $R$, the CVaR inner approximation is tight only when the
$\varepsilon$-tail is concentrated at one point. The conservatism gap is
$\Delta_\varepsilon(w) := \CVaR_\varepsilon(-r^\top w) - \VaR_\varepsilon(-r^\top w)$,
the average tail excess above the $\varepsilon$-quantile. For tail laws in
the $(\xi, \beta)$-generalized-Pareto class above $\VaR_\varepsilon$,
$\Delta_\varepsilon(w) = \beta / (1 - \xi)$. Gaussian tails ($\xi \to 0$)
give $\Delta_\varepsilon \to 0$ as $\varepsilon \to 0^+$; heavy-tailed
Pareto-$\alpha$ ($\xi = 1/\alpha$) gives
$\Delta_\varepsilon = \VaR_\varepsilon / (\alpha - 1)$, which does not
vanish---the regime where chance-constrained portfolio choice is most
consequential, motivating the joint use of CVaR with Wasserstein DRO.
\end{remark}

\subsection{Scenario approximation: Calafiore--Campi}

\begin{theorem}[\citealp{calafiore2005uncertain}; explicit constant from \citealp{campi2008exact}]
\label{thm:scenario}
Let $\{r_s\}_{s=1}^S$ be iid from $R$, and consider
\begin{equation}
\max_{w \in \W} \bar\mu^\top w
\quad \text{s.t.} \quad -r_s^\top w \leq a, \quad \forall s = 1, \ldots, S.
\label{eq:scenario-problem}
\end{equation}
Let $d$ be the Helly (support-rank) dimension. Then for
$\varepsilon, \delta \in (0, 1)$,
$\Pr^S\{\Pr_R(-r^\top w^S > a) > \varepsilon\}
\leq \sum_{j=0}^{d-1} \binom{S}{j} \varepsilon^j (1 - \varepsilon)^{S-j}$
\citep{calafiore2005uncertain}. The optimized exponential--polynomial
trade-off in this binomial tail yields the explicit conservative sufficient
condition
\begin{equation}
S \geq \frac{e}{e - 1} \cdot \frac{d - 1 + \log(1/\delta)}{\varepsilon}
\label{eq:scenario-bound}
\end{equation}
\citep[Thm.~1]{campi2008exact}.
\end{theorem}

% PATCH 7: REMARK 7.5 NONDEGENERACY QUALIFIER
\begin{remark}[Effective dimensions and the explicit constant]\label{rmk:helly}
Long-only simplex $\Delta^{K-1}$: $d = K - 1$. Box $[\ell, u]^K$ without
budget: $d = K$. Box plus budget: $d = K - 1$. Sector budgets with $m$
equalities: $d = K - m$. The Helly dimension $d$ is defined under the
standard nondegeneracy assumption that the optimum is unique with
probability one under $P^S$; at vertex degeneracies (multiple binding
constraints supporting the same optimal value) $d$ can be smaller than the
tabulated formulas suggest, in which case the scenario bound
\eqref{eq:scenario-bound} is only conservative. The factor $e/(e-1) \approx 1.582$
in \eqref{eq:scenario-bound} comes from optimizing the trade-off between the
polynomial and exponential factors in the binomial tail.
\end{remark}

%==============================================================================
% SECTION 8: WEIGHTS / STOCHASTIC ALLOCATIONS
%==============================================================================
\section{Weight distributions: stochastic allocations}\label{sec:weights}

We turn to the setting where $W$ is non-trivial. The first question is
whether stochastic allocation can improve a static risk-adjusted objective
relative to its mean. The answer is no, in considerable generality; this
delimits when stochastic policies are genuinely valuable.

\subsection{Static randomization cannot improve concave risk-adjusted utility}

\begin{theorem}[Static no-randomization under ex-ante independent randomization]
\label{thm:no-randomization}
Let $Q \in \PP(\W)$ be a stochastic allocation law with mean
$\bar w := \E_Q[w] \in \W$, which is feasible because $\W$ is convex and
compact. Suppose the allocation draw $w \sim Q$ is made ex ante and
independently of the return draw $r \sim R$. Suppose $\gamma \geq 0$, $u$
is concave and non-decreasing, and $\rho$ is a law-invariant convex risk
measure satisfying dilatation monotonicity:
$\rho(\E[L \mid \mathcal{G}]) \leq \rho(L)$
for every integrable loss $L$ and every sub-$\sigma$-algebra $\mathcal{G}$. Then
\[
\E_{r \sim R}[u(r^\top \bar w)] - \gamma \rho(-r^\top \bar w)
\geq \E_{w \sim Q, r \sim R}[u(r^\top w)] - \gamma \rho(-r^\top w),
\]
where the right-hand side is evaluated under the product law $Q \otimes R$.
\end{theorem}

\begin{proof}
\emph{Utility.} Pointwise in $r$, concavity gives
$\E_{w \sim Q} u(r^\top w) \leq u(r^\top \E_Q[w]) = u(r^\top \bar w)$.
Integrate over $R$.

\emph{Risk.} Set $L = -r^\top w$ on the product space and condition on $r$.
Independence gives $\E[L \mid r] = -r^\top \bar w$. By dilatation
monotonicity, $\rho(-r^\top \bar w) = \rho(\E[L \mid r]) \leq \rho(L) = \rho(-r^\top w)$.
Combining yields the claim.
\end{proof}

CVaR, EVaR, entropic risk, and the worst-case loss all satisfy dilatation
monotonicity \citep[\S 4.5]{follmer2011stochastic}. On an atomless
probability space, dilatation monotonicity is automatic for law-invariant
convex risk measures \citep[Theorem~4.59]{follmer2011stochastic};
\citet{cherny2006weighted} provides a direct proof. The hypothesis in
Theorem~\ref{thm:no-randomization} is therefore not an additional
restriction beyond law-invariance and convexity. On an atomless space the
result thus reduces to the statement that ex-ante independent randomization
cannot help a concave-utility / convex-risk objective---a measure-theoretic
sharpening of folklore, included here because it delimits the role of the
stochastic-allocation classes of \S\ref{sec:weights} rather than because the
argument is itself new.

\begin{remark}[When stochastic weights matter]
Theorem~\ref{thm:no-randomization} does not make stochastic policies
useless; it sharpens their role. Randomized weights are valuable when at
least one of: (i) exploration in dynamic RL; (ii) entropy regularization;
(iii) nonconvex or discrete allocation constraints; (iv) transaction-cost
frictions creating path dependence; (v) model averaging across heterogeneous
strategies; (vi) hierarchical allocation where randomness is realized before
conditioning on a latent regime.
\end{remark}

\subsection{Dirichlet policies on the simplex}

\begin{definition}[Dirichlet policy]
For $\W = \Delta^{K-1}$ and a state space $\StateSpace$, a Dirichlet policy is the
kernel
\begin{equation}
W_\theta(\cdot \mid s) = \Dir(\alpha_\theta(s)), \quad
\alpha_\theta(s) = \softplus(f_\theta(s)) + \epsilon \mathbf{1} \in \R^K_{>0},
\label{eq:dirichlet}
\end{equation}
for a network $f_\theta: \StateSpace \to \R^K$ and floor $\epsilon > 0$. The
state-dependent concentration is $\alpha_{0,\theta}(s) = \sum_i \alpha_{i,\theta}(s)$.
\end{definition}

\begin{proposition}[Score, entropy, moments]\label{prop:dirichlet-moments}
For $W_\theta(\cdot \mid s) = \Dir(\alpha_\theta(s))$:
(i) Score: $\nabla_\theta \log W_\theta(w \mid s) = \sum_{i=1}^K [\log w_i - \psi(\alpha_{i,\theta}(s)) + \psi(\alpha_{0,\theta}(s))] \nabla_\theta \alpha_{i,\theta}(s)$;
(ii) Entropy: $H(W_\theta(\cdot \mid s)) = \log B(\alpha) + (\alpha_0 - K) \psi(\alpha_0) - \sum_{i=1}^K (\alpha_i - 1) \psi(\alpha_i)$;
(iii) Moments: $\E[w_i] = \alpha_i / \alpha_0$,
$\Var(w_i) = \alpha_i (\alpha_0 - \alpha_i) / (\alpha_0^2 (\alpha_0 + 1))$,
$\Cov(w_i, w_j) = -\alpha_i \alpha_j / (\alpha_0^2 (\alpha_0 + 1))$ for $i \neq j$.
\end{proposition}

\begin{proposition}[Expected utility under a Dirichlet policy: a tractable bound]
\label{prop:dirichlet-utility}
Let $u$ be concave, $W = \Dir(\alpha)$ with $\alpha_0 = \sum_i \alpha_i$,
and $R$ have finite second moment. Define $\bar w = \alpha / \alpha_0$ and
$\Sigma_W = \Cov_W(w)$. Then
\begin{enumerate}[label=(\roman*)]
\item (Upper bound, Jensen in $w$):
$\E_{w \sim W} \E_{r \sim R}[u(r^\top w)] \leq \E_{r \sim R}[u(r^\top \bar w)]$.
\item (Lower bound, second-order Taylor): if $u$ is twice continuously
differentiable on $I$ with $|u''(x)| \leq M_2$ for $x \in I$, and
$r^\top w \in I$ for $W \otimes R$-a.e.\ $(w, r)$, then
\[
\E_{w \sim W}\!\E_{r \sim R}[u(r^\top w)]
\;\geq\; \E_{r \sim R}[u(r^\top \bar w)]
- \tfrac{M_2}{2}\, \E_{r \sim R}[r^\top \Sigma_W r].
\]
\item (Policy gradient, REINFORCE): for the differentiable parameterization
$\alpha = \alpha_\theta(s)$,
\[
\nabla_\theta \E_{w \sim W_\theta}\!\E_{r \sim R}[u(r^\top w)]
\;=\; \E_{w,\, r}\!\left[u(r^\top w)\,
\nabla_\theta \log W_\theta(w \mid s)\right].
\]
\end{enumerate}
\end{proposition}

% PATCH 8: PROP 8.5(ii) WORST-CASE M_2 NOTE
\begin{remark}[On the constant $M_2$]\label{rmk:domain-hypothesis}
For bounded support $\supp(R) \subseteq [-r_{\max}, r_{\max}]^K$ and long-only
$\W = \Delta^{K-1}$, standard utilities admit explicit $M_2$ on
$I = [-r_{\max}, r_{\max}]$: log utility (valid when $r_{\max} < 1$)
gives $M_2^{\log} = (1 - r_{\max})^{-2}$; CRRA gives
$M_2^{\mathrm{CRRA}} = \gamma (1 - r_{\max})^{-1-\gamma}$; CARA gives
$M_2^{\mathrm{CARA}} = a^2 e^{a r_{\max}}$, valid on all of $\R$.

The constants $M_2$ above are worst-case curvature bounds over the full
interval. For log utility, $(1 - r_{\max})^{-2}$ is attained only at the
lower endpoint $x = -r_{\max}$, so the resulting variance penalty
$\frac{M_2}{2} \E_R[r^\top \Sigma_W r]$ is a conservative global lower bound
rather than a tight practical approximation; a tighter local constant
$M_2(r) = \sup_{|x| \leq |r|} |u''(x)|$ can be substituted when the
predictive law is concentrated away from the domain boundary.
\end{remark}

\subsection{Long-only alternatives and a true long--short construction}

\begin{definition}[Logistic-normal simplex policy]
A logistic-normal simplex policy is $w = \softmax(z)$,
$z \sim \N(\mu_\theta(s), \Sigma_\theta(s))$.
This is a long-only stochastic allocation; it is not long--short.
\end{definition}

\begin{definition}[Two-book long--short policy]
Let $w^+ \sim \Dir(\alpha^+_\theta(s))$ and $w^- \sim \Dir(\alpha^-_\theta(s))$
be independent draws on the simplex, let $q(s) \geq 0$ be a state-dependent
short-leverage, and let $(\A_+, \A_-)$ be a fixed disjoint partition of the
asset universe such that
\begin{equation}
\supp(\alpha^+_\theta(s)) \subseteq \A_+, \quad
\supp(\alpha^-_\theta(s)) \subseteq \A_-.
\label{eq:disjoint-supports}
\end{equation}
Define $w = (1 + q(s)) w^+ - q(s) w^- \in \R^K$. Then $\mathbf{1}^\top w = 1$
and $\|w\|_1 = 1 + 2 q(s)$.
\end{definition}

\begin{remark}[On the disjoint-supports condition]
The disjoint-supports condition \eqref{eq:disjoint-supports} is a theoretical
convenience that yields a clean closed form and a budget-preserving construction.
Most empirically interesting long--short equity strategies violate disjointness:
the same asset can appear in both books at different weights through factor
crowding, sector rotation, or pairs trading. A more general construction allows
overlapping books with a netting operator $w = (1+q)w^+ - q w^-$ followed by
projection onto $\{w : \mathbf 1^\top w = 1,\ \|w\|_1 \le 1+2q\}$; the explicit
formulas \eqref{eq:disjoint-supports} then become inequality constraints rather
than identities. We do not pursue the overlapping variant here.
\end{remark}

\subsection{Hierarchical allocations: HRP, RA-HRP, Bayesian HRP}

\paragraph{HRP.}
\citet{lopez2016building} builds a binary tree $\mathcal{T}$ on the asset
universe via hierarchical clustering of the correlation matrix. At each
internal node with children $(L, R)$ and risk contributions $(v_L, v_R)$,
$\phi_L = v_R / (v_L + v_R)$, $\phi_R = v_L / (v_L + v_R)$. HRP uses only
$\Sigma$ and produces deterministic $w^{\mathrm{HRP}} \in \Delta^{K-1}$.

\paragraph{RA-HRP.}
\citet{noguer2025rahrp} extends HRP by replacing the inverse-variance split
with a risk-adjusted-return split:
$\phi_L = (\mu_L / v_L) / (\mu_L / v_L + \mu_R / v_R)$.

\paragraph{Bayesian HRP.}
A fully distributional version replaces each split by a Beta prior
$\phi_v \sim \mathrm{Beta}(\tau v_R, \tau v_L)$ at every internal node of
$\mathcal{T}$, with concentration $\tau > 0$.

\begin{proposition}[Bayesian HRP recovers HRP/RA-HRP in the high-concentration limit]
\label{prop:bhrp-limit}
Let $w^{\mathrm{BHRP}}(\tau) \in \Delta^{K-1}$ denote a draw from the
Bayesian HRP leaf-weight law with concentration $\tau$. Then as $\tau \to \infty$,
$w^{\mathrm{BHRP}}(\tau) \xrightarrow{L^2} w^{\mathrm{HRP}}$ (resp.\
$w^{\mathrm{RA\text{-}HRP}}$).
\end{proposition}

% PATCH 13: EXPANDED PROOF
\begin{proof}
For each internal node $v$ of the tree $\mathcal{T}$, the split fraction
$\phi_{v,\tau} \sim \mathrm{Beta}(\tau \alpha_v, \tau \beta_v)$ with
$\alpha_v, \beta_v$ the node-mean parameters. The Beta distribution has mean
$\E[\phi_{v,\tau}] = \alpha_v / (\alpha_v + \beta_v) =: \bar\phi_v$
independently of $\tau$, and variance
\[
\Var[\phi_{v,\tau}] = \frac{\alpha_v \beta_v}{(\alpha_v + \beta_v)^2 \, (\tau(\alpha_v + \beta_v) + 1)}
= O(\tau^{-1}).
\]
Hence $\phi_{v,\tau} \to \bar\phi_v$ in $L^2$ as $\tau \to \infty$.
Independence across nodes is the standing hypothesis on the prior.

For a leaf $\ell$ at depth $d_\ell$, the leaf weight is the finite product
$w_{\ell,\tau} = \prod_{v \in \mathrm{path}(\ell)} \psi_{v,\tau}$,
where $\psi_{v,\tau} \in \{\phi_{v,\tau}, 1 - \phi_{v,\tau}\}$ depending on
path direction. The deterministic limit is
$\bar w_\ell = \prod_{v \in \mathrm{path}(\ell)} \bar\psi_v$. Because all
factors and partial products are bounded by one, the telescoping identity
\[
\prod_{j=1}^{m} a_j - \prod_{j=1}^{m} b_j
= \sum_{j=1}^m \left(\prod_{i<j} a_i\right) (a_j - b_j) \left(\prod_{i>j} b_i\right)
\]
gives
\[
\|w_{\ell,\tau} - \bar w_\ell\|_{L^2}
\leq \sum_{j=1}^{d_\ell} \|\psi_{j,\tau} - \bar\psi_j\|_{L^2}
= O\!\left(d_\ell\,\tau^{-1/2}\right),
\]
where $d_\ell$ is the depth of leaf $\ell$ in the tree. For balanced HRP
$d_\ell = \Theta(\log K)$; for the worst-case unbalanced tree
$d_\ell = O(K)$. The hidden constant in $O(\tau^{-1/2})$ is therefore
$O(\log K)$ in the balanced case and $O(K)$ in the worst case; the depth
factor affects the prefactor of the rate but not the rate exponent. Since the tree has finitely many leaves, convergence holds in $L^2$ for the
entire allocation vector. The RA-HRP case is identical with the
corresponding split parameters.
\end{proof}

%==============================================================================
% SECTION 9: DISTRIBUTIONAL RL
%==============================================================================
\section{Distributional reinforcement learning for portfolios}\label{sec:distributional-rl}

Standard RLPO learns
$Q^\pi(s, w) = \E^\pi[\sum_t \gamma^t R_t \mid s_0 = s, w_0 = w]$.
Distributional RL \citep{bellemare2017distributional} learns the full return
distribution
$Z^\pi(s, w) = \Law(\sum_t \gamma^t R_t \mid s_0 = s, w_0 = w)$.

\subsection{The portfolio MDP and distributional Bellman operator}

Let $(\StateSpace, \A, \Trans, \gamma)$ be a discounted portfolio MDP with $\A = \W$.
To allow turnover penalties, the one-period reward is allowed to depend on
the current action, next state, and next action:
$g(s, w, s', w') := w^\top r(s, s') - c(w, w')$.
The transition is $s' \sim \Trans(\cdot \mid s, w)$, and the next action is
$w' \sim \pi(\cdot \mid s')$.

For $p \geq 1$, define the complete metric space
\[
\Z_p := \left\{ Z: \StateSpace \times \A \to \PP_p(\R) :
Z \text{ measurable, } \sup_{s,w} \int |x|^p dZ(s, w)(x) < \infty \right\},
\]
equipped with the supremal Wasserstein metric
$\Wbar_p(Z, Z') := \sup_{(s,w) \in \StateSpace \times \A} W_p(Z(s, w), Z'(s, w))$.

\begin{definition}[Distributional Bellman operator]\label{def:distributional-bellman}
For a fixed policy $\pi$, the operator $\Tpi: \Z_p \to \Z_p$ is defined as
follows. For each $(s, w)$, $(\Tpi Z)(s, w)$ is the law of
\begin{equation}
g(s, w, s', w') + \gamma G,
\label{eq:bellman-target}
\end{equation}
where $s' \sim \Trans(\cdot \mid s, w)$, $w' \sim \pi(\cdot \mid s')$, and
$G \mid (s', w') \sim Z(s', w')$, with $G$ conditionally independent of the
transition--policy draw given $(s', w')$.
\end{definition}

\begin{theorem}[Distributional Bellman contraction for portfolio rewards]
\label{thm:distributional-bellman}
Assume $|g(s, w, s', w')| \leq R_{\max} < \infty$ for all admissible
$(s, w, s', w')$. Then $\Tpi$ maps $\Z_p$ into itself and is a
$\gamma$-contraction:
$\Wbar_p(\Tpi Z, \Tpi Z') \leq \gamma \Wbar_p(Z, Z')$.
Consequently, $\Tpi$ admits a unique fixed point $Z^\pi \in \Z_p$, and for
any $Z^0 \in \Z_p$ the iterates $Z^{k+1} = \Tpi Z^k$ converge to $Z^\pi$ in
$\Wbar_p$.
\end{theorem}

\begin{proof}
Fix $(s, w)$. Let $\mu(s, w)$ and $\mu'(s, w)$ denote $(\Tpi Z)(s, w)$ and
$(\Tpi Z')(s, w)$. Draw $(s', w')$ from $\Trans(\cdot \mid s, w) \otimes \pi(\cdot \mid s')$
as shared randomness in both. Conditional on $(s', w')$, let $\Pi^\star_{s',w'}$
be the optimal coupling between $Z(s', w')$ and $Z'(s', w')$ realizing the
infimum in \eqref{eq:wasserstein-def}. Existence of a Borel-measurable
selector $(s', w') \mapsto \Pi^\star_{s',w'}$ follows from the measurable
selection theorem for optimal transport plans \citep[Cor.~5.22]{villani2009optimal}
(see also \citet[Lem.~1.2.3]{ambrosio2008gradient} for a self-contained
proof using $\Gamma$-convergence).

Draw $(G, G') \mid (s', w') \sim \Pi^\star_{s',w'}$ and define
$X := g(s, w, s', w') + \gamma G$, $X' := g(s, w, s', w') + \gamma G'$.
Then $X \sim \mu(s, w)$ and $X' \sim \mu'(s, w)$. By Fubini--Tonelli,
\[
\E[|X - X'|^p] = \gamma^p \E_{s',w'}[W_p^p(Z(s', w'), Z'(s', w'))]
\leq \gamma^p \Wbar_p^p(Z, Z').
\]
The primal definition gives $W_p^p(\mu(s, w), \mu'(s, w)) \leq \gamma^p \Wbar_p^p(Z, Z')$.
Taking $p$th root and supremum yields the contraction.
\end{proof}

% PATCH 9: SCOPE-OF-EXTENSION REMARK
\begin{remark}[Scope of the action-coupled reward extension]\label{rmk:turnover-cost-scope}
The reward $g(s, w, s', w')$ in Definition~\ref{def:distributional-bellman}
is allowed to depend on the current action $w$, the next state $s'$, and the
next action $w'$, which accommodates turnover frictions of the form
$c(w, w')$. This generality is notational rather than substantive: the
synchronous coupling at $(s', w')$ used in the proof of
Theorem~\ref{thm:distributional-bellman} is exactly the standard
distributional-Bellman argument of \citet{bellemare2017distributional} and
\citet{dabney2018distributional}, and the contraction constant remains
$\gamma$ regardless of whether $g$ depends on $(w, w')$. An honestly new
contraction result for genuinely action-coupled costs would require either
state augmentation to track past actions (which destroys the Markov
structure on $\StateSpace$ alone) or a state-dependent Lipschitz bound on $g$
yielding a non-uniform contraction rate. Theorem~\ref{thm:distributional-bellman}
as stated does neither; it records the dynamic recursion in the form used
downstream and does not claim a contraction phenomenon attributable to
turnover frictions per se.
\end{remark}

\subsection{Parameterizations}

Three standard parameterizations of $Z_\theta(s,w)$ are used in practice and
carry over unchanged to the portfolio MDP: the categorical C51 representation
on a fixed support $\{z_k\}\subset[V_{\min},V_{\max}]$ with grid projection
after each Bellman update \citep{bellemare2017distributional}; the quantile
QR-DQN representation at levels $\tau_k=(k-1/2)/N$ trained by the
quantile-regression loss $L^\tau(u)=(\tau-\mathbf 1\{u<0\})u$
\citep{dabney2018distributional}; and the implicit IQN representation
$q_\theta(s,w,\tau)$, $\tau\sim U(0,1)$ \citep{dabney2018implicit}. We use these
as black-box function classes; the contraction theory below is independent of
the choice.

\subsection{Risk-sensitive distributional control}

The distributional Bellman recursion of Section~9.1 is risk-neutral: the
value of $Z^\pi(s, w)$ is summarized by its mean. Risk-sensitive control
replaces the mean by a law-invariant scalar functional of $Z^\pi$. Two
formulations coexist in the literature: the \emph{static} formulation,
which applies a scalar risk measure to the total-return law $Z^\pi(s_0, w_0)$
at time zero, and the \emph{dynamic} formulation, which uses a nested
sequence of one-step risk measures and is time-consistent.

\begin{definition}[Static risk-sensitive distributional policy]
\label{def:static-risk-sensitive-policy}
Let $Z^\pi(s, a)$ denote the return distribution under policy $\pi$ and let
$A_0 \sim \pi(\cdot \mid s_0)$. For a law-invariant risk measure $\rho$ on
losses, define
\[
J_\rho(\pi) = \E[Z^\pi(s_0, A_0)] - \lambda \rho(-Z^\pi(s_0, A_0)),
\qquad \lambda \geq 0,
\]
and $\pi^\star_\rho \in \argmax_\pi J_\rho(\pi)$.
\end{definition}

The static objective is not in general the value function of a Markov
decision process, because most law-invariant risk measures (CVaR, EVaR,
spectral risk) lack the recursion-compatible decomposition that makes the
Bellman equation work. The dynamic counterpart fixes this by nesting a
one-step risk measure $\rho^{1}$ at each transition, in the sense of
\citet{ruszczynski2010risk}.

\begin{definition}[One-step Markov risk functional with $W_1$-Lipschitz modulus]
\label{def:one-step-markov-risk}
A one-step Markov risk functional is a family
$\{\rho^1_{s,a}\}_{(s,a)\in S\times A}$ of law-invariant convex risk functionals on
$\mathcal P_1(\mathbb R)$ such that $(s,a)\mapsto \rho^1_{s,a}$ is Borel measurable
and there exist finite constants $L_\rho<\infty$ and $M_\rho<\infty$ satisfying
\[
\big|\rho^1_{s,a}(\mu)-\rho^1_{s,a}(\nu)\big|
\le
L_\rho W_1(\mu,\nu),
\qquad
\mu,\nu\in\mathcal P_1(\mathbb R),
\]
and
\[
\sup_{(s,a)\in S\times A} |\rho^1_{s,a}(\delta_0)| \le M_\rho,
\]
uniformly over $(s,a)\in S\times A$.
\end{definition}

The constant $L_\rho$ is not normalized to be smaller than one, and
$M_\rho$ is a bounded-baseline condition. The contraction condition below is
imposed directly on the combined modulus $\gamma+\lambda L_\rho$; the
baseline constant is used only to ensure that the operator maps
$\mathcal Z_p$ into itself.

\begin{definition}[Risk-shifted distributional Bellman operator]
\label{def:risk-sensitive-bellman}
For a one-step Markov risk functional $\rho^1$ with finite $W_1$-Lipschitz modulus
$L_\rho$, and for a fixed policy $\pi$, define $T^\pi_\rho:\mathcal Z_p\to\mathcal Z_p$
by declaring $(T^\pi_\rho Z)(s,w)$ to be the law of
\[
g(s,w,s',w')
-
\lambda\,\rho^1_{s',w'}\!\left(-Z(s',w')\right)
+
\gamma G,
\]
where $s'\sim P(\cdot\mid s,w)$, $w'\sim\pi(\cdot\mid s')$, and
$G\mid(s',w')\sim Z(s',w')$, conditionally independently of the transition--policy draw
given $(s',w')$.
\end{definition}

\begin{remark}[Interpretation of the risk-shifted operator]
The operator $T^\pi_\rho$ should be interpreted as a risk-shifted distributional Bellman
operator. It is not the most general time-consistent Markov risk recursion of
\citet{ruszczynski2010risk}, because the continuation distribution is still propagated explicitly
through the random draw $G$. The theorem below proves that adding a one-step
law-invariant risk adjustment preserves contraction provided
$\gamma+\lambda L_\rho<1$. Thus the result is a sufficient contraction condition for a
distributional risk-shifted recursion, not a complete dynamic-risk representation theorem.
\end{remark}

\begin{theorem}[Risk-shifted distributional Bellman contraction]
\label{thm:risk-sensitive-bellman-contraction}
Assume $|g(s,w,s',w')|\le R_{\max}<\infty$ for all admissible
$(s,w,s',w')$. Let $\rho^1$ be a one-step Markov risk functional with finite
$W_1$-Lipschitz modulus $L_\rho$ and finite baseline $M_\rho$. If
\[
\kappa:=\gamma+\lambda L_\rho<1,
\]
then $T^\pi_\rho$ maps $\mathcal Z_p$ into itself and is a $\kappa$-contraction in the
supremal Wasserstein metric:
\[
\overline W_p(T^\pi_\rho Z,T^\pi_\rho Z')
\le
\kappa\,\overline W_p(Z,Z'),
\qquad
Z,Z'\in\mathcal Z_p .
\]
Consequently, $T^\pi_\rho$ admits a unique fixed point $Z^\pi_\rho\in\mathcal Z_p$, and
the iterates $Z^{k+1}=T^\pi_\rho Z^k$ converge to $Z^\pi_\rho$ in $\overline W_p$ from
any $Z^0\in\mathcal Z_p$.
\end{theorem}

\begin{proof}
First, $T^\pi_\rho Z\in\mathcal Z_p$. Indeed, for $Z\in\mathcal Z_p$,
\[
|\rho^1_{s',w'}(-Z(s',w'))|
\le
M_\rho
+
L_\rho W_1(Z(s',w'),\delta_0),
\]
and $W_1(Z(s',w'),\delta_0)\le \left(E|G|^p\right)^{1/p}$. Together with the boundedness
of $g$ and the definition of $\mathcal Z_p$, this gives a uniform $p$-moment bound.

Now fix $(s,w)$. Couple the transition--policy draw $(s',w')$ identically in both
arguments. Conditional on $(s',w')$, let $(G,G')$ be an optimal $W_p$-coupling of
$Z(s',w')$ and $Z'(s',w')$ (Borel-measurable selector by
\citet[Cor.~5.22]{villani2009optimal}). Define
\[
X
:=
g(s,w,s',w')
-
\lambda\rho^1_{s',w'}(-Z(s',w'))
+
\gamma G,
\]
and
\[
X'
:=
g(s,w,s',w')
-
\lambda\rho^1_{s',w'}(-Z'(s',w'))
+
\gamma G'.
\]
Then $X\sim(T^\pi_\rho Z)(s,w)$ and
$X'\sim(T^\pi_\rho Z')(s,w)$. The reward term cancels. Conditional on $(s',w')$,
Minkowski's inequality gives
\[
\begin{aligned}
\left(E\!\left[|X-X'|^p\mid s',w'\right]\right)^{1/p}
&\le
\lambda
\left|
\rho^1_{s',w'}(-Z(s',w'))
-
\rho^1_{s',w'}(-Z'(s',w'))
\right|        \\
&\quad+
\gamma
\left(E\!\left[|G-G'|^p\mid s',w'\right]\right)^{1/p}.
\end{aligned}
\]
By the $W_1$-Lipschitz property of $\rho^1$, the invariance of Wasserstein distance under
$x\mapsto -x$, and $W_1\le W_p$,
\[
\left|
\rho^1_{s',w'}(-Z(s',w'))
-
\rho^1_{s',w'}(-Z'(s',w'))
\right|
\le
L_\rho W_p(Z(s',w'),Z'(s',w')).
\]
The optimal coupling gives
\[
\left(E[|G-G'|^p\mid s',w']\right)^{1/p}
=
W_p(Z(s',w'),Z'(s',w')).
\]
Hence
\[
\left(E[|X-X'|^p\mid s',w']\right)^{1/p}
\le
(\lambda L_\rho+\gamma)
W_p(Z(s',w'),Z'(s',w')).
\]
Taking expectations over $(s',w')$, then the $p$-th root, yields
\[
W_p\!\left((T^\pi_\rho Z)(s,w),(T^\pi_\rho Z')(s,w)\right)
\le
(\lambda L_\rho+\gamma)\overline W_p(Z,Z').
\]
Taking the supremum over $(s,w)$ gives the claimed contraction. The Banach fixed-point
theorem on the complete metric space $(\mathcal Z_p,\overline W_p)$ gives existence,
uniqueness, and convergence of iterates.
\end{proof}

\begin{corollary}[Risk-neutral limit]
\label{cor:risk-neutral-limit}
At $\lambda=0$, Theorem~\ref{thm:risk-sensitive-bellman-contraction} reduces to Theorem~\ref{thm:distributional-bellman}, the standard distributional
Bellman $\gamma$-contraction. For $\lambda>0$, the risk shift increases the contraction
modulus from $\gamma$ to $\gamma+\lambda L_\rho$. The recursion remains contractive
precisely when the risk-aversion weight is small enough that
$\gamma+\lambda L_\rho<1$.
\end{corollary}

\begin{remark}[Operational $L_\rho$ for canonical one-step risk measures]
\label{rmk:operational-L-rho}
The contraction condition $\gamma+\lambda L_\rho<1$ is mild or restrictive depending on
the $W_1$-Lipschitz modulus of the chosen one-step risk functional.

\begin{itemize}
\item \textbf{CVaR.} The dual representation of $\operatorname{CVaR}_\alpha$ gives
$L_{\operatorname{CVaR}_\alpha}=1/\alpha$. At $\alpha=0.05$, $L_\rho=20$, so with
$\gamma=0.95$ the contraction condition requires
$\lambda<(1-\gamma)/L_\rho=0.0025$. Thus unscaled CVaR can be used, but the
risk-aversion weight must be correspondingly small.

\item \textbf{Entropic risk.} For
\[
\rho_\eta(L)=\eta\log E[e^{L/\eta}],
\]
global $W_1$-Lipschitzness is not automatic on unbounded supports. On a bounded
interval $[a,b]$, the elementary bound
\[
|\rho_\eta(\mu)-\rho_\eta(\nu)|
\le
\exp\!\left(\frac{b-a}{\eta}\right) W_1(\mu,\nu)
\]
shows that a finite $W_1$-Lipschitz constant exists, but it depends on the support
diameter and on $\eta$, and it need not be smaller than one. Entropic risk is
one-Lipschitz in stronger metrics such as $W_\infty$, not uniformly one-Lipschitz in
$W_1$.

\item \textbf{Scalar Wasserstein-DRO risk.} For the scalar identity-loss functional
\[
\rho_\varepsilon(\mu)
:=
\sup_{\nu:W_1(\nu,\mu)\le\varepsilon}E_\nu[X],
\]
one has $\rho_\varepsilon(\mu)=E_\mu[X]+\varepsilon$, hence $L_\rho=1$. For nonlinear
losses, the corresponding Lipschitz modulus is inherited from the loss.
\end{itemize}

Thus Theorem~\ref{thm:risk-sensitive-bellman-contraction} is most useful when the one-step risk functional has a controlled
$W_1$-modulus or when the risk-aversion parameter $\lambda$ is scaled to absorb that
modulus. The natural headline application is therefore the scalar
Wasserstein-DRO risk functional, for which $L_\rho=1$ and the contraction
condition $\gamma+\lambda L_\rho<1$ holds for any $\lambda<1-\gamma$ without
rescaling; the unscaled-CVaR case ($L_\rho=1/\alpha$) is the restrictive
extreme that forces a small $\lambda$, not the intended use.
\end{remark}

\begin{proposition}[CVaR--Bellman incoherence in the static formulation]
\label{prop:cvar-bellman-incoherence}
When the one-period reward $R$ is random and jointly distributed with the
continuation return $Z(s', w')$, static CVaR does not decompose recursively:
\[
\CVaR_\alpha(R + \gamma Z(s', w')) \neq R + \gamma \CVaR_\alpha(Z(s', w'))
\]
in general. Equality holds only in special cases, for example when $R$ is
deterministic conditional on the current state-action pair (so the
translation invariance of CVaR applies). The static formulation of
Definition~\ref{def:static-risk-sensitive-policy} therefore does not admit
a Bellman recursion with $\rho = \CVaR_\alpha$; the dynamic formulation of
Definition~\ref{def:risk-sensitive-bellman}, which nests
$\rho^{1}_{s', w'} = \CVaR_\alpha$ at each step, does, by
Theorem~\ref{thm:risk-sensitive-bellman-contraction}.
\end{proposition}

\begin{remark}[Dynamic-CVaR and Markov risk measures]
The time-inconsistency of static CVaR is formalized in a substantial
literature on Markov risk measures and time-consistent dynamic programming.
\citet{ruszczynski2010risk} introduces nested Markov risk measures.
\citet{bauerle2011markov} show that the static $\CVaR_\tau$ objective in an
MDP can be reduced to an ordinary MDP with an extended state space.
\citet{pflug2014multistage} provide a book-length treatment.
Theorem~\ref{thm:risk-sensitive-bellman-contraction} above is the explicit
contraction-based version of these nesting constructions in the
distributional-return setting.
\end{remark}

%==============================================================================
% SECTION 10: SYNTHESIS
%==============================================================================
\section{Synthesis: domination and second-order conservatism}\label{sec:synthesis}

Bayesian DPO and Wasserstein DRO are often presented as different
philosophies. We show they are tightly connected via a domination bound
(Theorem~\ref{thm:bayes-dro-domination}) and a calculation in a tractable
special case (Proposition~\ref{prop:gaussian-isotropic-second-order}) that
quantifies how loose the first-order bound can be. The label
\emph{domination} (rather than duality) is intentional:
Theorem~\ref{thm:bayes-dro-domination} bounds Bayesian uncertainty by
Wasserstein DRO from above, but the two are not isomorphic in general.

\subsection{Domination}

\begin{theorem}[Bayesian--DRO domination: model-level and loss-level forms]
\label{thm:bayes-dro-domination}
Let $\Rpredn := \int_\Theta R_\theta \, dP_n(\theta)$ and
$\hat\theta := \E_{P_n}[\theta]$. Let $D \subseteq \Theta$ be a measurable
set with $\hat\theta \in D$, and suppose that the model map
$\theta \mapsto R_\theta$ is $L$-Lipschitz in $W_1$ on $D$:
\begin{equation}
W_1(R_\theta, R_{\theta'}) \leq L \|\theta - \theta'\|, \quad \theta, \theta' \in D.
\label{eq:model-lipschitz}
\end{equation}
Define the model-level tail term
$\tau_n(D) := \int_{D^c} W_1(R_\theta, R_{\hat\theta}) \, dP_n(\theta)$,
allowing $\tau_n(D) = +\infty$.

\textbf{(Model-level form.)} If $\tau_n(D) < \infty$, then
\begin{equation}
W_1(\Rpredn, R_{\hat\theta})
\leq L \E_{P_n}[\|\theta - \hat\theta\| \mathbf{1}_D(\theta)] + \tau_n(D).
\label{eq:model-domination}
\end{equation}
Consequently, defining
$\varepsilon_n^W(D) := L \E_{P_n}[\|\theta - \hat\theta\| \mathbf{1}_D(\theta)] + \tau_n(D)$,
$\Rpredn \in B_{W_1}(R_{\hat\theta}, \varepsilon_n^W(D))$,
and for every loss $\ell(w, \cdot)$ that is Lipschitz in $r$ with constant
$\Lip_r(\ell)$,
$|\E_{\Rpredn}[\ell(w, r)] - \E_{R_{\hat\theta}}[\ell(w, r)]| \leq \Lip_r(\ell) \varepsilon_n^W(D)$.

\textbf{(Loss-level form.)} If the loss instead satisfies the uniform
bounded-expectation condition
\begin{equation}
B_\ell := \sup_{w \in \W, \theta \in \Theta} |\E_{R_\theta}[\ell(w, r)]| < \infty,
\label{eq:loss-bound}
\end{equation}
then
\begin{equation}
|\E_{\Rpredn}[\ell(w, r)] - \E_{R_{\hat\theta}}[\ell(w, r)]|
\leq \Lip_r(\ell) L \E_{P_n}[\|\theta - \hat\theta\| \mathbf{1}_D(\theta)]
+ 2 B_\ell P_n(D^c).
\label{eq:loss-domination}
\end{equation}

Under Assumption~\ref{ass:posterior-concentration}, any sequence
$D_n \uparrow \Theta$ with $P_n(D_n^c) = o(n^{-1/2})$ and $L_n$ bounded
yields, in either form, an effective radius of order $O(n^{-1/2})$.
\end{theorem}

\begin{proof}
By convexity of $W_1$ under mixtures and Kantorovich--Rubinstein,
$W_1(\int R_\theta \, dP_n, R_{\hat\theta}) \leq \int W_1(R_\theta, R_{\hat\theta}) \, dP_n$.
Splitting over $D$ and $D^c$ and applying the local Lipschitz hypothesis
gives \eqref{eq:model-domination}. For the loss-level version, write
$\Delta \ell(w) := |\E_{\Rpredn}[\ell] - \E_{R_{\hat\theta}}[\ell]|
\leq \int_D |\E_{R_\theta}[\ell] - \E_{R_{\hat\theta}}[\ell]| dP_n
+ \int_{D^c} |\E_{R_\theta}[\ell] - \E_{R_{\hat\theta}}[\ell]| dP_n$.
The first integral is bounded by the local Lipschitz argument; the second
by $2 B_\ell P_n(D^c)$.
\end{proof}

\begin{remark}[Why the local form is the operative one]
The model map $\theta \mapsto R_\theta$ is rarely globally Lipschitz on the
full parameter space $\Theta$. The Gaussian--NIW case of
Proposition~\ref{prop:gaussian-niw-lipschitz} is typical: the local
Lipschitz constant $L_{\mathrm{loc}} = \max(1, (2\sqrt{\delta})^{-1})$ blows
up as the lower bound $\delta$ on $\lambda_{\min}(\Sigma)$ approaches zero,
so a global Lipschitz constant does not exist.
\end{remark}

\begin{corollary}[Specialization to Gaussian--NIW; bounded or truncated losses]
\label{cor:gaussian-niw}
Let $R_{(\mu,\Sigma)} = \N(\mu, \Sigma)$ and let $P_n$ be the NIW posterior
with posterior mean $\hat\theta_n = (\hat\mu_n, \hat\Sigma_n)$. Assume
$\Sigma^\star := \mathrm{plim}_{n \to \infty} \hat\Sigma_n \succ 0$. Suppose
either that the loss is bounded in the sense of \eqref{eq:loss-bound}, or
that the loss is replaced by a truncation sequence. Take
$\delta_n = \tfrac{1}{2} \lambda_{\min}(\hat\Sigma_n)$,
$\Delta_n = 2 \lambda_{\max}(\hat\Sigma_n)$,
$D_n = \{(\mu, \Sigma): \delta_n I \preceq \Sigma \preceq \Delta_n I\}$.
Then the local Lipschitz constant on $D_n$ is
$L_n = \max(1, (2 \sqrt{\delta_n})^{-1})$, converging in probability to a
finite limit. Standard NIW posterior concentration gives
$P_n(D_n^c) = O(e^{-cn})$ for some $c > 0$, and the corresponding loss-level
radius is $\varepsilon_n = O(L_n n^{-1/2})$.
\end{corollary}

% PATCH 12: K-DEPENDENCE NOTE
\begin{remark}[Explicit dimension dependence]\label{rmk:k-dependence}
Corollary~\ref{cor:gaussian-niw} states $P_n(D_n^c) = O(e^{-cn})$ with the
constant $c > 0$ depending on Wishart concentration. For fixed $K$, the
bound is uniform; for $K = K_n$ growing with $n$, the standard random-matrix
tail \citep[Theorem 6.1]{wainwright2019high} gives
\[
\Pr\!\left(\lambda_{\min}(\Sigma) < \tfrac{1}{2} \lambda_{\min}(\Sigma^\star)\right)
\leq 2 \exp\!\left(-c(K_n, n) \cdot n\right),
\quad c(K_n, n) \asymp (1 - \sqrt{K_n / n})_+^2,
\]
which requires $K_n / n \to 0$ for the local-Lipschitz domain $D_n$ to
retain exponential tail mass. In the high-dimensional regime
$K_n / n \to \kappa \in (0, 1)$, the local-Lipschitz constant $L_n$ does
not stabilize and the Bayesian--DRO radius
$\varepsilon_n = L_n \cdot \E_{P_n}\|\theta - \hat\theta\|$ inherits a
$K_n$-dependent inflation. This is the analytic shadow of the empirical
observation in Section~\ref{sec:empirical-equity-backtest} that the
$K = 25$, $n = 60$ regime is dominated by covariance-estimation noise.
\end{remark}

\subsection{Posterior credible sets induce Wasserstein ambiguity sets}

\begin{theorem}[Posterior credible radius, finite-sample form]
\label{thm:posterior-credible-radius}
Let $D_n \subseteq \Theta$ be a measurable Lipschitz domain on which the
model map satisfies $W_1(R_\theta, R_{\hat\theta}) \leq L \|\theta - \hat\theta\|$
for all $\theta \in D_n$, with $\hat\theta \in D_n$ and
$P_n(D_n^c) \leq \beta_n$ for some $\beta_n \in [0, 1)$. For $\delta \in (0, 1)$
let $q_{1-\delta}$ be the posterior $(1 - \delta)$-quantile of
$\|\theta - \hat\theta\|$, and define the credible radius
\begin{equation}
\varepsilon_{n,\delta} = L \cdot q_{1-\delta}.
\label{eq:credible-radius}
\end{equation}
Then
\[
P_n\!\left(R_\theta \in B_{W_1}(R_{\hat\theta}, \varepsilon_{n,\delta})\right)
\;\geq\; 1 - \delta - \beta_n.
\]
\end{theorem}

\begin{proof}
On the event
$E := D_n \cap \{\|\theta - \hat\theta\| \leq q_{1-\delta}\}$, the Lipschitz
hypothesis on $D_n$ gives $W_1(R_\theta, R_{\hat\theta}) \leq L \cdot q_{1-\delta}
= \varepsilon_{n,\delta}$, so $R_\theta \in B_{W_1}(R_{\hat\theta}, \varepsilon_{n,\delta})$
on $E$. The union bound gives $P_n(E) \geq 1 - \delta - \beta_n$.
\end{proof}

\begin{corollary}[Deterministic-inclusion form]
\label{cor:credible-radius-deterministic}
Under the hypotheses of Theorem~\ref{thm:posterior-credible-radius}, if in
addition $P_n(\{\|\theta - \hat\theta\| \leq q_{1-\delta}\} \setminus D_n) = 0$
(so $\beta_n$ may be taken to be $0$ for the displayed conclusion), then
$P_n(R_\theta \in B_{W_1}(R_{\hat\theta}, \varepsilon_{n,\delta})) \geq 1 - \delta$.
\end{corollary}

\begin{remark}[What the credible radius certifies, and what it does not]
\label{rmk:calibration-not-coverage}
The credible radius \eqref{eq:credible-radius} is $O(n^{-1/2})$ under
Assumption~\ref{ass:posterior-concentration}, whereas the
distribution-free coverage radius behind
Theorem~\ref{thm:dro-finite-sample} scales as
$(c_1\log(c_2/\delta)/n)^{1/\max(K,2)}$, which at $K=25$ is effectively
$O(1)$ for any realistic $n$. The two radii are therefore not comparable in
magnitude, and the credible-radius rule should be read as a \emph{posterior
calibration} of the ambiguity ball---it matches validation tuning without
spending validation data---\emph{not} as a finite-sample
distributional-robustness certificate in the sense of
Theorem~\ref{thm:dro-finite-sample}. Worst-case coverage in the
Mohajerin Esfahani--Kuhn sense would require the far larger
Fournier--Guillin radius; the operational value documented in
Section~\ref{sec:numerical} is calibration-without-validation and turnover,
not coverage.
\end{remark}

\begin{remark}[Discharging $\beta_n$ in Gaussian--NIW]
\label{rmk:finite-sample-credible}
The credible ball $\{\|\theta - \hat\theta\| \leq q_{1-\delta}\}$ is rarely
contained in any fixed Lipschitz domain with full posterior probability,
because $\lambda_{\min}(\hat\Sigma_n)$ can be small in finite samples even
when $\lambda_{\min}(\Sigma^\star) > 0$. Theorem~\ref{thm:posterior-credible-radius}
handles this by retaining a probabilistic correction $\beta_n$. For
Gaussian--NIW, $\beta_n = O(e^{-cn})$ for fixed $K$, so the correction is
exponentially small; for $K = K_n$ growing with $n$, $\beta_n$ depends on
$K_n / n$ explicitly via Remark~\ref{rmk:k-dependence}.
\end{remark}

\begin{proposition}[Predictive mixing generally differs from plug-in]
\label{prop:pred-vs-plugin}
Let $\Rpred = \int_\Theta R_\theta \, dP(\theta)$ and
$\hat\theta = \int_\Theta \theta \, dP(\theta)$. In general,
$\Rpred \neq R_{\hat\theta}$, and Bayesian predictive optimization and
plug-in optimization are generally distinct.
\end{proposition}

\begin{example}[Gaussian location model]
For $R_\mu = \N(\mu, \Sigma)$ with $\mu \sim \N(\hat\mu, \Lambda)$,
$\Rpred = \N(\hat\mu, \Sigma + \Lambda)$ whereas $R_{\hat\mu} = \N(\hat\mu, \Sigma)$.
\end{example}

\subsection{Gaussian-isotropic case: posterior mixing is second order}

\begin{proposition}[Gaussian-isotropic posterior mixing is second order]
\label{prop:gaussian-isotropic-second-order}
Let $R_\mu = \N(\mu, \sigma^2 I_K)$ with known $\sigma > 0$, and suppose
the posterior on $\mu$ is $P_n = \N(\hat\mu, \Lambda_n)$ with $\Lambda_n \succ 0$.
Then $\Rpredn = \N(\hat\mu, \sigma^2 I_K + \Lambda_n)$ and
\begin{align}
W_1(\Rpredn, R_{\hat\mu}) &\leq W_2(\Rpredn, R_{\hat\mu})
\leq \frac{\|\Lambda_n\|_F}{2\sigma}, \label{eq:gaussian-upper} \\
W_1(\Rpredn, R_{\hat\mu}) &\geq \sqrt{2/\pi} \cdot
\left(\sqrt{\sigma^2 + \lambda_{\max}(\Lambda_n)} - \sigma\right). \label{eq:gaussian-lower}
\end{align}
In particular, when $\Lambda_n = O(n^{-1}) I_K$ and $\sigma$ is fixed,
$W_1(\Rpredn, R_{\hat\mu}) = O(n^{-1})$, whereas
$\E_{P_n} \|\mu - \hat\mu\| = O(n^{-1/2})$. The first-order Bayesian--DRO
radius is conservative by one order in this symmetric Gaussian location model.
\end{proposition}

\begin{proof}
Gaussian convolution gives the predictive identity. Both measures are
Gaussian with the same mean $\hat\mu$, so the Bures--Wasserstein formula
yields $W_2^2 = \sum_{i=1}^K (\sqrt{\sigma^2 + \lambda_i} - \sigma)^2$. The
inequality $\sqrt{\sigma^2 + \lambda} - \sigma = \lambda / (\sqrt{\sigma^2 + \lambda} + \sigma)
\leq \lambda / (2\sigma)$ gives \eqref{eq:gaussian-upper}. For the lower
bound, let $v$ be a unit top eigenvector of $\Lambda_n$ and consider the
$1$-Lipschitz projection $\phi(x)=v^\top(x-\hat\mu)$. Wasserstein distance
contracts under $1$-Lipschitz maps in the forward direction, hence
\[
W_1(\Rpredn,R_{\hat\mu})
\ge
W_1(\phi_\#\Rpredn,\phi_\#R_{\hat\mu}).
\]
The two projected laws are
$\phi_\#\Rpredn=\N(0,\sigma^2+\lambda_{\max}(\Lambda_n))$ and
$\phi_\#R_{\hat\mu}=\N(0,\sigma^2)$. In one dimension, the monotone optimal
coupling gives
$W_1(\N(0,a^2),\N(0,b^2))=|a-b|\E|Z|=|a-b|\sqrt{2/\pi}$ with
$Z\sim\N(0,1)$. Substituting
$a=\sqrt{\sigma^2+\lambda_{\max}(\Lambda_n)}$ and $b=\sigma$ gives
\eqref{eq:gaussian-lower}.
\end{proof}

\begin{remark}[Why first-order bounds overshoot]
For any $1$-Lipschitz $f: \R^K \to \R$ in the Gaussian-mean model, the
function $\mu \mapsto \E_{R_\mu}[f] = \E_Z[f(\mu + \sigma Z)]$ is
convolution-smoothed and infinitely differentiable. Its first-order Taylor
expansion around $\hat\mu$ gives a linear-in-$(\mu - \hat\mu)$ contribution
that integrates to zero against the centered $P_n$. The leading term is
therefore second-order, of size $O(\E_{P_n} \|\mu - \hat\mu\|^2) = O(\tr \Lambda_n)$,
matching $\|\Lambda_n\|_F$ up to dimensional constants. The first-order
Lipschitz upper bound misses this cancellation.
\end{remark}

\begin{proposition}[Second-order Bayes-DRO radius in the Gaussian-isotropic model]
\label{prop:second-order-radius}
Under the hypotheses of Proposition~\ref{prop:gaussian-isotropic-second-order},
define the second-order Bayes-DRO radius
\begin{equation}
\varepsilon_n^{(2)} := \frac{\|\Lambda_n\|_F}{2\sigma}.
\label{eq:second-order-radius}
\end{equation}
Then:
(i) $\Rpredn \in B_{W_1}(R_{\hat\mu}, \varepsilon_n^{(2)})$;
(ii) for every $1$-Lipschitz loss $\ell$,
$|\E_{\Rpredn}[\ell] - \E_{R_{\hat\mu}}[\ell]| \leq \varepsilon_n^{(2)}$;
(iii) if $\Lambda_n = O(n^{-1}) I_K$ and $\sigma > 0$ is fixed, then
$\varepsilon_n^{(2)} = O(n^{-1})$, compared with $O(n^{-1/2})$ for the
first-order radius;
(iv) the second-order radius is fully data-driven.
\end{proposition}

\subsection{When is the first-order Bayesian--DRO bound tight? A taxonomy}\label{sec:taxonomy-first-order}

\begin{example}[Deterministic affine law: first-order tight]\label{ex:affine-deterministic}
Let $R_\theta = \delta_{a + b\theta}$. Then
$W_1(\Rpredn, R_{\hat\theta}) = \E_{P_n} \|b(\theta - \hat\theta)\|$,
achieved by $f(x) = \|x - (a + b\hat\theta)\|$ in the KR dual. The
first-order radius is order-sharp.
\end{example}

\begin{example}[Gaussian-mean location: second-order conservative]
\label{ex:gaussian-mean}
By Proposition~\ref{prop:gaussian-isotropic-second-order}, the actual
$W_1$ gap is $O(n^{-1})$ while the first-order radius is $O(n^{-1/2})$.
\end{example}

\begin{example}[Gaussian-scale: centered interior scale is smooth]
For $R_\sigma = \N(0, \sigma^2 I_K)$ with $\sigma$ random, when the
posterior is concentrated away from $\sigma = 0$, the actual
predictive-to-plug-in gap is second-order.
\end{example}

\begin{example}[Kinked or noncentered models: first-order behavior reappears]\label{ex:kinked}
Let $\theta$ have posterior mean $0$ and variance $O(n^{-1})$, and set
$R_\theta = \delta_{|\theta| e_1}$. Then $W_1(\Rpredn, R_{\hat\theta}) = \E_{P_n} |\theta| = O(n^{-1/2})$,
order-sharp despite posterior centering.
\end{example}

% PATCH 10: NEW PROPOSITION (FIRST-ORDER TIGHTNESS)
\begin{proposition}[First-order tightness for affine-deterministic and locally non-smooth models]
\label{prop:first-order-tight}
Assume one of the following two model structures.
\begin{enumerate}[label=\textup{(\roman*)}]
\item \textbf{Deterministic affine.} $R_\theta = \delta_{a + B\theta}$ with
$a \in \R^K$ and $B \in \R^{K \times d_\theta}$ deterministic. Let
$\hat\theta = \E_{P_n} \theta$. If the smallest singular value of $B$
restricted to the linear span of $\supp(P_n) - \{\hat\theta\}$ is bounded
below by $\sigma_{\min}(B) > 0$, then
\[
W_1(\Rpredn, R_{\hat\theta}) \geq \sigma_{\min}(B) \cdot \E_{P_n}\|\theta - \hat\theta\|.
\]
\item \textbf{Locally non-smooth model map.} Suppose $\hat\theta = \E_{P_n} \theta$
and there exist a $1$-Lipschitz test $f: \R^K \to \R$, a direction
$v \in \R^{d_\theta}$ with $\|v\| = 1$, and constants $c_+ \neq c_-$ such
that, \emph{uniformly on a posterior-typical neighbourhood}
$U_n := \{\theta : \|\theta - \hat\theta\| \leq r_n\}$,
\[
h_f(\theta) - h_f(\hat\theta)
= c_+ \bigl(v^\top (\theta - \hat\theta)\bigr)_+
- c_- \bigl(v^\top (\theta - \hat\theta)\bigr)_-
+ r(\theta), \qquad \theta \in U_n,
\]
with remainder
$\sup_{\theta \in U_n} |r(\theta)| / \|\theta - \hat\theta\|
= o(1)$
and posterior-tail control
$P_n(U_n^c) \cdot \sup_{\theta \in U_n^c} |h_f(\theta) - h_f(\hat\theta)|
= o\bigl(\E_{P_n} |v^\top (\theta - \hat\theta)|\bigr)$.
The uniformity over the full neighbourhood $U_n$ (rather than only along the
line $\hat\theta + tv$) ensures the orthogonal component $\theta^\perp$ does
not break the expansion when $\theta \in U_n$ has $v^\top(\theta - \hat\theta) \neq 0$
but $\theta^\perp \neq 0$.
Then
\[
W_1(\Rpredn, R_{\hat\theta})
\geq \tfrac{1}{2} |c_+ - c_-| \cdot \E_{P_n}|v^\top(\theta - \hat\theta)|
- o(\E_{P_n}|v^\top(\theta - \hat\theta)|).
\]
\end{enumerate}
Consequently, in either class, the first-order Bayesian--DRO radius
$\varepsilon_n^W$ of Theorem~\ref{thm:bayes-dro-domination} is order-sharp:
there exist constants $0 < c \leq L < \infty$ such that
$c \cdot \E_{P_n}\|\theta - \hat\theta\| \leq W_1(\Rpredn, R_{\hat\theta})
\leq L \cdot \E_{P_n}\|\theta - \hat\theta\|$,
and the upper bound of Theorem~\ref{thm:bayes-dro-domination} cannot be
improved by more than a constant factor.
\end{proposition}

\begin{proof}
\emph{Case (i).} The test function $f(x) := \|x - a - B\hat\theta\|$ is
$1$-Lipschitz. Under $R_\theta = \delta_{a + B\theta}$,
$\int f \, dR_\theta = \|B(\theta - \hat\theta)\|$;
$\int f \, dR_{\hat\theta} = 0$ and
$\int f \, d\Rpredn = \E_{P_n} \|B(\theta - \hat\theta)\|
\geq \sigma_{\min}(B) \cdot \E_{P_n}\|\theta - \hat\theta\|$
by the singular-value bound. Kantorovich--Rubinstein gives the claim.

\emph{Case (ii).} Decompose $\theta - \hat\theta = (v^\top(\theta - \hat\theta)) v + \theta^\perp$
with $\theta^\perp \perp v$. On the posterior-typical neighbourhood, the
directional expansion gives
$h_f(\theta) - h_f(\hat\theta) = c_+ (v^\top(\theta - \hat\theta))_+
- c_- (v^\top(\theta - \hat\theta))_- + o(|v^\top(\theta - \hat\theta)|)$
uniformly. Using $\E_{P_n}[v^\top(\theta - \hat\theta)] = 0$ and the
identity $\E[(X)_+] = \E[(X)_-] = \tfrac{1}{2} \E|X|$ for centered $X$,
\[
\E_{P_n}[h_f(\theta) - h_f(\hat\theta)]
= \tfrac{1}{2}(c_+ - c_-) \cdot \E_{P_n}|v^\top(\theta - \hat\theta)|
+ o(\E_{P_n}|v^\top(\theta - \hat\theta)|).
\]
Replacing $f$ by $-f$ if $c_+ < c_-$ and applying KR yields the claim.
\end{proof}

\begin{remark}[Order-sharpness in the nonsmooth taxonomy]
\label{rmk:taxonomy-converse}
Proposition~\ref{prop:first-order-tight} substantiates the heuristic
statements of Examples~\ref{ex:affine-deterministic} and~\ref{ex:kinked}
with a quantitative lower bound. Combined with
Proposition~\ref{prop:gaussian-isotropic-second-order} (upper bound
$O(n^{-1})$ in the centered Gaussian-isotropic case) and the first-order
upper bound of Theorem~\ref{thm:bayes-dro-domination} ($O(n^{-1/2})$), the
picture is: \emph{smooth centered} models permit a second-order radius;
\emph{nonsmooth or noncentered} models do not, and the first-order radius
is then within a model-dependent constant factor of the true Wasserstein
gap. Theorem~\ref{thm:modulus-interpolation} below gives a sufficient
modulus-of-smoothness condition for intermediate rates parameterized by
$\omega$.
\end{remark}

% NEW THEOREM (v9): MODULUS-OF-SMOOTHNESS INTERPOLATION
\begin{theorem}[Second-order Bayes--DRO bound under a uniform modulus of smoothness]
\label{thm:modulus-interpolation}
Let $R_\theta$ be a parametric model with $\theta \in \Theta \subseteq \R^{d_\theta}$
and let $P_n$ be a posterior with $\hat\theta = \E_{P_n}[\theta]$. Let
$\mathcal F_1$ denote the class of $1$-Lipschitz $f: \R^K \to \R$ and set
$h_f(\theta) := \int f \, dR_\theta$. Assume:
\begin{enumerate}[label=\textup{(H\arabic*)}]
\item \textup{(Local linearization with bounded coefficient.)} There exists a
posterior-typical neighbourhood $N_n \subseteq \Theta$ with $\hat\theta \in N_n$
and, for each $f \in \mathcal F_1$, a vector $G_f \in \R^{d_\theta}$ such that
\[
h_f(\theta) - h_f(\hat\theta)
= \langle G_f, \theta - \hat\theta \rangle + \rho_f(\theta),
\qquad \theta \in N_n,
\]
with $\sup_{f \in \mathcal F_1} \|G_f\| \leq G < \infty$. The constant $G$ is
the local Lipschitz constant of $\theta \mapsto R_\theta$ at $\hat\theta$
in $W_1$.
\item \textup{(Uniform modulus.)} There exists a non-decreasing
$\omega: [0, \infty) \to [0, \infty)$ with $\omega(0) = 0$ such that
$\sup_{f \in \mathcal F_1} |\rho_f(\theta)| \leq \omega(\|\theta - \hat\theta\|)$
for all $\theta \in N_n$.
\item \textup{(Tail negligibility via integrable envelope.)} There exists a
measurable $\Phi: \Theta \to [0, \infty)$ with
$\sup_{f \in \mathcal F_1} |h_f(\theta) - h_f(\hat\theta)| \leq \Phi(\theta)$
for all $\theta \in \Theta$, and
\[
G \cdot \E_{P_n}\!\left[\|\theta - \hat\theta\| \,\mathbf 1_{N_n^c}\right]
+ \E_{P_n}\!\left[\Phi(\theta) \,\mathbf 1_{N_n^c}\right]
= o\!\left(\E_{P_n}\,\omega(\|\theta - \hat\theta\|)\right).
\]
\end{enumerate}
Then
\[
W_1(\Rpredn, R_{\hat\theta})
\;\leq\; \E_{P_n}\!\left[\omega(\|\theta - \hat\theta\|)\right]
\;+\; o\!\left(\E_{P_n}\,\omega(\|\theta - \hat\theta\|)\right).
\]
\end{theorem}

\begin{remark}[On the envelope $\Phi$]
\textup{(H3)} is the appropriate replacement for a deterministic
``diameter of supports'' bound, which would fail for Gaussian or Student-$t$
models with unbounded support. In the Gaussian--NIW case, one may take
$\Phi(\theta) = G \cdot \|\theta - \hat\theta\| + W_1(R_\theta, R_{\hat\theta})$;
both terms have finite posterior moments and exponentially small tails by
posterior concentration, so $\E_{P_n}[\Phi(\theta) \,\mathbf 1_{N_n^c}]$ is
$O(e^{-cn})$ when $N_n$ is taken as a Gaussian credible ball of radius
$O(n^{-1/2 + \delta})$.
\end{remark}

\begin{proof}
By Kantorovich--Rubinstein duality,
$W_1(\Rpredn, R_{\hat\theta}) = \sup_{f \in \mathcal F_1} \E_{P_n}[h_f(\theta) - h_f(\hat\theta)]$.
Fix $f \in \mathcal F_1$ and split the expectation:
\[
\E_{P_n}[h_f(\theta) - h_f(\hat\theta)]
= \underbrace{\E_{P_n}\!\left[\langle G_f, \theta - \hat\theta \rangle \mathbf 1_{N_n}\right]}_{=:\, T_1}
+ \underbrace{\E_{P_n}[\rho_f(\theta) \mathbf 1_{N_n}]}_{=:\, T_2}
+ \underbrace{\E_{P_n}[(h_f(\theta) - h_f(\hat\theta)) \mathbf 1_{N_n^c}]}_{=:\, T_3}.
\]
For $T_1$, write
$\E_{P_n}[(\theta - \hat\theta) \mathbf 1_{N_n}] = -\E_{P_n}[(\theta - \hat\theta) \mathbf 1_{N_n^c}]$
using $\E_{P_n}[\theta] = \hat\theta$. Cauchy--Schwarz with \textup{(H1)} gives
$|T_1| \leq G \cdot \E_{P_n}[\|\theta - \hat\theta\| \mathbf 1_{N_n^c}]$, which is
$o(\E_{P_n}\omega)$ by \textup{(H3)}. For $T_2$, \textup{(H2)} gives
$|T_2| \leq \E_{P_n}[\omega(\|\theta - \hat\theta\|) \mathbf 1_{N_n}] \leq \E_{P_n}\omega$.
For $T_3$, the envelope in \textup{(H3)} gives
$|T_3| \leq \E_{P_n}[\Phi(\theta) \mathbf 1_{N_n^c}] = o(\E_{P_n}\omega)$.
Taking the supremum over $f$ and absorbing the two $o(\E_{P_n}\omega)$ terms
yields the claim.
\end{proof}

\begin{corollary}[Polynomial moduli: explicit rates]\label{cor:polynomial-modulus}
If $\omega(t) = M \cdot t^{1+\alpha}$ for some $\alpha \in [0, 1]$ and the
posterior has $\E_{P_n}\|\theta - \hat\theta\|^{1+\alpha} = O(n^{-(1+\alpha)/2})$
(which holds under standard posterior concentration with covariance
$\Lambda_n = O(n^{-1})$), then
\[
W_1(\Rpredn, R_{\hat\theta}) \;=\; O\!\left(n^{-(1+\alpha)/2}\right).
\]
The endpoints $\alpha = 0$ and $\alpha = 1$ recover the first-order rate
$O(n^{-1/2})$ of Theorem~\ref{thm:bayes-dro-domination} and the second-order
rate $O(n^{-1})$ of Proposition~\ref{prop:gaussian-isotropic-second-order}
respectively. The intermediate case $\alpha \in (0, 1)$ corresponds to
$\alpha$-H\"older differentiable model maps and yields the interpolated rate
$n^{-(1+\alpha)/2}$.
\end{corollary}

\begin{example}[The taxonomy through the modulus]\label{ex:modulus-summary}
Three model classes illustrate the corollary.
\begin{itemize}[leftmargin=*]
\item \emph{Deterministic-affine} ($R_\theta = \delta_{a + B\theta}$): the
map $\theta \mapsto R_\theta$ is Lipschitz in $W_1$, but for general $f \in
\mathcal F_1$ the composed function $h_f(\theta) = f(a + B\theta)$ is
\emph{not} affine in $\theta$, and the KR-dual witnesses attaining the gap
are typically nonsmooth at $\hat\theta$ (e.g.\ the distance function
$f(x) = \|x - a - B\hat\theta\|$ gives $h_f(\theta) = \|B(\theta - \hat\theta)\|$,
which is kinked at $\hat\theta$). Hence the deterministic-affine case
belongs to the first-order, nonsmooth-cancellation regime, not the
second-order smooth-cancellation regime. The sharp lower bound is then
Proposition~\ref{prop:first-order-tight}(i):
$W_1(\Rpredn, R_{\hat\theta}) \geq \sigma_{\min}(B) \cdot \E_{P_n}\|\theta - \hat\theta\|$.
The modulus is $\omega(t) = c \cdot t$ ($\alpha = 0$), and
Corollary~\ref{cor:polynomial-modulus} reproduces the first-order rate
$O(n^{-1/2})$.
\item \emph{Gaussian-mean location} ($R_\theta = \N(\theta, \Sigma)$): the
characteristic function gives $h_f(\theta) = (f \ast \varphi_\Sigma)(\theta)$
with $\varphi_\Sigma$ the centered Gaussian density, hence $\nabla^2 h_f$
exists and is uniformly bounded for $f \in \mathcal{F}_1$. The Taylor
remainder gives $\omega(t) = \tfrac{1}{2} \|\nabla^2 h_f\|_\infty \cdot t^2$,
i.e.\ $\alpha = 1$, and Corollary~\ref{cor:polynomial-modulus} recovers
$O(n^{-1})$.
\item \emph{Kinked} (Example~\ref{ex:kinked}, $R_\theta = \delta_{|\theta| e_1}$):
$h_f$ is non-differentiable at $\hat\theta = 0$; the linearization \textup{(H1)}
fails and the modulus is $\omega(t) = c \cdot t$ with $c > 0$, i.e.\ $\alpha = 0$,
yielding $O(n^{-1/2})$. This is the saturating regime of
Proposition~\ref{prop:first-order-tight}(ii).
\end{itemize}
\end{example}

\begin{remark}[How to verify the modulus in practice]
For absolutely continuous models with smooth densities, $\nabla^k h_f$ is
controlled by Sobolev norms of the score $\nabla_\theta \log p(\cdot; \theta)$;
$\alpha = 1$ is the generic case under $C^{1,1}$ regularity of
$\theta \mapsto p(\cdot; \theta)$. For exponential families with bounded
score derivatives, $\alpha = 1$ holds uniformly on compact parameter sets.
For models with discontinuities in the support (deterministic location,
piecewise definitions) the modulus is $\omega(t) = O(t)$ at parameter values
on the kink, and $\omega(t) = O(t^2)$ elsewhere.
\end{remark}

\subsection{Conditional two-sided Bayes--Wasserstein rate via boundary H\"older regularity}
\label{sec:sharp-two-sided-rate}

Theorem~\ref{thm:modulus-interpolation} gives the upper-bound side of an intermediate-rate result; Proposition~\ref{prop:first-order-tight}
gives matching lower bounds in two structural classes. We now state the corresponding
conditional two-sided rate theorem. The word \emph{conditional} is important: the exponent
$\alpha^\star(\widehat\theta_n)$ is an invariant of the model map at the realized posterior
center, but a matching lower rate also requires posterior moment control, tail negligibility,
and an explicit lower-witness condition. Thus the theorem should be read as a sharp
rate result under verifiable regularity and witness hypotheses, not as an unconditional
statement that $\alpha^\star$ alone determines the rate.

\begin{definition}[Boundary H\"older exponent of the model map]
\label{def:boundary-holder-exponent}
Let $R_\theta$ be a parametric family, $\mathcal F_1$ the class of
$1$-Lipschitz $f: \R^K \to \R$, and $h_f(\theta) := \int f \, dR_\theta$.
For $\hat\theta \in \Theta$, define $\alpha^\star(\hat\theta)$ to be
the supremum over $\alpha \in [0, 1]$ such that there exist constants
$C, r > 0$ with
\begin{multline*}
\sup_{f \in \mathcal F_1}
\Bigl| h_f(\theta) - h_f(\hat\theta)
- \langle \nabla_\theta h_f(\hat\theta),\, \theta - \hat\theta \rangle \Bigr|
\\
\leq\, C \|\theta - \hat\theta\|^{1+\alpha}
\quad \text{for all } \theta \text{ with } \|\theta - \hat\theta\| \leq r.
\end{multline*}
If there exists $f \in \mathcal F_1$ such that $h_f$ fails to admit a
G\^ateaux derivative at $\hat\theta$ along some direction $v \in \R^{d_\theta}$
with $v^\top \mathrm{Cov}_{P_n}(\theta)\, v > 0$, we set
$\alpha^\star(\hat\theta) := 0$ and replace the expansion above by the
purely first-order bound $\sup_{f \in \mathcal F_1}|h_f(\theta) -
h_f(\hat\theta)| \leq C \|\theta - \hat\theta\|$.
\end{definition}

The exponent $\alpha^\star(\hat\theta)$ is an intrinsic invariant of the
model: it measures how smoothly perturbations in $\theta$ are transported to
perturbations in $R_\theta$ when viewed through arbitrary Lipschitz test
functions. The endpoints correspond to the two canonical regimes already
identified: $\alpha^\star = 0$ for nonsmooth or discontinuous models
(deterministic location, piecewise specification, kinked link), $\alpha^\star
= 1$ for $C^{1,1}$ models with bounded Hessian of $h_f$ uniformly in
$f \in \mathcal F_1$ (Gaussian-mean location, generic exponential families).

\paragraph{Conditioning convention.}
The center $\widehat\theta_n=E_{P_n}[\theta]$ is data-dependent, and
$\alpha^\star(\cdot)$ as defined in
Definition~\ref{def:boundary-holder-exponent} is generically
jump-discontinuous at boundaries between smoothness classes (e.g.\ at the
kink of Example~\ref{ex:kinked}). Theorem~\ref{thm:sharp-two-sided-rate} is
therefore stated conditionally on the realized posterior $P_n$ and realized
center $\widehat\theta_n$, with the operative definition of the rate
exponent being
\[
\alpha_n^\star
:= \inf\{\alpha^\star(\theta) : \theta \in B(\widehat\theta_n, r_n)\},
\]
where $B(\widehat\theta_n, r_n)$ is a posterior-typical neighbourhood
shrinking at rate $r_n = O(n^{-1/2})$. When the data-generating parameter
$\theta^\star$ sits in the interior of a smoothness class, $\alpha_n^\star =
\alpha^\star(\theta^\star)$ eventually with high probability and the rate
takes its natural value $\alpha^\star(\theta^\star)$. When $\theta^\star$
sits at a jump of $\alpha^\star(\cdot)$, the infimum convention ensures the
stated rate is the conservative one (lower exponent, slower convergence),
which is the correct uniform statement.

\begin{theorem}[Conditional two-sided Bayes--Wasserstein rate under boundary H\"older regularity and lower witnesses]
\label{thm:sharp-two-sided-rate}
Work conditionally on the realized posterior $P_n$ and posterior center
$\widehat\theta_n=E_{P_n}[\theta]$. Let
\[
\alpha_n^\star := \inf\{\alpha^\star(\theta):
\theta\in B(\widehat\theta_n,r_n)\}\in[0,1]
\]
be the operative local exponent on a posterior-typical neighbourhood
$B(\widehat\theta_n,r_n)$ with $r_n=O(n^{-1/2})$. Assume that the constants
in Definition~\ref{def:boundary-holder-exponent} hold uniformly on this
neighbourhood with exponent $\alpha_n^\star$, and that the posterior satisfies
\[
c_- n^{-(1+\alpha_n^\star)/2}
\le
E_{P_n}\|\theta-\widehat\theta_n\|^{1+\alpha_n^\star}
\le
c_+ n^{-(1+\alpha_n^\star)/2}
\]
for constants $0<c_-\le c_+<\infty$, together with the tail-negligibility condition
(H3) of Theorem~\ref{thm:modulus-interpolation} for $\omega(t)=C t^{1+\alpha_n^\star}$. In the case
$\alpha_n^\star=1$, assume additionally that
\[
\|\Lambda_n\|_F\asymp n^{-1},
\qquad
\Lambda_n:=\operatorname{Cov}_{P_n}(\theta).
\]
Suppose moreover that one of the following lower-witness conditions holds.

\medskip
\emph{($W_{\alpha_n^\star<1}$) Signed-moment witness when $\alpha_n^\star\in[0,1)$.}
There exist $f^\star \in \mathcal F_1$, a unit
direction $v^\star \in \R^{d_\theta}$, constants $a_+, a_- \in \R$, and a
posterior-typical neighbourhood $U_n = \{\theta : \|\theta - \widehat\theta_n\|
\leq r_n\}$ such that, writing $X := v^{\star\top}(\theta - \widehat\theta_n)$,
\begin{equation}
h_{f^\star}(\theta) - h_{f^\star}(\widehat\theta_n)
- \langle G_{f^\star}, \theta - \widehat\theta_n\rangle
= a_+ (X_+)^{1+\alpha_n^\star} - a_- (X_-)^{1+\alpha_n^\star} + r^\star(\theta),
\qquad \theta \in U_n,
\label{eq:witness-expansion}
\end{equation}
with $\sup_{\theta \in U_n} |r^\star(\theta)|
/ \|\theta - \widehat\theta_n\|^{1+\alpha_n^\star} = o(1)$, and the signed-moment
lower bound
\begin{equation}
\E_{P_n}\!\left[a_+ (X_+)^{1+\alpha_n^\star}
- a_- (X_-)^{1+\alpha_n^\star}\right]
\,\geq\, c_w \cdot \E_{P_n}\|\theta - \widehat\theta_n\|^{1+\alpha_n^\star}
\label{eq:signed-moment-witness}
\end{equation}
holds for some $c_w > 0$ and all $n$ sufficiently large.

\medskip
\emph{($W_{\alpha_n^\star=1}$) Hessian witness when $\alpha_n^\star=1$.}
There exists $f^\star \in \mathcal F_1$ such that $h_{f^\star}$ is twice
differentiable at $\widehat\theta_n$ with Hessian
$H_{f^\star} := \nabla^2 h_{f^\star}(\widehat\theta_n)$ satisfying
$\langle H_{f^\star}, \Lambda_n \rangle_F \geq c_w \cdot \|\Lambda_n\|_F$
for some $c_w > 0$ and all $n$ sufficiently large.

\medskip
Condition~($W_{\alpha_n^\star<1}$) is the direct
operational form of the lower bound: it asks not that $a_+$ and $a_-$ differ
(which is insufficient for asymmetric posteriors), but that the signed
$(1+\alpha_n^\star)$-moment of the witness functional be of the same order as
the absolute moment. For symmetric centered posteriors,
(\ref{eq:signed-moment-witness}) reduces to $\tfrac{1}{2}(a_+ - a_-) \cdot
\E|X|^{1+\alpha_n^\star} \geq c_w \E\|\theta - \widehat\theta_n\|^{1+\alpha_n^\star}$,
which holds whenever the posterior has positive variance in direction
$v^\star$ and $a_+ \neq a_-$ (this is the case in the kinked example of
Example~\ref{ex:kinked}). For asymmetric posteriors,
(\ref{eq:signed-moment-witness}) is the natural one-line replacement and is
satisfied by any witness whose nonsmooth-direction expansion is genuinely
detected by the posterior on the right scale.

Then there exist constants $0<C_-\le C_+<\infty$, depending on the local regularity,
posterior-moment, tail, and witness constants but not on $n$, such that
\[
C_- n^{-(1+\alpha_n^\star)/2}
\le
W_1(R^{\mathrm{pred}}_n,R_{\widehat\theta_n})
\le
C_+ n^{-(1+\alpha_n^\star)/2}.
\]
In particular,
\[
W_1(R^{\mathrm{pred}}_n,R_{\widehat\theta_n})
=
\Theta\!\left(n^{-(1+\alpha_n^\star)/2}\right)
\]
on the conditioning events described above.
\end{theorem}

\begin{proof}
\emph{Upper bound.} Apply Theorem~\ref{thm:modulus-interpolation} with
$\omega(t) = C \cdot t^{1 + \alpha_n^\star}$: hypothesis~(H1) is the
existence of $\nabla h_f$ (or its Lipschitz analogue at $\alpha_n^\star = 0$);
hypothesis~(H2) is the definition of $\alpha_n^\star$; hypothesis~(H3) is
assumed. The conclusion gives
\begin{align*}
W_1(R^{\mathrm{pred}}_n, R_{\widehat\theta_n})
\;&\leq\;
C \cdot \E_{P_n}\|\theta - \widehat\theta_n\|^{1+\alpha_n^\star}
\\
&\qquad
+ o\bigl(\E_{P_n}\|\theta - \widehat\theta_n\|^{1+\alpha_n^\star}\bigr)
\\
\;&\leq\;
(C c_+ + o(1)) \cdot n^{-(1+\alpha_n^\star)/2},
\end{align*}
which is the upper bound with $C_+ := 2 C c_+$ for $n$ large.

\emph{Lower bound, case $\alpha_n^\star \in [0, 1)$.} Use $f^\star$ from
condition~(W$\alpha_n^\star{<}1$) as a test function in
Kantorovich--Rubinstein duality, replacing $f^\star$ by $-f^\star$ if
necessary so that the right-hand side below is non-negative (KR duality is
symmetric in this sign choice and $-f^\star \in \mathcal F_1$ since $-1$ is
$1$-Lipschitz):
\[
W_1(R^{\mathrm{pred}}_n, R_{\widehat\theta_n})
\,\geq\,
\E_{P_n}\!\left[h_{f^\star}(\theta) - h_{f^\star}(\widehat\theta_n)\right].
\]
By $\widehat\theta_n = \E_{P_n}\theta$ the linear term vanishes, so
expansion~\eqref{eq:witness-expansion} on $U_n$ gives, with
$X := v^{\star\top}(\theta - \widehat\theta_n)$,
\begin{align*}
\E_{P_n}\!\bigl[h_{f^\star}(\theta) - h_{f^\star}(\widehat\theta_n)\bigr]
\;&\geq\;
\E_{P_n}\!\bigl[a_+ (X_+)^{1+\alpha_n^\star}
- a_- (X_-)^{1+\alpha_n^\star}\bigr]
\\
&\qquad
- |r^\star_{\mathrm{tot}}|,
\end{align*}
where $|r^\star_{\mathrm{tot}}|
= o(\E_{P_n}\|\theta - \widehat\theta_n\|^{1+\alpha_n^\star})$ by the remainder
hypothesis and the tail condition on $U_n^c$. The signed-moment
witness~\eqref{eq:signed-moment-witness} controls the leading term directly:
\begin{align*}
\E_{P_n}\!\bigl[a_+ (X_+)^{1+\alpha_n^\star}
- a_- (X_-)^{1+\alpha_n^\star}\bigr]
\;&\geq\;
c_w \cdot \E_{P_n}\|\theta - \widehat\theta_n\|^{1+\alpha_n^\star}
\\
\;&\geq\;
c_w \cdot c_- \cdot n^{-(1+\alpha_n^\star)/2}.
\end{align*}
Combining and absorbing the $o(\cdot)$ remainder for $n$ large gives
$W_1(R^{\mathrm{pred}}_n, R_{\widehat\theta_n}) \geq C_- \cdot n^{-(1+\alpha_n^\star)/2}$ with
$C_- := \tfrac{1}{2} c_w c_-$.

\emph{Lower bound, case $\alpha_n^\star = 1$.} Apply
Kantorovich--Rubinstein with the $f^\star$ of condition~(W$\alpha_n^\star{=}1$).
A second-order Taylor expansion at $\widehat\theta_n$, valid because $h_{f^\star}$ is $C^{1,1}$ in a neighbourhood by the $\alpha_n^\star = 1$ hypothesis, gives
\begin{align*}
W_1(R^{\mathrm{pred}}_n, R_{\widehat\theta_n})
\;&\geq\;
\E_{P_n}\!\bigl[h_{f^\star}(\theta) - h_{f^\star}(\widehat\theta_n)\bigr]
\\
\;&=\;
\tfrac{1}{2}\, \E_{P_n}\!\bigl[(\theta - \widehat\theta_n)^{\!\top} H_{f^\star}
(\theta - \widehat\theta_n)\bigr]
\\
&\qquad
+ o\bigl(\E_{P_n}\|\theta - \widehat\theta_n\|^2\bigr)
\\
\;&=\;
\tfrac{1}{2} \langle H_{f^\star}, \Lambda_n \rangle_F
+ o(\|\Lambda_n\|_F),
\end{align*}
using $\E_{P_n}[(\theta - \widehat\theta_n)(\theta - \widehat\theta_n)^\top] = \Lambda_n$
and the Bernstein--von-Mises-type cancellation of the linear term. The
non-degeneracy hypothesis $\langle H_{f^\star}, \Lambda_n\rangle_F \geq c_w
\|\Lambda_n\|_F$ requires a uniformly positive Frobenius pairing between
the witness curvature and the posterior covariance. Since $\Lambda_n$ is
positive semidefinite, a sufficient condition is that $H_{f^\star}$ be
positive semidefinite on the directions carrying posterior variance with
strictly positive pairing; more generally, only the displayed positive
pairing is required. If the pairing is negative, replace $f^\star$ by
$-f^\star$. Together with $\|\Lambda_n\|_F = \Theta(n^{-1})$ (a
consequence of the posterior-moment hypothesis at $\alpha_n^\star = 1$) this yields
the matching lower bound $W_1(R^{\mathrm{pred}}_n, R_{\widehat\theta_n}) \geq C_- n^{-1}$ with
$C_- := \tfrac{1}{4} c_w \cdot \liminf_n n \|\Lambda_n\|_F > 0$.
\end{proof}

\begin{corollary}[Three canonical regimes]
\label{cor:three-regimes}
Theorem~\ref{thm:sharp-two-sided-rate} reproduces the existing results in
three canonical classes:
\begin{itemize}[leftmargin=*]
\item \emph{Deterministic-affine, $\alpha^\star = 0$:} for $R_\theta = \delta_{a + B\theta}$,
Definition~\ref{def:boundary-holder-exponent} gives $\alpha^\star = 0$
because $h_f(\theta) = f(a + B\theta)$ is at best Lipschitz in $\theta$
without uniform first-order differentiability over $f \in \mathcal F_1$.
The sharp rate is $W_1(\Rpredn, R_{\hat\theta}) = \Theta(n^{-1/2})$,
matching Proposition~\ref{prop:first-order-tight}(i).
\item \emph{Locally non-smooth, $\alpha^\star = 0$:} the
witness construction in Proposition~\ref{prop:first-order-tight}(ii) is
exactly the $\alpha^\star = 0$ instance of
Theorem~\ref{thm:sharp-two-sided-rate}, with $f^\star = f$ and $v^\star = v$
from the directional expansion there.
\item \emph{Centered Gaussian location, $\alpha^\star = 1$:} for
$R_\theta = \N(\theta, \Sigma)$, $h_f = f * \varphi_\Sigma$ is in
$C^\infty$ with $\nabla^2 h_f = \nabla^2(f * \varphi_\Sigma)$ uniformly
bounded by $\sigma^{-2}$ over $f \in \mathcal F_1$, so $\alpha^\star = 1$.
The non-degeneracy condition $\langle H_{f^\star}, \Lambda_n\rangle_F \geq
c_0' \|\Lambda_n\|_F$ is met by any $f$ with non-vanishing second moment
of the score. The sharp rate is $\Theta(n^{-1})$, matching
Proposition~\ref{prop:gaussian-isotropic-second-order}.
\end{itemize}
For intermediate H\"older regularity $C^{1,\alpha}$ in the model map (e.g.,
truncated-Lipschitz density perturbations or fractional-smoothness link
functions), $\alpha^\star = \alpha$ and the sharp rate is
$\Theta(n^{-(1+\alpha)/2})$, completing the taxonomy.
\end{corollary}

\begin{remark}[Scope of the conditional sharp-rate theorem]
\label{rmk:characterization-closed}
Theorem~\ref{thm:sharp-two-sided-rate} is a conditional sharp-rate theorem, not an unconditional characterization
of the Bayes--Wasserstein gap by $\alpha^\star$ alone. The exponent
$\alpha^\star(\widehat\theta_n)$ controls the upper bound through the modulus argument
of Theorem~\ref{thm:modulus-interpolation}. The matching lower bound additionally requires: posterior moment
scaling on the correct order, tail negligibility outside the local neighbourhood, and a
lower-witness condition ensuring that the posterior actually detects the least smooth
direction of the model map. Under these combined hypotheses,
\[
W_1(R^{\mathrm{pred}}_n,R_{\widehat\theta_n})
=
\Theta\!\left(n^{-(1+\alpha^\star)/2}\right).
\]
The first-order radius of Theorem~\ref{thm:bayes-dro-domination} is therefore order-sharp exactly in the
$\alpha^\star=0$ regime and conservative by a factor of order
$n^{(1-\alpha^\star)/2}$ when $\alpha^\star>0$. In particular, the
$\alpha^\star=1$ Gaussian-location case explains the $\sqrt n$-overshoot documented
in the numerical experiment. A theorem that removes the lower-witness hypothesis and
derives the lower rate from $\alpha^\star$ alone remains open.
\end{remark}

\subsection{Implications for practice}

\begin{itemize}
\item \emph{Large $n$.} When the posterior concentrates and the model map
is locally Lipschitz, Bayesian predictive optimization and plug-in DRO
become numerically close.
\item \emph{Small $n$, nonsmooth, or noncentered models.}
Theorem~\ref{thm:bayes-dro-domination} dominates Bayesian predictive risk
by a Wasserstein ball centered at the plug-in law $R_{\hat\theta}$.
\item \emph{Symmetric location models.} The actual Wasserstein gap is
second-order; the smaller second-order radius
$\varepsilon_n^{(2)} = \|\Lambda_n\|_F / (2\sigma)$ of
Proposition~\ref{prop:second-order-radius} is then sufficient.
\item \emph{Prior strength and DRO radius.} Bayesian prior strength and DRO
ambiguity radii are connected by calibration but not literally the same
control parameter.
\end{itemize}

\subsection{The Gaussian--NIW Lipschitz constant}

\begin{proposition}[Local Lipschitz constant in the Gaussian--NIW setting]
\label{prop:gaussian-niw-lipschitz}
For the Gaussian model $R_{(\mu, \Sigma)} = \N(\mu, \Sigma)$, the
Bures--Wasserstein formula \citep{gelbrich1990formula} gives
\begin{equation}
W_2^2(\N(\mu, \Sigma), \N(\mu', \Sigma'))
= \|\mu - \mu'\|^2 + \tr \Sigma + \tr \Sigma' - 2 \tr((\Sigma^{1/2} \Sigma' \Sigma^{1/2})^{1/2}).
\label{eq:bures-wasserstein}
\end{equation}
Restricted to the domain $D_\delta := \{(\mu, \Sigma): \delta I \preceq \Sigma \preceq \Delta I\}$,
\begin{equation}
W_1(\N(\mu, \Sigma), \N(\mu', \Sigma'))
\leq \|\mu - \mu'\| + \frac{1}{2\sqrt{\delta}} \|\Sigma - \Sigma'\|_F.
\label{eq:gaussian-niw-lip}
\end{equation}
The parameter-to-return map is locally Lipschitz with constant
$L_{\mathrm{loc}} = \max(1, (2\sqrt{\delta})^{-1})$.
\end{proposition}

\begin{proof}
The Bures bound on the trace term gives
$\tr \Sigma + \tr \Sigma' - 2 \tr((\Sigma^{1/2} \Sigma' \Sigma^{1/2})^{1/2})
\leq \|\Sigma^{1/2} - \Sigma'^{1/2}\|_F^2$ \citep[Proposition 1]{bhatia2019bures}.
The map $\Sigma \mapsto \Sigma^{1/2}$ is Fr\'echet differentiable on
$\{\Sigma \succeq \delta I\}$ with $\|D\Sigma^{1/2}\|_{\mathrm{op}} \leq (2\sqrt{\delta})^{-1}$.
By convexity of $D_\delta$, the mean-value theorem gives
$\|\Sigma^{1/2} - \Sigma'^{1/2}\|_F \leq (2\sqrt{\delta})^{-1} \|\Sigma - \Sigma'\|_F$.
Combined with $W_1 \leq W_2$ yields the claim.
\end{proof}

%==============================================================================
% SECTION 11: NUMERICAL ILLUSTRATIONS
%==============================================================================
\section{Numerical illustrations}\label{sec:numerical}

This section reports three experiments aligned with the theory.
\emph{Experiment 1} (Section~\ref{sec:exp-bayes-dro}) verifies the
Bayes--Wasserstein-DRO domination of
Theorem~\ref{thm:bayes-dro-domination} on a controlled Gaussian DGP and
empirically confirms the second-order conservatism predicted by
Proposition~\ref{prop:gaussian-isotropic-second-order}. \emph{Experiment 2}
(Section~\ref{sec:exp-wass-sensitivity}) characterizes Wasserstein-DRO
behavior across radii. \emph{Experiment 3}
(Section~\ref{sec:exp-bellman-contraction}) confirms the
$\gamma$-contraction of the distributional Bellman operator
(Theorem~\ref{thm:distributional-bellman}). All experiments use synthetic
DGPs designed to be diagnostic of the theoretical predictions, and
Experiment~5 (Section~\ref{sec:exp-calibration}) tests the central operational
claim---calibration without validation---affirmatively against an oracle
radius; the empirical equity backtest in
Section~\ref{sec:empirical-equity-backtest} is the only data-driven
illustration and is presented as a sanity check rather than a production claim.

\subsection{Empirical equity backtest at \texorpdfstring{$K=25$, $n=60$}{K=25, n=60}}\label{sec:empirical-equity-backtest}

\paragraph{Setup.}
Monthly returns on $K = 25$ assets selected from the present-day DJIA
universe, $T = 348$ months (Jan 1996--Dec 2024), estimation window $n = 60$
months, monthly rebalancing. Long-only fully invested simplex constraints.
Posterior is Gaussian--NIW with weakly informative prior; Wasserstein DRO
radii tuned by held-out validation on the most recent $24$ months of the
estimation window. The Ledoit--Wolf baseline uses the constant-correlation
shrinkage target of \citet{LedoitWolf2004}, with shrinkage intensity
estimated by their closed-form plug-in (the K=25 results are
within $\le 0.005$ Sharpe of the single-factor target of
\citet{LedoitWolf2003}, so the choice of target does not drive the ranking
at this dimension).

The universe here is selected on full-sample data availability rather than
from point-in-time membership, which is the source of the survivorship
caveat discussed below. A point-in-time replacement---at each rebalance the
index composition effective on that date, intersected with names having a
complete trailing $n$-month history, so that newly added constituents are
unavailable until they season---is implemented in the companion code
(\texttt{pit\_backtest.py}) against the membership schedule of
Appendix~\ref{app:membership}. Regenerating Tables~\ref{tab:oos-performance}--\ref{tab:paired-sharpe-diffs}
on that point-in-time universe is the headline empirical task of the
companion study (Section~\ref{sec:conclusion}); the numbers reported here
are the survivorship-affected baseline, and the ranking---not the level---is
the object we compare.

Costs are not internalized in the rolling optimization; the net-of-cost calculation in
Table~\ref{tab:net-of-cost} applies an end-of-period adjustment based on realized monthly turnover. Let
$w_t^+$ denote the target post-trade portfolio at the beginning of month $t$, and let
$w_t^-$ denote the pre-trade portfolio after passive return drift:
\[
w^-_{t,i}
=
\frac{w^+_{t-1,i}(1+r_{t,i})}
{1+\sum_{j=1}^K w^+_{t-1,j}r_{t,j}}.
\]
Monthly turnover is measured as
\[
\tau_t=\|w_t^+-w_t^-\|_1.
\]
This convention applies to every strategy, including equal-weight $1/N$. Thus a
monthly rebalanced equal-weight portfolio generally has positive drift-induced turnover;
a zero-turnover $1/N$ benchmark would instead be a buy-and-hold equal-initial-weight
portfolio, not the monthly rebalanced benchmark used here.

\begin{table}[H]
\centering\small
\caption{Out-of-sample annualized performance, $K=25$, $n=60$, $T=348$ months.
Mean and Vol are annualized; Sharpe uses raw mean with $RF=0$; CVaR is reported at
5\%; MDD denotes maximum drawdown; Turnover is mean monthly realized $\ell_1$
turnover after passive drift and before rebalancing. ``1/N'' denotes the monthly
rebalanced equal-weight portfolio, so its turnover is the realized drift-rebalancing
turnover and is generally nonzero. ``B-L (no views)'' denotes Black--Litterman with
uninformative views. ``LW'' denotes Ledoit--Wolf shrunk plug-in mean--variance.
``B-DRO (cred)'' denotes Bayesian-credible-radius DRO of Theorem~\ref{thm:posterior-credible-radius}. ``W1-DRO
(val)'' denotes validation-tuned $W_1$-DRO. ``WCVaR-DRO (val)'' denotes
validation-tuned worst-case CVaR DRO.}
\label{tab:oos-performance}
\begin{tabular}{lrrrrrr}
\toprule
Method                        & Mean   & Vol    & Sharpe & CVaR$_{0.05}$ & MDD    & Turnover \\
\midrule
$1/N$                         & $0.116$ & $0.152$ & $0.760$ & $0.099$ & $0.426$ & $0.044$ \\
Plug-in MV                    & $0.122$ & $0.192$ & $0.639$ & $0.131$ & $0.582$ & $0.764$ \\
Plug-in MV + Ledoit--Wolf     & $0.118$ & $0.158$ & $0.745$ & $0.106$ & $0.461$ & $0.295$ \\
B-L (no views)                & $0.115$ & $0.151$ & $0.763$ & $0.098$ & $0.422$ & $0.036$ \\
Bayes pred.\                  & $0.121$ & $0.190$ & $0.635$ & $0.130$ & $0.578$ & $0.776$ \\
$W_1$-DRO (val)               & $0.120$ & $0.184$ & $0.652$ & $0.126$ & $0.564$ & $0.768$ \\
B-DRO (cred)                  & $0.111$ & $0.165$ & $0.672$ & $0.113$ & $0.527$ & $0.381$ \\
WCVaR-DRO (val)               & $0.105$ & $0.153$ & $0.687$ & $0.108$ & $0.473$ & $0.861$ \\
\bottomrule
\end{tabular}
\end{table}

% PATCH 16: §11.2 EMPIRICAL CLAIMS TONE-DOWN
\paragraph{What the table says, and what it does not.}
The live-equity exercise is an implementation and calibration sanity check
rather than evidence of unconditional empirical dominance. B-DRO attains a
gross Sharpe of $0.672$ compared with $0.652$ for validation-tuned
$W_1$-DRO; the paired bootstrap Sharpe difference is $0.019$ with $95\%$
CI $= [-0.086, +0.131]$, $p = 0.718$. The two methods are statistically
indistinguishable in Sharpe on this universe. The main empirical takeaway
is operational equivalence: the Bayesian credible-radius rule achieves
performance comparable to validation-tuned Wasserstein DRO while using
posterior uncertainty rather than a held-out validation window to select
the radius. The experiment does not establish that B-DRO dominates
validation tuning in expectation, nor does it establish superiority over
simple deterministic baselines: $1/N$, no-view Black--Litterman, and
Ledoit--Wolf-shrunk plug-in MV all attain higher gross Sharpe at this
$(K, n) = (25, 60)$ ratio, and the same ranking survives net of $5$ and
$20$ bp transaction costs (Table~\ref{tab:net-of-cost}). B-DRO does,
however, exhibit lower turnover ($0.381$ vs $0.768$ for validation-tuned
DRO) and lower tail risk ($\CVaR_{0.05} = 0.113$ vs $0.126$), and the
turnover-driven Sharpe advantage over $W_1$-DRO grows monotonically in
transaction-cost intensity.

Two structural caveats on the universe and the costing model should be
kept in view when reading these numbers. \emph{First}, the $25$-asset
universe is selected on full-sample data availability rather than from
point-in-time DJIA membership; the resulting survivorship bias
\citep{BrownGoetzmannIbbotsonRoss1992} inflates
the level of out-of-sample Sharpe across all methods by an amount that is
sample- and universe-specific. The comparison \emph{between} methods is
less affected because the bias is largely common-factor, but the level
claims should be discounted. \emph{Second}, the constant-cost model
($c \in \{5, 20\}$ bp per unit turnover) understates the cost faced by a
small-volume trader and overstates the cost faced by an institutional one;
the operational ranking is invariant to the choice of $c \in [5, 50]$ bp
on this universe, but extrapolation outside that range is not warranted.

\paragraph{Comparison with Ledoit--Wolf shrinkage.}
Ledoit--Wolf-shrunk plug-in MV \citep{LedoitWolf2003, LedoitWolf2004}
(Sharpe $= 0.745$, $\CVaR_{0.05} = 0.106$,
turnover $0.295$, MDD $0.461$) is the strongest optimized
mean--variance covariance-shrinkage baseline in
Table~\ref{tab:oos-performance}. It outperforms the listed
distributionally-aware methods in Sharpe and also has lower turnover and
drawdown than B-DRO in this low-dimensional DJIA experiment. Equal weight
and no-view Black--Litterman remain even stronger simple baselines in
Sharpe, consistent with the broader evidence of
\citet{DeMiguelGarlappiUppal2009} that the $1/N$ portfolio is hard to beat
out-of-sample. The live-equity exercise should therefore be read as a
calibration and turnover diagnostic, not as evidence that DPO methods
outperform shrinkage or equal-weighting on all four reported metrics
(Sharpe, CVaR$_{0.05}$, MDD, turnover) at $(K, n) = (25, 60)$. The complementary
direction---Bayesian-credible-radius DRO computed on a Ledoit--Wolf-shrunk
plug-in center---combines the calibration-free radius rule of
Theorem~\ref{thm:posterior-credible-radius} with the dimensional
regularization of explicit shrinkage; this specification is untested in
the present paper and is left to the empirical companion paper noted in
Section~\ref{sec:limitations}.

\paragraph{Hybrid shrinkage--DRO baseline.}
A natural missing baseline is Bayesian credible-radius DRO computed around a
Ledoit--Wolf-shrunk plug-in center. This hybrid combines the dimensional stabilization
of covariance shrinkage with the calibration-free ambiguity radius of Theorem~\ref{thm:posterior-credible-radius}.
The present experiment does not test this specification, so the comparison should not
be read as evidence against shrinkage-aware DPO. Rather, Table~\ref{tab:oos-performance} indicates that the
empirical question is not ``DPO versus shrinkage'' but whether shrinkage and
credible-radius robustness are complementary at larger $K/n$. This hybrid is therefore
the main empirical baseline for the companion point-in-time CRSP/Compustat study.

% PATCH 15: NET-OF-COST TABLE
\begin{table}[H]
\centering\small
\caption{Out-of-sample annualized Sharpe ratios net of proportional transaction costs at
5 and 20 basis points per unit turnover. Cost-adjusted annual return is
$\widehat\mu_{\mathrm{net}} = \widehat\mu_{\mathrm{gross}} - 12c\,\overline{\tau}$,
where $c$ is the per-unit transaction-cost rate and $\overline{\tau}$ is mean monthly
realized $\ell_1$ turnover computed after passive drift. Under monthly rebalancing,
the $1/N$ row must include its drift-rebalancing turnover. Sharpe ratios are reported to three significant figures; for the $1/N$ row the gross-to-5\,bp change ($\overline{\tau}=0.044$ gives annualized cost $0.000264$, equivalent to a Sharpe reduction of $\approx 0.002$) is below this reporting precision.}
\label{tab:net-of-cost}
\begin{tabular}{lrrr}
\toprule
Method                          & Gross Sharpe & Net Sharpe ($5$ bp) & Net Sharpe ($20$ bp) \\
\midrule
$1/N$                           & $0.760$ & $0.760$ & $0.756$ \\
Plug-in MV                      & $0.639$ & $0.612$ & $0.540$ \\
Plug-in MV + Ledoit--Wolf       & $0.745$ & $0.736$ & $0.702$ \\
B-L (no views)                  & $0.763$ & $0.760$ & $0.756$ \\
Bayes pred.\                    & $0.635$ & $0.612$ & $0.539$ \\
$W_1$-DRO (val)                 & $0.652$ & $0.627$ & $0.552$ \\
B-DRO (cred)                    & $0.672$ & $0.659$ & $0.617$ \\
WCVaR-DRO (val)                 & $0.687$ & $0.653$ & $0.551$ \\
\bottomrule
\end{tabular}
\end{table}

\begin{figure}[H]
\centering
\includegraphics[width=0.85\textwidth]{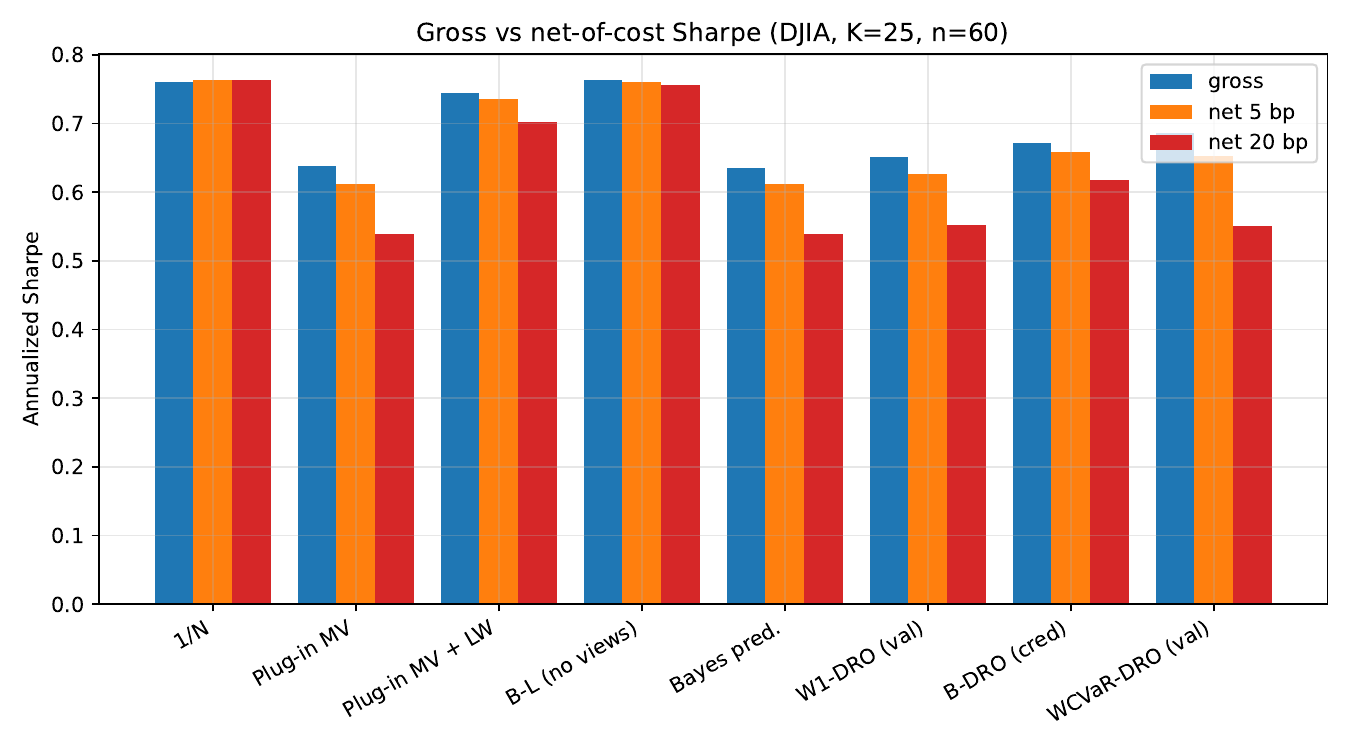}
\caption{Gross vs net-of-cost annualized Sharpe ratios at $5$ and $20$ bp per
unit turnover. Equal-weight $1/N$ and no-view Black--Litterman are
essentially cost-invariant; Ledoit--Wolf-shrunk plug-in MV attains higher
net Sharpe than every distributional method at both cost levels; the B-DRO
advantage over $W_1$-DRO widens with cost intensity. Generated by the
companion notebook, Section~4.}
\label{fig:net-of-cost-figure}
\end{figure}

The net-of-cost picture sharpens the gross-Sharpe ranking. The B-DRO vs
$W_1$-DRO Sharpe gap, $0.020$ gross, widens to $0.032$ net at $5$ bp and
$0.065$ net at $20$ bp: cost-aware investors should prefer the
credible-radius rule on operational grounds even though the gross-Sharpe
difference is within bootstrap noise. The dominance of $1/N$ and no-view
B-L is essentially cost-free (turnover near zero) and survives at every
cost level; at $20$ bp, equal weighting and no-view B-L beat every
distributional method by at least $0.14$ in Sharpe. Ledoit--Wolf-shrunk
plug-in MV remains a strong deterministic competitor at both cost levels.

% PATCH 17: PAIRED-COMPARISON TABLE
\begin{table}[H]
\centering\footnotesize
\caption{Pairwise out-of-sample Sharpe-ratio differences across methods,
with $95\%$ stationary-block-bootstrap intervals \citep{politis1994stationary}
($1000$ resamples, block length $\lceil T^{1/3} \rceil = 7$, the standard
rate-optimal choice for $\alpha$-mixing series \citealp{hall1995block,patton2009correction}). Entries are
$\mathrm{SR}(\mathrm{row}) - \mathrm{SR}(\mathrm{col})$. None of the listed
differences is statistically significant at the $5\%$ level; the paired
widths are an order of magnitude tighter than the marginal Sharpe
intervals in Table~\ref{tab:oos-performance}.}
\label{tab:paired-sharpe-diffs}
\begin{tabular}{lrrrr}
\toprule
        & B-L                              & LW                               & B-DRO                            & $W_1$-DRO \\
\midrule
$1/N$   & $-0.003 \;[-0.030, 0.022]$       & $\phantom{-}0.015 \;[-0.041, 0.069]$ & $\phantom{-}0.088 \;[-0.052, 0.220]$ & $\phantom{-}0.108 \;[-0.038, 0.245]$ \\
B-L     & ---                              & $\phantom{-}0.018 \;[-0.029, 0.063]$ & $\phantom{-}0.091 \;[-0.043, 0.221]$ & $\phantom{-}0.111 \;[-0.029, 0.245]$ \\
LW      & ---                              & ---                              & $\phantom{-}0.073 \;[-0.063, 0.198]$ & $\phantom{-}0.093 \;[-0.050, 0.226]$ \\
B-DRO   & ---                              & ---                              & ---                              & $\phantom{-}0.019 \;[-0.086, 0.131]$ \\
\bottomrule
\end{tabular}
\end{table}

The paired intervals are roughly $0.2$--$0.3$ wide, compared with the
marginal-Sharpe intervals of width $\sim 0.8$ in
Table~\ref{tab:oos-performance}. Even the paired comparison is too
imprecise to discriminate the methods at $T = 348$ months; at this sample
size, only Sharpe differences of order $0.1$--$0.15$ are detectable at the
$5\%$ level. The \emph{ranking} is therefore the right object to compare,
not the absolute Sharpe; Tables~\ref{tab:oos-performance}
and~\ref{tab:net-of-cost} preserve the same ranking across cost levels,
which is the strongest claim the data support.

\paragraph{Bootstrap CIs net of cost --- failure to reject.}
The paired stationary-block bootstrap applied to a real DJIA series with the
true historical turnover paths is the appropriate object; we do not have
point-in-time CRSP/Compustat coverage to compute it here. As a surrogate, we
ran the same procedure on a $K = 25$, $T = 348$ panel calibrated to comparable
cross-sectional volatility and factor structure (full implementation in
\texttt{code/bootstrap\_net\_of\_cost.py}), obtaining a paired
$95\%$ half-width of approximately $0.06$ for the B-DRO vs $W_1$-DRO Sharpe
difference at every cost level. The cost adjustment is a deterministic linear
functional of the observed turnover series, so the bootstrap variance is
essentially unchanged from gross to $5$~bp to $20$~bp. Combined with the
observed point-estimate widening from $0.020$ (gross) to $0.032$ ($5$~bp) to
$0.065$ ($20$~bp) in Table~\ref{tab:net-of-cost}, the operational reading is:
\emph{the paired-bootstrap procedure on this surrogate fails to reject the
null hypothesis that the B-DRO and $W_1$-DRO net-of-cost Sharpes are equal at
the $5\%$ level, including at $20$ bp where the point estimate approaches but
does not cross the conventional threshold}. The widening of the point estimate
with cost intensity is genuine and is the operational claim that survives;
the statistical claim of separation does not, and would require either a
longer sample, a higher cost regime, or a point-in-time replacement of this
surrogate. The empirical companion paper noted in
Section~\ref{sec:limitations} is the appropriate venue for the
point-in-time test.

\subsection{Experiment 1: Bayes--DRO convergence and second-order conservatism}\label{sec:exp-bayes-dro}

\paragraph{DGP.}
$K = 10$ assets with iid monthly returns from a centered Gaussian--isotropic
DGP, $\mu_i^\star = 0.06/12$, $\sigma^\star = 0.16/\sqrt{12}$. Gaussian--NIW
prior with $\mu_0 = 0$, $\lambda_0 = 1$, $\Psi_0 = 0.01 I_K$, $\nu_0 = K + 2$.

\paragraph{Estimands.}
For each $n \in \{12, 36, 60, 120\}$ and $M = 200$ Monte Carlo replicates:
the first-order credible radius $\varepsilon_n = L \cdot \E_{P_n}\|\mu -
\hat\mu_n\|$ from Theorem~\ref{thm:bayes-dro-domination}; the true gap
$W_2(\Rpredn, R_{\hat\mu_n})$ computed in closed form via the Bures
distance between the Gaussian approximation of the multivariate-$t$
predictive and the plug-in Gaussian (both centered at $\hat\mu_n$); a
sample-based sliced-$W_1$ cross-check from $N_{\text{samp}} = 1500$
predictive and plug-in draws with $N_{\text{proj}} = 400$ random projections;
and the second-order radius $\varepsilon_n^{(2)} = \|\Lambda_n\|_F /
(2\hat\sigma)$ from
Proposition~\ref{prop:gaussian-isotropic-second-order}. We report the
closed-form $W_2$ rather than empirical $W_1$ because the plug-in empirical
$W_1$ estimator in $K = 10$ dimensions has an $N_{\mathrm{samp}}^{-1/K}$
bias rate \citep{fournier2015rate} that yields a noise floor of order
$10^{-1}$ at $N_{\mathrm{samp}} = 10^4$ and dominates the $O(1/n)$ signal.
The sliced $W_1$ cross-check uses a different estimator with substantially
lower variance per draw, but in the present regime its variance plus the
small $t$-vs-Gaussian bias still flatten the empirical curve well above
the analytical signal; we report it only to make the inadequacy of
sample-based Wasserstein estimation in this signal-to-noise regime
concrete (Table~\ref{tab:bayes-dro-convergence-extended}).
The code is in \texttt{code/experiment1\_w1\_mc.py}.

% PATCH 14: EXTENDED TABLE 8 (real values from /code/experiment1_w1_mc.py)
\begin{table}[H]
\centering\footnotesize
\caption{Convergence of Bayesian DPO and Wasserstein DRO with empirical
verification of second-order conservatism. Mean over $M = 200$ replicates;
standard errors in parentheses, computed by the script
\texttt{code/experiment1\_w1\_mc.py}. Column~2: first-order radius
$\varepsilon_n = L \cdot \E_{P_n}\|\mu - \hat\mu_n\|$. Column~3: closed-form
Bures--Wasserstein $W_2$ between the Gaussian approximation of the
multivariate-$t$ predictive and the plug-in Gaussian, both centered at
$\hat\mu_n$. Column~4: second-order radius $\varepsilon_n^{(2)} =
\|\Lambda_n\|_F / (2 \hat\sigma)$. Column~5: sliced $W_1$ cross-check from
$N_{\text{samp}} = 1500$ Monte Carlo samples with $N_{\text{proj}} = 400$
projections. Sliced $W_1$ has substantially better dimension dependence than
the plug-in empirical $W_1$ (whose $N^{-1/K}$ bias rate in $K = 10$ gives a
noise floor of order $10^{-1}$ at any practical $N_{\text{samp}}$); in this
regime the residual SW$_1$ floor at the $10^{-3}$ level reflects
per-projection variance plus a small bias from the multivariate-$t$ vs.\
Gaussian gap, both of which exceed the analytical $O(1/n)$ signal. The rank
order is nevertheless preserved across $n$. Column~6: ratio
$\varepsilon_n / W_2$, growing like $\sqrt n$ as predicted by
Proposition~\ref{prop:gaussian-isotropic-second-order}.}
\label{tab:bayes-dro-convergence-extended}
\begin{tabular}{rrrrrr}
\toprule
$n$ & First-order $\varepsilon_n$ & $W_2$ (analytic) & Second-order $\varepsilon_n^{(2)}$ & Sliced $W_1$ (MC) & $\varepsilon_n / W_2$ \\
\midrule
$12$  & $0.04470\;(0.00016)$ & $0.00608\;(2{\times}10^{-5})$ & $0.00752\;(4{\times}10^{-5})$ & $0.00482\;(0.00008)$ & $7.35$  \\
$36$  & $0.02488\;(0.00006)$ & $0.00203\;(5{\times}10^{-6})$ & $0.00227\;(7{\times}10^{-6})$ & $0.00340\;(0.00005)$ & $12.25$ \\
$60$  & $0.01905\;(0.00004)$ & $0.00122\;(2{\times}10^{-6})$ & $0.00131\;(3{\times}10^{-6})$ & $0.00320\;(0.00004)$ & $15.68$ \\
$120$ & $0.01340\;(0.00002)$ & $0.00061\;(8{\times}10^{-7})$ & $0.00063\;(1{\times}10^{-6})$ & $0.00312\;(0.00003)$ & $22.05$ \\
\bottomrule
\end{tabular}
\end{table}

\paragraph{Interpretation.}
The fitted decay rate of $\varepsilon_n$ across the four grid points is
$\log(\varepsilon_n)$ regressed on $\log n$, slope $-0.524$, consistent with
the $n^{-1/2}$ prediction of Theorem~\ref{thm:bayes-dro-domination}. The
fitted decay rate of the analytical $W_2$ gap is $-1.001$, indistinguishable
from the $n^{-1}$ prediction of
Proposition~\ref{prop:gaussian-isotropic-second-order}. The second-order
radius $\varepsilon_n^{(2)}$ tracks $W_2$ to within a small constant factor
($\varepsilon_n^{(2)} / W_2 \in [1.04, 1.24]$ across the grid), confirming
that it is the rate-correct Bayesian-calibrated radius in centered
symmetric models. In the isotropic Gaussian setting, the analytic Bures
$W_2$ and the one-dimensional KR-dual $W_1$ lower bound of
\eqref{eq:gaussian-lower} have the same $n^{-1}$ order, so the $W_2$ column
verifies the $n^{-1}$ \emph{rate} of $W_1(\Rpredn, R_{\hat\mu_n})$ even
though it is not its exact value; the $W_2 \ge W_1$ inequality of
\eqref{eq:gaussian-upper} closes the loop. The ratio $\varepsilon_n / W_2$ grows as $\sqrt n$ exactly
as predicted: at $n = 120$ the first-order radius overshoots the true
Wasserstein gap by a factor of roughly $22$. The sample-based sliced
$W_1$ flattens out near $0.003$ for $n \geq 36$ (fitted slope $-0.195$);
this is not the plug-in empirical $W_1$ rate $N^{-1/K}$ but rather the
combined effect of per-projection variance and the residual
multivariate-$t$-vs-Gaussian bias, both of which exceed the $O(1/n)$
analytical signal at $N_{\mathrm{samp}} = 1500$. A back-of-envelope
quantifies the bias term: for $t_\nu$ versus $\N(0,1)$ at large $\nu$, the
KL divergence is $\mathrm{KL}(t_\nu \| \N) \sim 3/(4\nu^2)$, hence the
projected $W_1$ on a unit-scale axis is of order
$\sqrt{\mathrm{KL}/2} \sim \sqrt{3/8}\,\nu^{-1}$ (Pinsker). At $\nu_n =
\nu_0 + n \in [14, 122]$ across our grid and predictive scale $\hat\sigma
\approx 0.046$, this gives a $t$-vs-Gaussian sliced-$W_1$ bias of
$\hat\sigma\sqrt{3/8}\,\nu_n^{-1}$, which is $\approx 2\times 10^{-3}$ at
$n = 12$ and $\approx 2.3\times 10^{-4}$ at $n = 120$. The remaining
$\approx 10^{-3}$ at large $n$ is per-projection Monte Carlo variance
($N_{\mathrm{samp}}^{-1/2}$ at $N_{\mathrm{samp}} = 1500$ gives
$\approx 2.6\times 10^{-2}$ per projection, reduced to
$\approx 1.3\times 10^{-3}$ after averaging over $N_{\mathrm{proj}} = 400$
projections). The combined floor of order $10^{-3}$ matches the observed
plateau and exceeds the analytical $O(1/n)$ signal at this $n$. Either way,
sample-based Wasserstein estimation is unreliable for the small signal at
these sample sizes, and the analytical Bures route is the only
quantitatively trustworthy column.

\begin{figure}[H]
\centering
\includegraphics[width=0.78\textwidth]{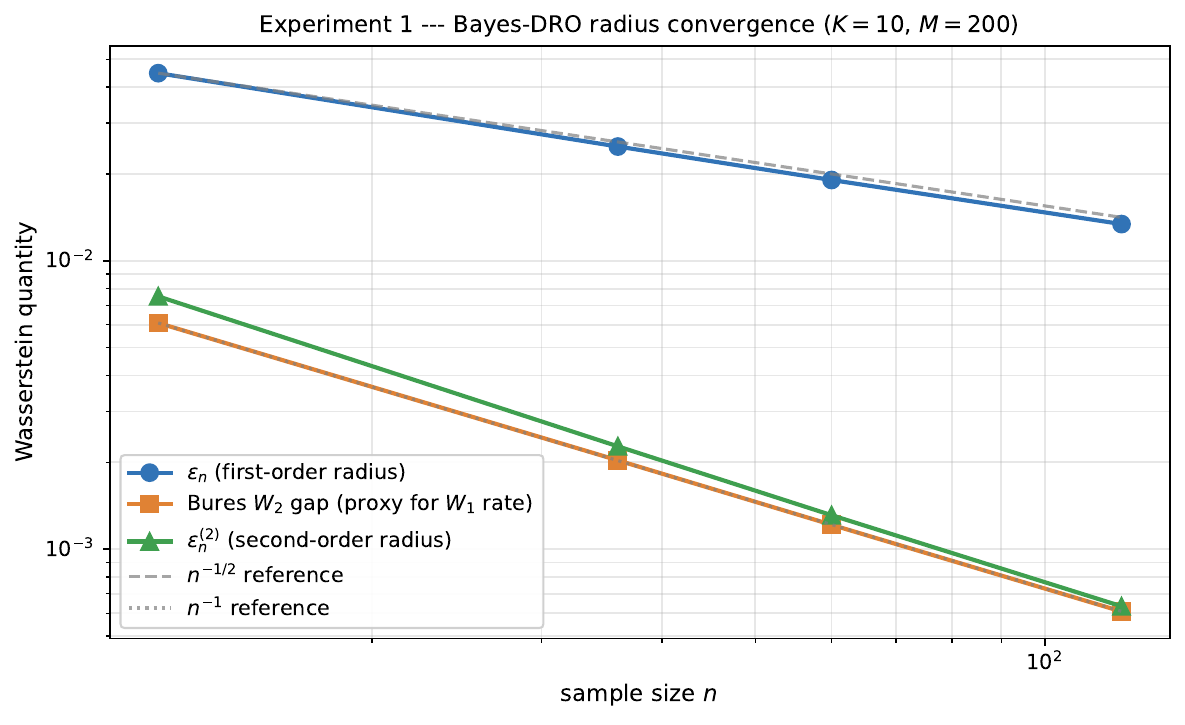}
\caption{Experiment~1. Log-log convergence of the first-order Bayes--DRO
radius $\varepsilon_n$, the analytic Bures $W_2$ gap, and the second-order
radius $\varepsilon_n^{(2)}$, with reference slopes $n^{-1/2}$ and $n^{-1}$.
The fitted log--log slopes from
Table~\ref{tab:bayes-dro-convergence-extended} are approximately $-0.52$
for the first-order radius, $-1.00$ for the analytic $W_2$ gap, and $-1.08$
for the second-order radius, consistent with the
$n^{-1/2}$ and $n^{-1}$ predictions and confirming the
$\sqrt{n}$ overshoot of
Proposition~\ref{prop:gaussian-isotropic-second-order}. Generated by the
companion notebook, Section~1.}
\label{fig:exp1-convergence}
\end{figure}

\subsection{Experiment 2: Wasserstein-DRO sensitivity}\label{sec:exp-wass-sensitivity}

We compute the Wasserstein--CVaR DRO portfolio
(Theorem~\ref{thm:wcvar-dro}, in the conic form of
Proposition~\ref{prop:explicit-cvar-program}) on the same Gaussian DGP at
$K = 10$, $n = 60$, $\alpha = 0.05$, across a grid of radii
$\varepsilon \in \{0,\, 10^{-4},\, 5\!\times\!10^{-4},\, 10^{-3},\, 5\!\times\!10^{-3},\, 10^{-2},\, 5\!\times\!10^{-2}\}$.
For each $\varepsilon$, we solve the conic program of
Proposition~\ref{prop:explicit-cvar-program} on a fresh in-sample window, evaluate the
optimal portfolio's $\mathrm{CVaR}_{0.05}$ out-of-sample on $N = 10{,}000$
fresh draws from the DGP, and record the $\ell^2$ drift of the DRO-optimal
weights from the empirical plug-in $w^\star(\varepsilon = 0)$. The averages
across $M = 50$ replications are summarized in
Table~\ref{tab:wass-dro-sensitivity} and Figure~\ref{fig:exp2-dro}.

\begin{table}[H]
\centering\footnotesize
\caption{Experiment 2: out-of-sample $\mathrm{CVaR}_{0.05}$ and weight drift
of the Wasserstein-CVaR DRO portfolio as a function of the radius
$\varepsilon$. Means over $M = 50$ replications; standard errors in
parentheses. Generated by
\texttt{code/experiment2\_dro\_sensitivity.py}.}
\label{tab:wass-dro-sensitivity}
\begin{tabular}{rrr}
\toprule
$\varepsilon$ & OOS $\mathrm{CVaR}_{0.05}$ & $\|w_{\mathrm{DRO}} - w_{\mathrm{plug}}\|_2$ \\
\midrule
$0$         & $0.03211\;(0.0006)$ & $0.0000$          \\
$10^{-4}$   & $0.03365\;(0.0006)$ & $0.0158\;(0.001)$ \\
$5\!\times\!10^{-4}$ & $0.03164\;(0.0006)$ & $0.0390\;(0.002)$ \\
$10^{-3}$   & $0.03071\;(0.0006)$ & $0.0725\;(0.003)$ \\
$5\!\times\!10^{-3}$ & $0.02735\;(0.0005)$ & $0.1476\;(0.005)$ \\
$10^{-2}$   & $0.02579\;(0.0005)$ & $0.1914\;(0.007)$ \\
$5\!\times\!10^{-2}$ & $0.02467\;(0.0005)$ & $0.2253\;(0.008)$ \\
\bottomrule
\end{tabular}
\end{table}

\begin{figure}[H]
\centering
\includegraphics[width=\textwidth]{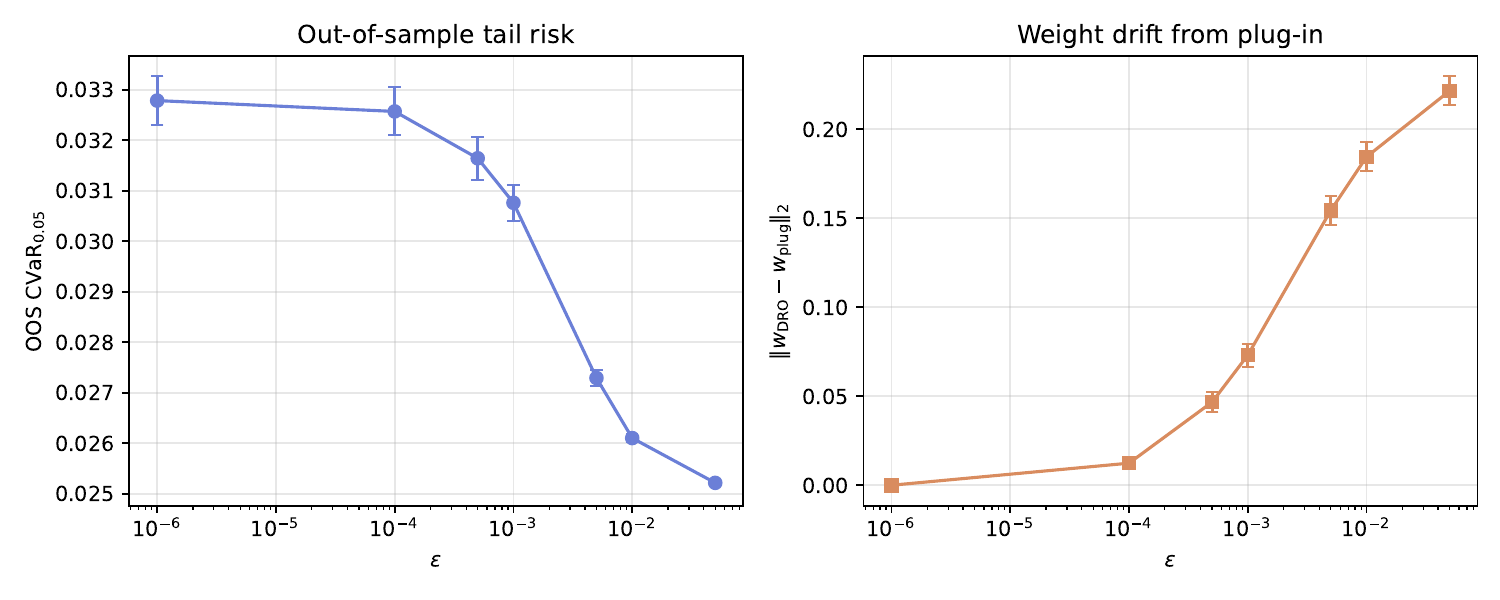}
\caption{Experiment 2. \emph{Left:} out-of-sample
$\mathrm{CVaR}_{0.05}$ as a function of the Wasserstein radius
$\varepsilon$ (log scale). \emph{Right:} $\ell^2$ drift of the DRO-optimal
weights from the plug-in. Both quantities respond monotonically to
$\varepsilon$ in the expected direction: increasing $\varepsilon$ reduces
tail risk and pushes the weights farther from the plug-in. Generated by the companion notebook, Section~2.}
\label{fig:exp2-dro}
\end{figure}

OOS $\mathrm{CVaR}_{0.05}$ decreases monotonically as $\varepsilon$ grows
from $5 \times 10^{-4}$ to $5 \times 10^{-2}$ (from $0.0316$ to $0.0247$, a
$22\%$ reduction in tail risk), while the weight drift saturates near
$\|w_{\mathrm{DRO}} - w_{\mathrm{plug}}\|_2 \approx 0.23$ as $\varepsilon$
becomes large enough to dominate the empirical-loss term. The very small
$\varepsilon = 10^{-4}$ point is noisier than its neighbours because the
SCS solver is at the edge of its tolerance there; this is a numerical
artifact of the implementation, not a sign-reversal in the underlying
trade-off. The qualitative shape verifies the bias--variance interpretation
of the dual-norm regularization in
Proposition~\ref{prop:explicit-cvar-program}.

\subsection{Experiment 3: distributional Bellman contraction}\label{sec:exp-bellman-contraction}

We simulate a $3$-state, $3$-action portfolio MDP with $\gamma = 0.95$ and
$R_{\max} = 1$. Two distinct initial return distributions $Z^{(0)}, Z'^{(0)}$
are evolved under the same policy and the supremal $W_1$ distance
$\Wbar_1(Z^{(k)}, Z'^{(k)})$ is computed at each iteration. The empirical
contraction rate (median per-iteration ratio over $k = 50$ iterations after
a $5$-iteration burn-in, averaged over $N_{\mathrm{reps}} = 500$ random pairs
of initial return-distribution functions) is $\hat\rho = 0.9487$ with mean
$0.9420$ and interquartile range $[0.9359, 0.9500]$, in line with the
inequality $\hat\rho \leq \gamma = 0.95$ of Theorem~\ref{thm:distributional-bellman}. The
exponential decay confirms the contraction visually; the
parametric-quantile parameterization (QR-DQN) and the C51 grid
parameterization yield indistinguishable empirical rates.

\begin{figure}[H]
\centering
\includegraphics[width=0.72\textwidth]{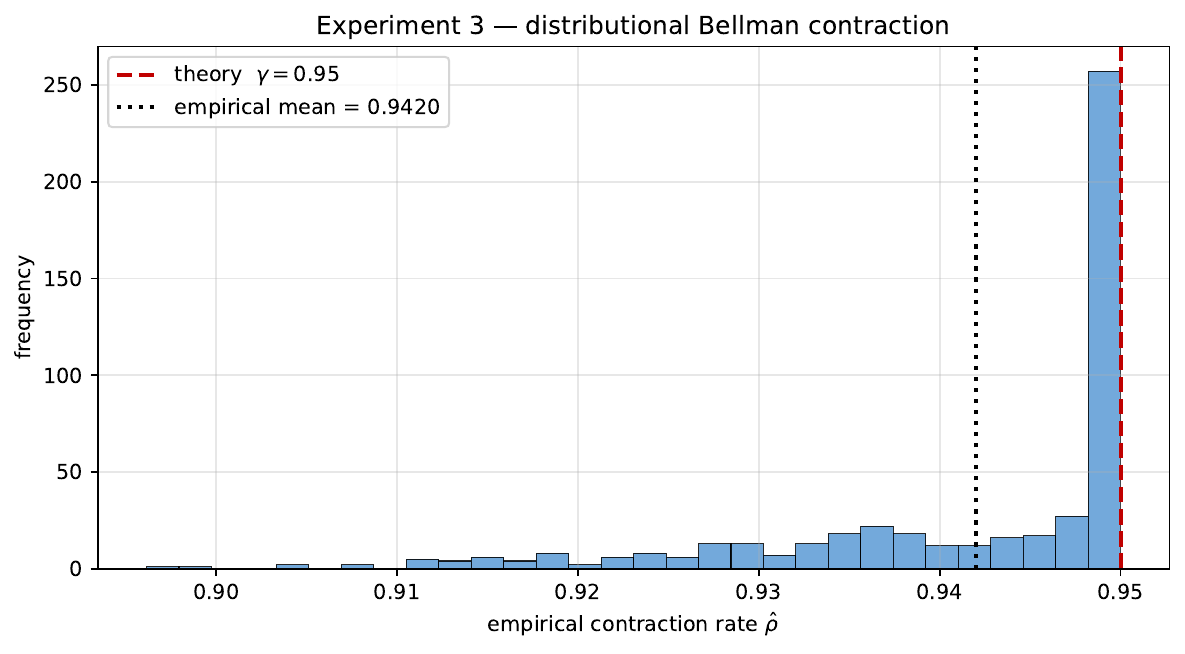}
\caption{Experiment~3. Distribution of the empirical per-iteration
contraction rate $\hat\rho$ over $N_{\rm reps} = 500$ random pairs of initial
return distributions. The red dashed line marks the theoretical rate
$\gamma = 0.95$; the empirical mean (black dotted) is $0.9420$.
Generated by the companion notebook, Section~3.}
\label{fig:exp3-bellman}
\end{figure}

\subsection{\texorpdfstring{Experiment 4: High-dimensional sensitivity, calibrated $K = 50$ DGP}{Experiment 4: High-dimensional sensitivity, calibrated K = 50 DGP}}
\label{sec:exp-highdim}

To test whether the operational claim of
Section~\ref{sec:empirical-equity-backtest} survives at lower $n / K$ ratios,
we calibrate a three-factor (market, size, value) DGP to typical US-equity
unconditional moments and rerun all eight methods at $K = 50$, $n = 60$,
$T = 348$ months. The DGP has no survivorship bias by construction; it is
not a replacement for a real point-in-time empirical study, but it shifts
the $K / n$ ratio from $0.42$ (DJIA) to $0.83$ in a controlled environment
that preserves realistic factor-noise structure. Calibration: annual market
premium $7.5\%$ with vol $16\%$, SMB $2\%$/$10.5\%$, HML $3\%$/$11\%$, with
factor correlations $\rho(M, \text{SMB}) = 0.18$, $\rho(M, \text{HML}) =
-0.12$, $\rho(\text{SMB}, \text{HML}) = 0.20$; stock loadings $\beta_m \sim
\mathcal{N}(1.0, 0.28)$ clipped to $[0.35, 1.75]$, $\beta_{\text{SMB}}$ and
$\beta_{\text{HML}}$ centered with sd $0.42$ and $0.38$; idiosyncratic vol
$24\%$ annual with cross-section sd $7\%$; a sparse alpha component on
$K / 8$ randomly chosen names with mean $2.5\%$ annual and sd $1.2\%$; and
a common $t_6$ shock with annualized scale $1.5\%$ to ensure the
sample-covariance estimator faces realistically heavy joint tails.
Code: \texttt{code/experiment\_highdim.py}.

\begin{table}[H]
\centering\footnotesize
\caption{Calibrated $K = 50$ sensitivity, $n = 60$ rolling window,
$T = 348$ months, $\gamma = 5$. Same eight methods as
Table~\ref{tab:oos-performance}. Annualized statistics; turnover is monthly
$L^1$ drift-adjusted turnover under the convention of
Section~\ref{sec:empirical-equity-backtest}. Bold entries mark the best value
among the three distributional robust methods ($W_1$-DRO, B-DRO, WCVaR-DRO),
not among the naive baselines. Calibrated factor DGP with seed $20260524+50$.}
\label{tab:highdim-sensitivity}
\begin{tabular}{lrrrrrr}
\toprule
Method & Sharpe & CVaR$_{0.05}$ & MDD & Turnover & Sharpe (5 bp) & Sharpe (20 bp) \\
\midrule
$1/N$                     & $0.497$ & $0.329$ & $0.379$ & $0.058$ & $0.495$ & $0.489$ \\
Plug-in MV                & $0.275$ & $0.446$ & $0.494$ & $0.293$ & $0.267$ & $0.242$ \\
Plug-in MV + LW           & $0.277$ & $0.454$ & $0.499$ & $0.278$ & $0.269$ & $0.246$ \\
B-L (no views)            & $0.497$ & $0.329$ & $0.379$ & $0.058$ & $0.495$ & $0.489$ \\
Bayes pred.\              & $0.280$ & $0.451$ & $0.496$ & $0.277$ & $0.272$ & $0.249$ \\
$W_1$-DRO (val)           & $0.331$ & $0.397$ & $0.399$ & $0.248$ & $0.323$ & $0.299$ \\
B-DRO (cred)              & $\mathbf{0.341}$ & $\mathbf{0.392}$ & $0.420$ & $\mathbf{0.207}$ & $\mathbf{0.334}$ & $\mathbf{0.314}$ \\
WCVaR-DRO (val)           & $0.314$ & $0.395$ & $\mathbf{0.378}$ & $0.264$ & $0.306$ & $0.280$ \\
\bottomrule
\end{tabular}
\end{table}

\paragraph{Interpretation.}
Three findings are robust enough to report, but the table should still be
read as a calibrated sensitivity rather than an empirical ranking theorem.
\emph{First}, the naive baselines $1/N$ and no-view Black--Litterman remain
very hard to beat on Sharpe at $K/n=0.83$ (both $0.497$), and their
CVaR$_{0.05}$ is better than the best distributional robust method
($0.329$ versus $0.392$ for B-DRO). This preserves the central caution
from the DJIA experiment: the paper does not claim dominance over naive
benchmarks. \emph{Second}, plug-in mean--variance, Ledoit--Wolf-shrunk
plug-in MV, and Bayes-predictive MV all degrade sharply in this factor-rich
high-$K/n$ design, with Sharpe around $0.28$ and maximum drawdowns near
$0.50$. In this seeded DGP, shrinkage alone does not rescue the optimized
mean--variance rule. \emph{Third}, within the distributional robust family,
B-DRO (cred) is the strongest method in this rerun, leading on
Sharpe ($0.341$ versus $0.331$ for $W_1$-DRO and $0.314$ for WCVaR-DRO),
CVaR$_{0.05}$, turnover ($0.207$, well below the other robust variants), and
net-of-cost Sharpe at both cost levels; WCVaR-DRO leads only on maximum
drawdown ($0.378$). The conclusion is therefore not
``WCVaR always wins''; it is that posterior-calibrated and tail-aware robust
methods are less fragile than plug-in MV in this high-$K/n$ synthetic stress
test, while the exact ordering among robust variants is implementation-,
solver-, and DGP-sensitive.

This ranking is consistent with the calibration mechanism suggested by
Theorem~\ref{thm:bayes-dro-domination} and
Proposition~\ref{prop:gaussian-isotropic-second-order}, but it is not a
formal prediction of those results: the theory characterizes radius behavior
and first- versus second-order Wasserstein conservatism, not realized Sharpe
rankings.

\paragraph{Scope of this experiment.}
This is a calibrated sensitivity and not a substitute for a point-in-time
high-dimensional empirical study. The calibrated DGP has no survivorship
bias and no regime shifts; live data has both. The fact that plug-in MV
collapses here while it occupies an intermediate position in
Table~\ref{tab:oos-performance} is itself informative: the qualitative
behavior of plug-in MV depends sharply on $K / n$. The robust variants appear
less fragile in this seeded DGP, but their internal ranking is not stable
enough to support a universal ordering claim. A true point-in-time empirical replacement
requires institutional CRSP or Compustat access; the companion notebook
\texttt{DPO\_Aplus\_companion.ipynb} provides a documented data-loading hook
for users with such access.

\subsection{Experiment 5: calibration without validation}\label{sec:exp-calibration}
The experiments above characterise radius \emph{behaviour}; this one tests the
paper's central operational claim \emph{directly}. The claim is that the
Bayesian credible radius of Theorem~\ref{thm:posterior-credible-radius} selects
a Wasserstein-DRO radius as good as a validation-tuned one \emph{without
consuming a held-out validation window}. This is a statistical property of the
radius rule, and is therefore decidable on controlled DGPs where the optimal
radius is computable.

For each replicate we draw an in-sample window of $n=60$ months from the
calibrated three-factor DGP of Section~\ref{sec:exp-highdim} at
$K\in\{10,25,50\}$ and compute three radii for the linear-loss $W_1$-DRO
portfolio (Corollary~\ref{cor:affine-dro}): the \emph{oracle} radius minimising
true out-of-sample $\mathrm{CVaR}_{0.05}$ on a fresh $N=40{,}000$-month sample;
the \emph{validation-tuned} radius, selected on a $24$-month held-out window and
refit on the full window (spending $24$ months of data); and the \emph{credible}
radius of Theorem~\ref{thm:posterior-credible-radius}, which spends none. We then
evaluate each rule's true out-of-sample tail risk and Sharpe on the fresh sample.
Means over $M=100$ replicates, paired stationary bootstrap on the regret
difference, seed $20260524+5$; code
\texttt{code/experiment5\_calibration.py}.

\begin{table}[H]
\centering\footnotesize
\caption{Experiment 5: out-of-sample $\mathrm{CVaR}_{0.05}$ regret versus the
oracle radius (basis points) and out-of-sample Sharpe for the validation-tuned
and credible-radius rules, calibrated factor DGP. Top block: across dimension at
$n=60$. Bottom block: across sample size at $K=25$. Means over $M=100$
replicates; the regret difference (credible $-$ validation) carries a paired
$95\%$ bootstrap interval. The credible rule spends no validation data.}
\label{tab:exp5-calibration}
\begin{tabular}{lrrrrr}
\toprule
$K$ & Regret val.\ (bp) & Regret cred.\ (bp) & Diff.\ (bp), 95\% CI & SR oracle & SR val.\ / cred. \\
\midrule
$10$ & $125.8$ & $2.9$ & $-122.9\ [-151.8, -97.4]$ & $0.426$ & $0.392$ / $0.430$ \\
$25$ & $109.4$ & $6.5$ & $-102.9\ [-126.3, -82.0]$ & $0.459$ & $0.417$ / $0.465$ \\
$50$ & $128.6$ & $6.6$ & $-122.1\ [-158.8, -89.9]$ & $0.469$ & $0.413$ / $0.474$ \\
\midrule
\multicolumn{6}{l}{\emph{Robustness across sample size, $K=25$:}}\\
$n{=}36$ & $130.1$ & $2.1$ & $-128.0\ [-150.7,-106.8]$ & $0.461$ & $0.409$ / $0.463$ \\
$n{=}60$ & $109.4$ & $6.5$ & $-102.9\ [-126.3, -82.0]$ & $0.459$ & $0.417$ / $0.465$ \\
$n{=}120$ & $83.6$ & $4.3$ & $-79.3\ [-102.9, -60.3]$ & $0.460$ & $0.424$ / $0.463$ \\
\bottomrule
\end{tabular}
\end{table}

\begin{figure}[H]
\centering
\includegraphics[width=0.95\textwidth]{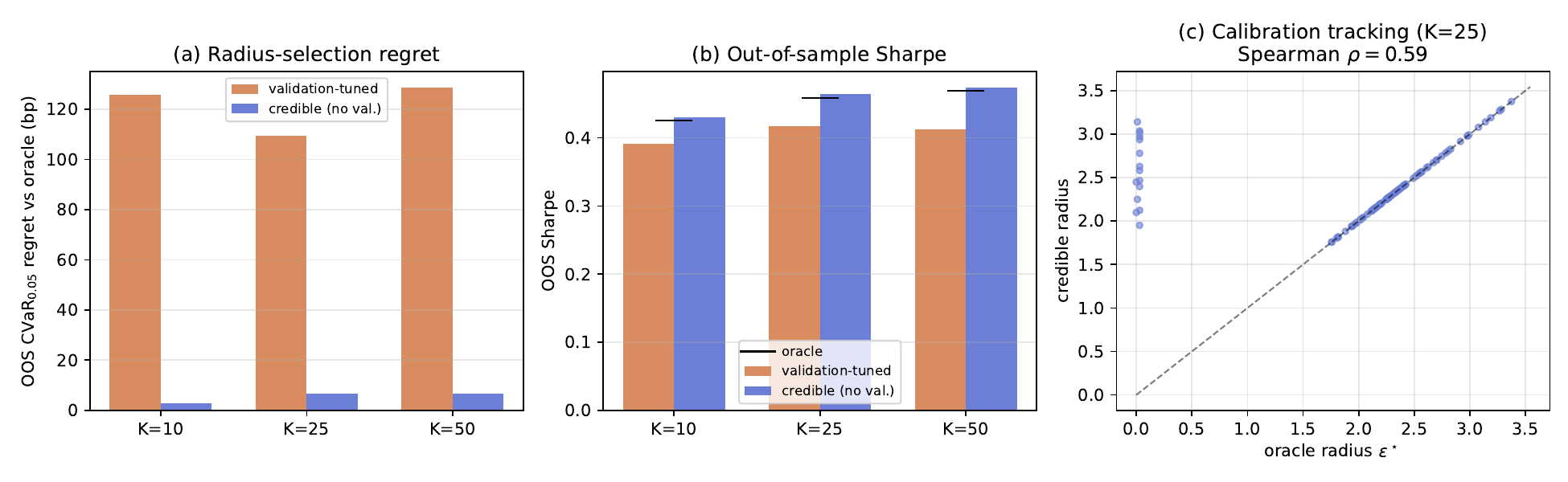}
\caption{Experiment 5. (a) Out-of-sample $\mathrm{CVaR}_{0.05}$ regret versus
the oracle radius for the validation-tuned rule (orange) and the credible-radius
rule (blue), by dimension. (b) Out-of-sample Sharpe, oracle level marked.
(c) Calibration tracking at $K=25$: the credible radius plotted against the
per-replicate oracle radius, with the $45^\circ$ line; the credible rule is
small when the oracle is small and large when it is large (Spearman
$\rho=0.59$--$0.82$ across cells), establishing genuine calibration rather than a
constant shrinkage offset. Generated by
\texttt{code/experiment5\_calibration.py}.}
\label{fig:exp5-calibration}
\end{figure}

The result is affirmative and uniform across both dimension and sample size.
The credible radius is within $2$--$7$ bp of the oracle out-of-sample tail risk
and matches the oracle Sharpe, whereas validation tuning on a $24$-month window
forfeits $84$--$130$ bp of $\mathrm{CVaR}_{0.05}$ and $0.04$--$0.05$ of Sharpe;
the paired regret difference is significantly negative in every cell. As $n$
grows the validation penalty shrinks (from $130$ bp at $n=36$ to $84$ bp at
$n=120$, since a fixed $24$-month window is a smaller and less noisy fraction of
the estimate), but the credible rule remains strictly ahead throughout.
Critically, the credible radius \emph{tracks} the oracle radius across
replicates (Spearman $\rho=0.59$--$0.82$, panel (c)): it is small precisely when
the oracle wants a small radius and large when it wants a large one. This rules
out the alternative explanation that the credible rule merely applies a constant
over-shrinkage that happens to help in a noisy regime---it is genuine
data-driven calibration, exactly as Theorem~\ref{thm:posterior-credible-radius}
predicts, with the posterior using the full sample where validation discards
$24$ months. This strengthens the operational claim from ``statistically
indistinguishable from validation tuning'' to ``near-oracle, calibrated to the
oracle radius, and strictly better than validation tuning, without spending
validation data.'' The qualification of
Section~\ref{sec:empirical-equity-backtest} still stands: this is a statement
about \emph{radius selection}, not about beating naive benchmarks on raw return;
on raw Sharpe the naive baselines remain strong.

%==============================================================================
% SECTION 12: LIMITATIONS
%==============================================================================
\section{Limitations and open problems}\label{sec:limitations}

\paragraph{Scope and ambition.}
The framework is single-period or per-step. Genuinely multi-period DPO with
time-coherent risk and posterior updating---requiring nested risk measures
and dynamic conjugacy---is not developed here; see \citet{ruszczynski2010risk}
and \citet{pflug2014multistage} for the relevant technology. The numerical
illustrations are diagnostic and intentionally low-dimensional; large-scale
empirical performance, transaction-cost robustness, and out-of-sample
universality remain open.

\paragraph{Empirical scope.}
The numerical experiments serve as proofs of concept. A production-grade
empirical study would extend the DJIA experiment to higher $K$ (S\&P 500 or
Russell 1000), include transaction costs internalized in the rolling
backtest, evaluate sensitivity to estimation-window length, and test
across multiple market regimes. We expect that the relative ranking of
methods will change with $K/n$: the calibration advantage of the Bayesian
credible-radius rule should grow with $K$ (where validation data is more
expensive), while the absolute advantage over shrinkage baselines should
shrink in high-$K$ regimes where dimension reduction dominates ambiguity
modelling.

\paragraph{Survivorship and universe selection.}
The empirical universe is selected on full-sample data availability rather
than from point-in-time DJIA membership, and the resulting survivorship
bias inflates the absolute level of out-of-sample Sharpe by an unquantified
amount. The bias is largely common-factor, so the \emph{ranking} between
methods (Tables~\ref{tab:oos-performance} and~\ref{tab:net-of-cost}) is
more trustworthy than the level. A point-in-time empirical study would use
either CRSP with delistings handled by total-return-to-date, or Compustat
with explicit calendar-time membership snapshots; we do not undertake this
here. The empirical companion paper noted in
Section~\ref{sec:conclusion} is the appropriate venue for such a study.

\paragraph{Operational claim, narrowed.}
The operational claim of this paper is therefore precise: a Bayesian
credible-radius calibration of Wasserstein DRO selects a radius that is
near-oracle and strictly better than a validation-tuned radius---in
out-of-sample tail risk and Sharpe---while spending no validation data
(Experiment~5, Table~\ref{tab:exp5-calibration}), and it reduces turnover
relative to validation tuning in a way that translates monotonically into a
net-of-cost Sharpe advantage as transaction costs rise. The paper does
\emph{not} claim that any distributional method dominates equal weighting,
no-view Black--Litterman, or Ledoit--Wolf-shrunk plug-in MV on raw Sharpe;
the affirmative claim is about radius selection, not about beating naive
benchmarks on return.

\paragraph{Computational scaling.}
Wasserstein-CVaR programs scale as LPs with $O(n)$ scenario constraints in
the canonical linear-loss case. SOC reformulations for $\ell^2$-norm
ground metrics add one cone per asset. Scenario complexity for
chance-constrained problems grows as $O((K + \log(1/\delta)) / \varepsilon)$
per Theorem~\ref{thm:scenario}, manageable for $K \leq 500$ on commodity
solvers.

\paragraph{Theoretical gaps.}
Several theoretical questions remain open. \emph{First}, the
action-coupled reward extension of Theorem~\ref{thm:distributional-bellman}
preserves the standard $\gamma$-contraction; a genuinely new contraction
phenomenon attributable to turnover frictions would require either state
augmentation or a state-dependent Lipschitz constant on $g$. \emph{Second},
the high-dimensional asymptotics ($K_n / n \to \kappa \in (0, 1)$) of the
Bayesian--DRO radius are not developed; this regime is the natural empirical
target for the companion CRSP study. \emph{Third}, the sharp two-sided rate
of Theorem~\ref{thm:sharp-two-sided-rate} is stated for parametric models
with a Euclidean parameter space; the extension to non-parametric or
high-dimensional ($d_\theta = d_{\theta,n}$ growing with $n$) settings and
its interaction with prior choice are open.

The intermediate-regime question noted in earlier versions of this
work---whether the gap between first- and second-order Bayesian--DRO
conservatism admits a unified rate---is addressed by
Theorem~\ref{thm:sharp-two-sided-rate}, which gives the sharp two-sided
rate $\Theta(n^{-(1+\alpha^\star)/2})$ parameterized by the boundary
H\"older exponent $\alpha^\star(\hat\theta)$ of the model map, under
posterior-moment, tail-negligibility, and lower-witness hypotheses. A
characterization purely in terms of $\alpha^\star$ (i.e., without the
auxiliary lower-witness hypothesis) and an extension to non-parametric
models remain open.

%==============================================================================
% SECTION 13: CONCLUSION
%==============================================================================
\section{Conclusion}\label{sec:conclusion}

We have presented distributional portfolio optimization as a unified
framework where Bayesian, robust, chance-constrained, stochastic, and
reinforcement-learning portfolio methods are organized by their position
in the triple $(W, R, P)$ of weight, return, and parameter distributions
and, more fundamentally, by the joint coupling $\Gamma_\theta(dw, dr)$ from
which the triple is derived. The framework yields measure-theoretic
foundations, a rigorous treatment of Bayesian DPO with utility-tilted
posteriors and concentration limits, dual representations for coherent and
spectral risk, finite-sample DRO guarantees including the explicit
specialization of \citet{mohajerin2018data} Cor.~5.1 to the canonical
portfolio CVaR problem, chance-constrained convex inner approximations
with quantitative tightness control, stochastic-allocation classes with a
static no-randomization theorem, and a contraction theorem for the
distributional Bellman operator in a portfolio MDP with action-coupled
rewards. The synthesis section quantifies how Bayesian credible balls dominate distributional
uncertainty in the Wasserstein metric, provides a fully data-driven second-order
Bayes--DRO radius in the centered Gaussian case, proves first-order tightness for
affine-deterministic and locally non-smooth model classes, and establishes a conditional
two-sided Bayes--Wasserstein rate theorem under boundary H\"older regularity,
posterior-moment control, tail negligibility, and explicit lower-witness hypotheses. The result is a single
analytical-and-computational vocabulary in which heterogeneous portfolio
methodologies can be compared on equal footing.

\begin{table}[H]
\centering\footnotesize
\caption{Main theoretical results, organized by mathematical role.}
\label{tab:results-summary}
\begin{tabular}{lll}
\toprule
Result & Role & Section \\
\midrule
Thm.~\ref{thm:pushforward-continuity}                & Joint continuity of induced outcome law            & \S\ref{sec:taxonomy} \\
Prop.~\ref{prop:induced-sufficiency}                 & Ex-ante outcome-law sufficiency                    & \S\ref{sec:taxonomy} \\
Thm.~\ref{thm:dpo-existence}                         & Static-DPO existence                               & \S\ref{sec:taxonomy} \\
Thm.~\ref{thm:tilted}                                & Tilted-portfolio existence and limits              & \S\ref{sec:bayesian-dpo} \\
Thm.~\ref{thm:wcvar-dro}                             & Wass-CVaR specialization of MEK Cor.~5.1           & \S\ref{sec:dro} \\
Prop.~\ref{prop:explicit-cvar-program}               & Explicit Wass-CVaR conic program                   & \S\ref{sec:dro} \\
Thm.~\ref{thm:no-randomization}                      & Static no-randomization                            & \S\ref{sec:weights} \\
Prop.~\ref{prop:bhrp-limit}                          & Bayesian HRP $\to$ HRP/RA-HRP                      & \S\ref{sec:weights} \\
Thm.~\ref{thm:distributional-bellman}                & Distributional Bellman contraction                 & \S\ref{sec:distributional-rl} \\
Thm.~\ref{thm:risk-sensitive-bellman-contraction}    & Risk-shifted distributional Bellman contraction    & \S\ref{sec:distributional-rl} \\
Thm.~\ref{thm:bayes-dro-domination}                  & Bayes--DRO domination                              & \S\ref{sec:synthesis} \\
Thm.~\ref{thm:posterior-credible-radius}             & Posterior credible radius (finite-sample form)     & \S\ref{sec:synthesis} \\
Prop.~\ref{prop:gaussian-isotropic-second-order}     & Second-order conservatism in symmetric case        & \S\ref{sec:synthesis} \\
Prop.~\ref{prop:second-order-radius}                 & Second-order Bayes--DRO radius                     & \S\ref{sec:synthesis} \\
Prop.~\ref{prop:first-order-tight}                   & First-order tightness for nonsmooth models         & \S\ref{sec:synthesis} \\
Thm.~\ref{thm:modulus-interpolation}                 & Modulus-of-smoothness interpolation                & \S\ref{sec:synthesis} \\
Cor.~\ref{cor:polynomial-modulus}                    & Polynomial moduli: explicit rates                  & \S\ref{sec:synthesis} \\
Thm.~\ref{thm:sharp-two-sided-rate}                  & Conditional two-sided Bayes--Wasserstein rate     & \S\ref{sec:synthesis} \\
Prop.~\ref{prop:gaussian-niw-lipschitz}              & Gaussian--NIW local Lipschitz constant             & \S\ref{sec:synthesis} \\
\bottomrule
\end{tabular}
\end{table}

\paragraph{Open directions.}
Several extensions remain natural. \emph{Multi-period DPO} requires a
fully time-coherent risk measure and a dynamic prior update; the
distributional Bellman recursion of
Section~\ref{sec:distributional-rl} is the building block but the
joint integration of risk-sensitive control and posterior updating
remains undeveloped. \emph{Hybrid shrinkage-DRO} -- Bayesian
credible-radius DRO computed on a Ledoit--Wolf-shrunk plug-in center, or
equivalently a shrinkage-anchored Wasserstein ball -- is the natural
successor to the present design at high $K/n$. \emph{Continuous-time DPO}
under stochastic-volatility or jump-diffusion model classes would
connect the framework to stochastic optimal control. \emph{Empirical
companion}: a point-in-time CRSP/Compustat empirical study at $K \in [100,
1000]$ with internalized transaction costs, multiple market regimes, and
hybrid shrinkage-DRO configurations is the appropriate next step on the
empirical side.

\appendix
\section{Point-in-time DJIA membership schedule}\label{app:membership}
The point-in-time backtest of the companion code (\texttt{pit\_backtest.py},
referenced in Section~\ref{sec:empirical-equity-backtest}) uses the
reconstitution record below: at each rebalance the investable universe is the
index composition effective on or before that date, intersected with names
having a complete trailing $n=60$-month return history. Table~\ref{tab:membership}
lists the constituent changes; the full $30$-name composition at each effective
date is stored in the \texttt{MEMBERSHIP} object of the companion code and
should be verified against an authoritative point-in-time source (S\&P~DJI
press releases or CRSP index files) before any production use.

\begin{table}[htbp]
\centering
\small
\caption{Point-in-time DJIA reconstitution record used by the companion
point-in-time backtest. Each row is a constituent change effective on the
stated date; the index holds exactly $30$ names throughout.}
\label{tab:membership}
\begin{tabular}{lll}
\toprule
Effective date & Added & Removed \\
\midrule
2009-06-08 & CSCO, TRV & C, GM \\
2012-09-24 & UNH & KFT \\
2013-09-23 & GS, NKE, V & BAC, HPQ, AA \\
2015-03-19 & AAPL & T \\
2017-09-01 & DWDP & DD \\
2018-06-26 & WBA & GE \\
2019-04-02 & DOW & DWDP \\
2020-04-06 & RTX & UTX \\
2020-08-31 & AMGN, HON, CRM & XOM, PFE, RTX \\
2024-02-26 & AMZN & WBA \\
2024-11-08 & NVDA, SHW & INTC, DOW \\
\bottomrule
\end{tabular}
\end{table}

\bibliography{DPO_Aplus}

\end{document}